\documentclass[journal]{IEEEtran}  
\pdfoutput=1

\usepackage{graphicx} 
\usepackage{subfigure}
\usepackage{float}
\usepackage{amsmath} 
\usepackage{mathrsfs}
\usepackage{amssymb}  
\usepackage[linesnumbered]{algorithm2e}
\usepackage{flushend}
\usepackage{cite}

\usepackage[usenames,dvipsnames]{xcolor}
\let\oldmarginpar\marginpar
\renewcommand\marginpar[1]{\oldmarginpar[\raggedleft\footnotesize #1]%
{\raggedright\footnotesize #1}}
\newcommand{\skcomment}[2]{{\color{red}#1}\marginpar{\color{TealBlue}\tiny #2}}
\renewcommand{\skcomment}[2]{#1}

\newcommand{\velpath}[1]{V[#1]}
\newcommand{\wait}[1]{\delta[#1]}
\newcommand{\arrive}[1]{t_A[#1]}
\newcommand{\depart}[1]{t_D[#1]}

\usepackage{theorem}
\newtheorem{theorem}{Theorem}[section]
\newtheorem{lemma}[theorem]{Lemma}
\newtheorem{proposition}[theorem]{Proposition}

\newtheorem{definition}[theorem]{Definition}

\newtheorem{assumption}[theorem]{Assumption}
\newtheorem{conjecture}[theorem]{Conjecture}
\newcommand{\qed}{\nobreak \ifvmode \relax \else
      \ifdim\lastskip<1.5em \hskip-\lastskip
      \hskip1.5em plus0em minus0.5em \fi \nobreak
      \vrule height0.75em width0.5em depth0.25em\fi}
\renewcommand{\qed}{\hfill$\Box$}      
\newenvironment{proof}[1][Proof]{\begin{trivlist}
\item[\hskip \labelsep {\bfseries #1}]}{\qed\end{trivlist}}

\newcommand\numberthis{\addtocounter{equation}{1}\tag{\theequation}}

\newcommand{\ie}{{\em i.e.}}
\newcommand{\eg}{{\em e.g.}}

\newcommand{\reals}{\mathbb{R}}
\newcommand{\naturals}{\mathbb{N}}
\DeclareMathOperator*{\argmin}{arg\,min}

\newcommand{\EE}{\mathbb{E}}
\DeclareMathAlphabet{\mathpzc}{OT1}{pzc}{m}{it}

\newcommand{\notraj}[1]{#1}

\title{\LARGE \bf
Polling-systems-based Autonomous Vehicle Coordination\\ in Traffic Intersections with No Traffic Signals
}

\author{David Miculescu\qquad\qquad\qquad\qquad Sertac Karaman
\thanks{This work was partially funded by the National Science Foundation through Grant \#1350685. 
}
\thanks{The authors are with the Department of Aeronautics and Astronautics and the Laboratory for Information and Decision Systems, Massachusetts Institute of Technology, Cambridge, MA 02139.
        {\tt\footnotesize \{dmicul, sertac\}@mit.edu}}%
}

\begin{document}

\maketitle
\thispagestyle{empty}
\pagestyle{empty}

\begin{abstract}
The rapid development of autonomous vehicles spurred a careful investigation of the potential benefits of all-autonomous transportation networks. Most studies conclude that autonomous systems can enable drastic improvements in performance. 
A widely studied concept is all-autonomous, collision-free intersections, where vehicles arriving in a traffic intersection with no traffic light adjust their speeds to cross safely through the intersection as quickly as possible. In this paper, we propose a coordination control algorithm for this problem, assuming stochastic models for the arrival times of the vehicles. 
The proposed algorithm provides provable guarantees on safety and performance. More precisely, it is shown that no collisions occur surely, and moreover a rigorous upper bound is provided for the expected wait time. The algorithm is also demonstrated in simulations.
The proposed algorithms are inspired by polling systems. In fact, the problem studied in this paper leads to a new polling system where customers are subject to differential constraints, which may be interesting in its own right.

\end{abstract}

\section{Introduction} \label{section:introduction}
Autonomous vehicles hold the potential to revolutionize transportation and logistics. 
Self-driving cars may reduce congestion and emissions, while substantially enhancing safety~\cite{Greenblatt:2015fa,Fagnant:2015cf}. Drones may be used for delivering goods in urban centers or for providing emergency supplies in disaster-struck areas~\cite{Sandvik:2014em,DAndrea:2014kja}.  
Almost all studies on autonomous vehicles, including those mentioned above, envision large {\em fleets} of autonomous vehicles that work in coordination to enable efficient transportation and logistics services. In fact, in many cases, well-coordinated fleets are key for economic and societal impact. For instance, self-driving cars create the most value when they are deployed as a fleet of autonomous taxis for mobility on demand~\cite{Fagnant:2014bg,Fagnant:2015cf}; delivery drones are most useful when they can relay packages from one to another~\cite{DAndrea:2014kja}.

In almost all such examples, efficient coordination algorithms can enable substantial savings in energy, increase in throughput and capacity, and reduction in delays and emissions. In fact, the overall performance of the system is often highly sensitive to the operator's choice of the underlying coordination algorithms~\cite{Fagnant:2014bg}. Thus, designing effective coordination algorithms is vital towards understanding the real potential for impact that autonomous vehicles provide. 


\begin{figure}
\centerline{\includegraphics[width=0.34\textwidth]{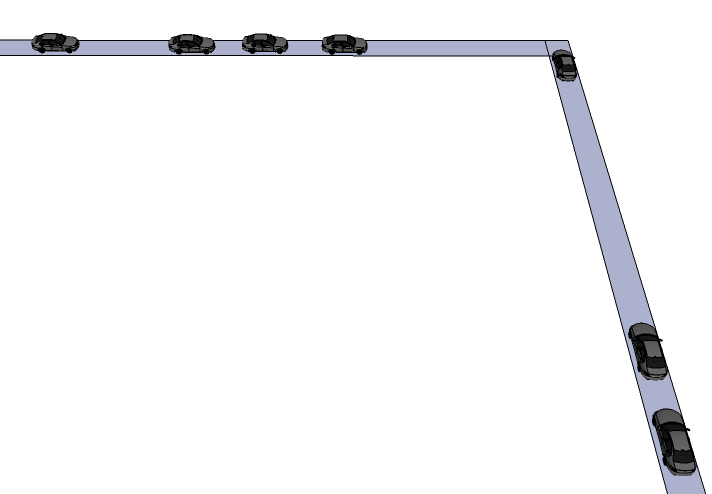}}
\caption{An illustration of a fully autonomous traffic intersection. Autonomous vehicles arriving near the intersection are fully controlled by a central control system. The control system ensures that the vehicles safely pass through the intersection; furthermore, the central control system provides provable guarantees on performance, \eg, an upper bound on the average time it takes a typical vehicle to go through this intersection region.} \label{figure:intersection_illustration}
\end{figure}

The application domain of such coordination algorithms is not limited to autonomous cars and delivery drones.
There are already a number of existing applications where effective coordination of a fleet of autonomous vehicles is key. 
For instance, robotic vehicles that shuttle materials to service packaging requests in Amazon warehouses~\cite{Guizzo:2008wj,Wurman:2008tp} and autonomous trucks that carry cargo containers in seaports~\cite{Anonymous:uclib5BZ,Hoshino:2006tf,Grunow:2004ck} can all benefit from better coordination algorithms to reduce delays and energy consumption, while increasing throughput. 

High-performance coordination algorithms for autonomous vehicles that have to work in tight spaces would have tremendous immediate impact on these application domains. They would also pave the way for the design of future urban centers with shared autonomous vehicles~\cite{Fagnant:2014bg}, next generation air transportation systems~\cite{Lyons:2012bk}, and industrial environments that house autonomous vehicles to shuttle goods~\cite{PeterRWurman:2008tp}. 

In systems involving high-performance autonomous vehicles, ``traffic intersections,'' where flows of vehicles intersect, are among the most critical. 
Efficient coordination algorithms are most beneficial in these places where interaction between vehicles is frequent. 
In this paper, we study ``signalless traffic intersections'' in which vehicles coordinate among each other to carefully adjust their speed in such a way that maximal throughput can be achieved with minimal delay. To be precise, we assume that vehicles arrive stochastically at the ends of two roads that intersect at their other end. (See Figure~\ref{figure:intersection_illustration}.) We assume that a central controller commands each vehicle in the system. The vehicles are bound by second-order differential constraints with bounded velocity and bounded acceleration. In this setting, we propose a novel coordination algorithms that provide guarantees on safety and performance. In doing so, we establish novel connections with existing results in polling systems literature, which may be of interest in its own right.

The importance of coordination algorithms for signalless traffic intersections has not gone unnoticed. In particular, there is a large and growing body of literature based on multi-agent simulations~\cite{KDresner:2009ul,TszChiuAu:2010vd,Fagnant:2014bg,Azimi:2012vd,Azimi:2013,Azimi:2014}, genetic algorithms~\cite{Onieva:2012dd}, token-based approaches~\cite{Liu:2013kc,Perronnet:2013}, auction-based approaches~\cite{Carlino:2013wn}, and discrete-time occupancy theory~\cite{Lu:2016}.
A comprehensive recent summary and comparison of some of these approaches are provided by Dai et al.~\cite{daiquality:2016}. 
These approaches propose practical algorithms that seem to perform well in simulation studies. 
However, most provide no performance guarantees.

Control-theoretic formulations and analyses have also been considered.
The existing literature includes a hierarchical-distributed coordination scheme that balances computational complexity and optimality~\cite{Tallapragada:2016}, an approximation algorithm based on an optimal control formulation~\cite{Hult:2015}, an MPC-based approach~\cite{Kim:2014}, a distributed control generated from its formal specifications~\cite{Duo:uT9p4lpt}, a game-theoretic cruise control approach~\cite{Zohdy:2013}, an analysis of intersecting vehicle flows~\cite{Lupu:2010vza, ZhiHongMao:2001cm} as well as approximation approaches based on hybrid systems theory~\cite{Colombo:2012uc,Hafner:2013to}.
Some of these approaches provide provable guarantees. However, many of them are conservative, and often fail to provide any guarantees on performance. 

Finally, problems similar to the one considered in the present paper were discussed in the context of air traffic control and conflict resolution~\cite{AlonsoAyuso:2011cr,Paielli:1997job,Pallottino:2002wca,Lalish:2008tm,Ishutkina:2012it,Lee:2009fta,Frazzoli:2001eza,Mao:2007csa,Mao:2001uh,Mao:2001ij,Mao:2005iea}. The work we report in the present paper is most related to the analysis of aircraft flows by Mao et al.~\cite{Mao:2001uh,Mao:2001ij,Mao:2007csa,Mao:2005iea}. Their analysis considers two intersecting flows of aircraft, and computes the maximum throughput that can be achieved by various maneuvers, such as heading change. The analysis is geometric and is for aircraft maneuvers; it considers a series of vehicles arriving one after another, aimed at finding the maximum throughput in the worst case. 
In contrast, our analysis is tailored for autonomous ground vehicles that move in a lane, and thus, have to slow down and speed up, but can not maneuver sideways. More importantly, our analysis takes into account stochastic arrivals, allowing performance results for the average case. The analysis in this paper may apply to aircraft models as well, extending existing literature in air traffic control for intersecting aircraft flows, with stochastic models.

The main contribution of this paper is a motion coordination algorithm for autonomous vehicles operating in traffic intersections with no traffic signals. Our algorithm is based on polling systems~\cite{Gross:1998vr}, and it provides provable guarantees on safety and performance. Our results extend many of the widely-used polling policies to traffic intersections of autonomous vehicles. This extension is non-trivial due to the differential constraints of the vehicles, which do not appear in the traditional application areas of polling systems, such as data networks. 

Our main theoretical result is the following. If the road length is larger than a certain threshold (which depends on the maximum speed and maximum acceleration of the vehicles), then no collisions occur surely, and, moreover, the delay that each vehicle experiences is bounded by the delay experienced by a customer arriving in a traditional polling system without the differential constraints. This result is remarkable. Essentially, as far as the delay is concerned, the differential constraints of the vehicles are irrelevant, as long as the road length is larger than a certain threshold; most of the existing results on delay in the polling systems literature directly apply.

The novel connections with polling systems and the new motion coordination algorithm provides us with an insight for signalless traffic intersection control: {\em smart platooning}, \ie, effective clustering of vehicles, is key for achieving good delay-throughput tradeoff. 
To achieve the best throughput, the vehicles in the same lane must slow down to form a cluster, and pass through the intersection as a platoon, in order to minimize the time it takes to switch over to the other lane. This behavior naturally comes out from the motion coordination algorithms. We emphasize this insight throughout the paper with examples, simulations, and analyses, when applicable.

Let us note at this point that the polling systems literature is fairly rich~\cite{Boxma:2011fm,Vishnevskii:2006dd,Takagi:1998vi,Boon:2011dw,Levy:1990hj}. Motivated by applications in communication systems, transportation systems, and manufacturing, the literature has flourished during the last few decades. Optimal polling policies were characterized for a large class of input processes~\cite{liu1992optimal} and analytical expressions were derived for a range of polling policies~\cite{Boxma:2011fm,Vishnevskii:2006dd,Takagi:1998vi}. These foundational results have been utilized in several application domains~\cite{Boon:2011dw,Levy:1990hj}.

The existing applications of polling systems also include urban traffic flows~\cite{Boon:2011dw}. 
However, to the best of our knowledge, the applications of polling systems in the context of all-autonomous traffic intersections is novel. Furthermore, the problem formulation presented in this paper can be generalized leading to a new class of polling systems where the customers are subject to differential constraints, which may be interesting in its own right. 
In this case, the customers must be ``steered'' to a suitable state before they can be serviced. Our results imply that, in a certain class of such polling systems, the differential constraints can be managed, \ie, the differentially-constrained polling system can achieve the same performance that its counterpart with no differential constraints achieves. 

Although this paper focuses on applications in urban transportation, let us emphasize that potential applications also include air transportation as well as warehouse automation and manufacturing. In particular, trajectory-based operations considered for the NextGen air transportation system (see, \eg,~\cite{Lyons:2012bk}) by the Federal Aviation Administration in the U.S. will enable trajectory planning and precision execution for aircraft. An effective use of the shared airspace may be possible by setting up virtual roads and intersections, where the aircraft coordinate their motion for increased performance. Furthermore, autonomous robotic vehicles servicing warehouses, factories, and transportation hubs (\eg, airports, seaports, train stations, {\em etc}.) may enable efficient transportation of goods and people, with effective motion coordination algorithms. 

A preliminary version of our results were presented in the Conference on Decision and Control (CDC) in 2014~\cite{miculescu:cdc14}. In addition to the results presented in the preliminary version of this paper, the present paper provides a full proof of safety and performance results for a slightly more general setting, provides new results on stability, and provides conjectures and open problems based on new simulation results. 

This paper is organized as follows. We formalize the problem of signalless intersection control with stochastic arrivals in Section~\ref{section:problem}. We describe our control policy in Section~\ref{section:control}, where we also introduce preliminaries, such as key results from the polling systems literature. We state our main theoretical results in Section~\ref{section:analysis}, and provide the results of supporting simulation studies in Section~\ref{section:simulations}. We conclude the paper with remarks in Section~\ref{section:conclusions}. We provide the technical proofs in the appendix.

\newpage

\section{PROBLEM DEFINITION} \label{section:problem}

We are interested in the problem of motion coordination through a traffic intersection, as depicted in Figure~\ref{figure:intersection_illustration}. 
Consider a traffic intersection where two orthogonal lanes intersect. Suppose each vehicle is subject to second order dynamics of the following form:
\begin{align}
\ddot{x}(t) &= u(t), \label{equation:system}
\end{align}
where $x(t)$ denotes the position of the front bumper of the vehicle, $0\le \dot{x}(t) \le v_\mathrm{m}$ is the maximum velocity constraint, and $\vert u(t) \vert \le a_\mathrm{m}$ is the maximum acceleration constraint.
%
%
The region where the two lanes intersect is called the {\em intersection region}. 
The portion of the road within a distance $L$ of the intersection region is called the {\em control region}.
We assume that 
the input signal $u(t)$ is directly controlled by a central control system for all vehicles that are in the control region.

\begin{figure}[b]
\centerline{
\includegraphics[clip=true, trim= 0in 0in 0in 0in, width=0.35\textwidth]{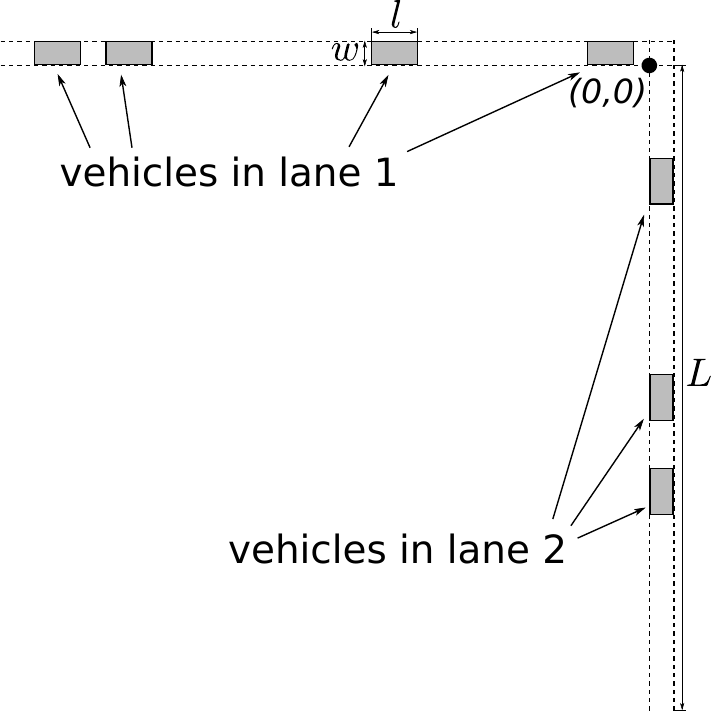}}
\vspace{-0.1in}

\caption{An illustration of the control region.}\label{figure:intersection2}
\end{figure}

This central control system does not know {\em a priori} the precise times that each vehicle will arrive at the control region. However, we assume that certain statistics of their interarrival times are available. 
More precisely, we model the arrival times $\{t_{i,k} : i \in \naturals\}$ as a suitable stochastic process, where $t_{i,k}$ is the time that the $i$th vehicle from lane $k \in \{1,2\}$ enters the control region (front bumper is exactly a distance of $L$ away from the intersection region).
In other words, the $i$th vehicle entering the control region from lane $k$ is at position $-L$ at time $t_{i,k}$, \ie, $x_{i,k}(t_{i,k}) = -L$, where $x_{i,k} (t)$ denotes the position of this vehicle at time $t$. 
From this point on, the position $x_{i,k}(t)$ of the same vehicle is governed according to the dynamics given in Equation~\eqref{equation:system}.  If the lane number of the $i$th vehicle is clear from context, we simplify the notation and drop the lane subscript, \ie, $x_i$ denotes the trajectory of the $i$th vehicle, $t_i$ denotes the time that the $i$th vehicle enters the control region, {\em et cetera}.
We assume that the vehicles enter the control region with maximum speed, \ie, $\dot{x}(t_{i,k}) = v_m$ for all $i \in \naturals$ and all $k \in \{1,2\}$.
To simplify notation, the $i$th vehicle that enters lane $k$ is at times referred to as vehicle $(i,k)$.

We represent each vehicle as a two-dimensional, rectangular rigid body with length $l$ and width $w$. The position of this rigid body is encoded with respect to one of the corners of the intersection region. See Figure~\ref{figure:intersection2}. 
The position of the rigid body represents the center of the front bumper of the vehicle. 
The orientation of this rigid body depends on which lane the vehicle is traveling in. 
To formalize, let us define the rigid body at position $y \in \reals$ in lane $k$ with $R(y,k) \subseteq \reals^2$, \ie, 
\begin{align*}
R(y,1) &:= \{(y_1,y_2) \in \reals :  y-l< y_1 < y, 0 < y_2 < w \}, \\
R(y,2) &:= \{(y_1,y_2) \in \reals :  0< y_1 < w, y-l < y_2 < y \}.
\end{align*}

With this notation, the $i$th vehicle in lane $k$ at time $t$ is represented by $R(x_{i,k}(t), k)$.
%
Let $I_k(t)$ denote the indices for all vehicles that are inside the control region at time $t$ and in lane $k$. For instance, if the 3rd, 4th, and 5th vehicles are in lane 1 at time $t$, then we have $I_1(t) = \{3,4,5\}$. Clearly, $I_k(t)$ is a set of consecutive natural numbers for all $t \ge 0$ and all $k \in \{1,2\}$.
Finally, we define safety as follows. 

\begin{definition}{\bf (Safety)}  \label{definition:safety}
The control region is said to be safe at time $t \in \reals_{\ge 0}$, if there are no pairwise collisions among the vehicles, i.e., 
$
R(x_{i,k}(t), k) \cap R(x_{j,l}(t), l) = \emptyset
$
for all $i \in I_k(t)$, all $j \in I_l(t)$ and all $k,l\in \{1,2\}$ satisfying $(i,k) \not= (j,l)$.
\end{definition}

For each vehicle, we define the {\em delay incurred in traveling through the system}, or {\em delay} for short, as follows. Recall that $t_{i,k}$ is the time that the $i$th vehicle enters lane $k$. Let $T_{i,k}$ denote the time that the same vehicle completely exits the intersection region, that is, the rear bumper of the vehicle is outside the intersection region. More precisely, $T_{i,k}$ is the time instance for which $x_{i,k}(T_{i,k}) = l + w$.
Then, the time it takes for the same vehicle to travel through the system is $T_{i,k} - t_{i,k}$. We define the delay of a vehicle as the difference between this time and $(L + l + w) v_m$, which is the time it would have taken the vehicle to travel through the system had the control region been completely empty, \ie, at full velocity $v_m$.
\begin{definition}{\bf (Delay)}  \label{definition:delay}
The delay for vehicle $(i,k)$ is defined as
$$
D_{i,k} := (T_{i,k} - t_{i,k}) - \frac{L + l + w}{v_m}.
$$
\end{definition}

In this paper, we are interested in developing coordination algorithms that govern this
\skcomment{vehicle-to-infrastructure (V2I)}{We often  put the abbreviations and shorthand in parenthesis, after the longer version of the word (not before).}
controller such that both performance (\eg, in terms of delay) and safety (avoiding collisions between vehicles) are simultaneously guaranteed.
We will discuss such algorithms in the next section.
At this point, we tacitly leave out an important part of the problem formulation: we do not specify anything about the arrival process, \ie, the distribution of the stochastic process $\{t_{i,k} : i \in \naturals\}$. We will discuss the arrival process as a part of our model and assumptions in Section~\ref{section:analysis}.

\section{Control Policy}\label{section:control}

In this section, we propose a policy for the intersection coordination problem we introduced in Section~\ref{section:problem}. Our policy is based on the polling policies that appear in the polling systems literature~\cite{Takagi:1998vi}. Before describing our policy, we briefly introduce queueing and polling systems and their analyses, in order to provide a framework for scheduling vehicles in an efficient manner. This framework is rich enough that we can apply it to our intersection management problem as defined in Section~\ref{section:problem}, providing us our performance bounds. 

\subsection{Background on Queueing and Polling Systems}\label{subsection:queue}

A widely-studied queueing model is the following. 
Suppose customers arrive into a single queue in which their requests are processed by a server, one by one, in the order that they arrive. Let $t_i$ denote the arrival time of the $i$th customer. Let $s_i$ denote the time it takes for this customer to be serviced by the server, once it gets the server's attention. Often, both $\{t_i : i \in \naturals \}$ and $\{s_i : i \in \naturals\}$ are stochastic processes. Given the statistics of arrival times and service times, the queueing theory literature asks the following types of questions: What is the average wait time for a typical customer? What is the time average queue length? Queueing theory aims to answer these kinds of questions for a variety of arrival and servicing models. This literature has found profound applications in a number of domains~\cite{Gross:1998vr}, including urban traffic~\cite{Larson:1999vs}.

Polling systems are extensions of queueing systems. In a polling system, the server services multiple sets of customers arriving in different queues. 
The server may choose to serve the first customer of any queue. However, the server must pay a set-up cost, also called the server switchover time, each time it serves customers coming from a queue that is different than the queue of the previous customer. 
This situation may arise, for example, in manufacturing machines that need to change their tooling each time they switch to processing a different kind of manufacturing good, where changing the tooling may require time. Another important application domain is data networks. For instance, consider a network switch that is serving packets arriving from different channels, and switching channels requires running a piece of setup software that takes a non-negligible amount of time to execute~\cite{Boon:2011dw}.

Let us describe a widely-studied polling system example that extends the queueing system example we discussed earlier. 
Consider a polling system with two queues and one server. Let $t_{i,k}$ denote the time that customer $i$ arrives in queue $k$, and let $s_{i,k}$ denote the amount of time it takes the server to service this customer, where $k\in \{1,2\}$. Both $\{t_{i,k} : i \in \naturals\}$ and $\{s_{i,k} : i \in \naturals\}$ are stochastic processes, for $k=1,2$.
The customers are serviced by the server one at a time. 
Switching from one queue to the other requires a switchover time, say $r$, which is also a random variable. That is, if the server last serviced a customer from queue 1 and it decides to service a customer from queue 2 next, then a switchover time of $r$ time units is incurred in addition to the service time of the customers; no switchover time is required, if the server decides to continue servicing customers from queue 1.\footnote{Strictly speaking, we must specify how the switchover time is incurred when all queues are empty. For example, in the wait-and-see system (see~\cite{Takagi:1998vi}), if all queues are empty, then the server waits in the queue it last serviced. If a customer from this queue arrives, then no switchover time is incurred. However, if a customer from a different queue arrives, then the server incurs the switchover time and switches to the other queue.}
%
%

Central to the operation of a polling system is a controller that decides which queue to serve next. This decision is a determining factor for the performance of the system, for instance, in terms of the average delay and queue length. 
In both cases, the control policy must trade off the following. On the one hand, it should switch between different queues often enough, in order to ensure customers in one queue do not wait too long for the server to process customers in a different queue; on the other hand, the server should not switch too often in order not to incur lengthy switchover times.
%
%
Ultimately, customers are clustered into chunks to be processed one right after another.
%
Best polling policies are those that achieve the best cluster sizes in an online manner.

The polling systems literature analyzes the performance of various polling policies~\cite{Takagi:1998vi}. Popular examples include the {\em exhaustive policy}, {\em gated policy}, and the {\em $k$-limited policy}. 
In the {\em exhaustive policy}, the server continues to service customers from the same queue until that queue is empty. 
In the {\em gated policy},  right after the server switches over to a new queue, it takes a snapshot of this queue; and the server services only those customers in the snapshot, but not the customers that arrive after the snapshot is taken.
%
In the {\em $k$-limited policy}, the server services the customers in the same queue until either $k$ customers are serviced or all customers in the queue are serviced, whichever comes first. Once these customers are serviced, the server switches to the next queue. 
These policies can be formalized easily. For details, we refer the reader, for example, to the seminal paper by Takagi~\cite{Takagi:1998vi}.

Most of the existing literature focuses on the case when the interarrival times of the customers have independent identical memoryless distribution, \ie, the arrival times process $\{t_{i,k} : i \in \naturals \}$ is a Poisson process for all queues $k \in \{1,2\}$. For example, the necessary and sufficient conditions for stability, the expected delay, and the steady state queue length are known for all of the three policies considered above~\cite{Takagi:1998vi}. In fact, these values can be computed when the number of queues is more than two and the intensity of arrival times is different across queues~\cite{Takagi:1998vi}. 
For more general arrival processes, performance bounds and stability results have been derived for the exhaustive and gated policies~\cite{altman1996bounds,altman1993performance,cruz1991calculus}.
Mean waiting times for exhaustive policy is always better than gated policy in both a symmetric polling system and in a heavily unbalanced polling system~\cite{rubin1983message}.

Unfortunately, it is analytically challenging to find optimal policies~\cite{duenyas1996heuristic,Takagi:1998vi}. However, approximation results for limiting cases are available. For example, on the one hand, the exhaustive policy is known to induce lower delay when compared to the gated polling system in a light-load regime, \ie, when the intensities of the arrival times are close to zero for both queues~\cite{Takagi:1998vi}; 
Furthermore, for the $G/G/1-$type polling system (general arrival process, general service time distributions, single server, no stochastic independence assumptions), the exhaustive policy was shown to be the optimal polling policy in the sense of minimizing {\it the total amount of unfinished work in the system}, which is the weighted sum (based on load) of the mean waiting times~\cite{levy1990dominance,boxma1993efficient,Vishnevskii:2006dd}. In a fully symmetric system, \ie, mean service times and arrival rates are the same across all queues, this total workload is exactly the mean customer delay. In a system in which mean service times are the same, but the arrival rates are allowed to differ across queues, the total workload is exactly the mean waited customer delay.
Moreover, optimal polling policies have been investigated for the $GI/GI/1-$type polling system (general arrival process, general service time distribution, single server, all stochastic processes are mutually independent). The exhaustive policy is optimal in the sense of minimizing the number of customers in the system \cite{liu1992optimal}.
For more general optimization criteria, such as fairness, discounting, including switchtime cost, similar results are available for the analytically-more-tractable $M/GI/1$-type polling system~\cite{hofri1987optimal,duenyas1996heuristic}. An extensive survey of optimal policies for various polling systems is given by Vishnevskii~\cite{Vishnevskii:2006dd}.

\subsection{Simulating Polling Systems Behavior}\label{subsection:simulateQueue}

The coordination algorithms we present below heavily rely on the polling systems policies. In particular, in a number of places we simulate the behavior of a polling policy forward in time. We devote this section to formalizing this procedure.

Consider a polling system with two queues and deterministic service and set-up times. That is, the customers arrive at times $\{t_{i,k} : i \in \naturals\}$, where $k \in \{1,2\}$; their service requires $s$ time units, and the server requires $r$ time units to switch queues. 
The arrival times of the customers are stochastic processes. For this section, consider the case when the variables $s$ and $r$ are fixed (\ie, deterministic and same) for all the customers. 
The coordination algorithm we present in Section~\ref{subsection:policy} relies on simulating the behavior of a polling system that has fixed service time $s$ and fixed set-up time $r$. 

For notational convenience, we represent a polling system with the symbol ${\cal P}$. 
The algorithms we describe below interact with a polling system through two procedures. The procedure ${\cal P}.{\tt AddToQueue}(k)$ adds one customer to queue $k$. 
The ${\cal P}.{\tt Simulate}()$ procedure simulates the behavior of the polling system, assuming no additional customers arrive. More precisely, the ${\cal P}.{\tt Simulate}()$ returns two sequences of time instances, namely ${\cal T}_1 = (\tau_{i_1,1}, \tau_{i_2,1}, \dots, \tau_{i_{n_1},1})$ and ${\cal T}_2 = (\tau_{j_1,2}, \tau_{j_2,2}, \dots,\tau_{j_{n_2},2} )$, where $\tau_{i,k}$ is the time that the server is scheduled to begin servicing the $i$th customer in queue $k$, assuming no additional customers arrive. 
We call $\tau_{i,k}$ the {\it schedule time} of vehicle $i$ in lane $k$.
Note that the behavior of the polling system ${\cal P}$ depends on its polling policy (\eg, exhaustive, gated, $k$-limited, {\em et cetera}), the service time $s$, the set-up time $r$, the number of customers in the two queues, and the time the server began servicing the current customer (if the server is currently serving any customers) or the time that the set-up operation started (if the server is currently switching to the next queue).
Given the values of all these variables, the ${\cal P}.{\tt Simulate}()$ procedure is well-defined, in the sense that the sets ${\cal T}_1$ and ${\cal T}_2$ are uniquely determined, if the service time $s$ and the set-up time $r$ are fixed.
Given ${\cal T}_1$ and ${\cal T}_2$, we define the corresponding {\it service order} of vehicles as
\begin{gather*}
\Big( (i_1,k_1), \ldots, (i_n,k_n) \Big),
\end{gather*}
such that the schedule time is strictly increasing in this sequence, \ie, for all $m > 1$,
\begin{align*}
	\tau_{i_{m-1},k_{m-1}}  < \tau_{i_m,k_m}.
\end{align*}
See Figure~\ref{figure:schedule} for an example of computing schedule times and service order under the exhaustive policy.
\begin{figure}[b]
	\centerline{\includegraphics[clip=true, trim=2.1in 6.82in 4.1in 3.71in, width=0.5\textwidth]{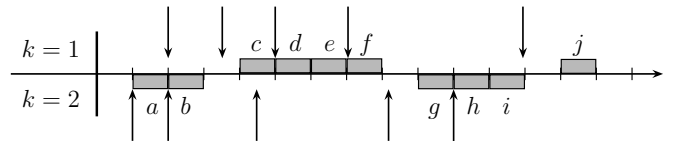}}
	\caption{Arrival times for vehicles in lanes 1 and 2 are depicted with the vertical arrows. Newly arrived vehicle $j$ is in lane 1. Shaded blocks show the result of simulating the polling system under the exhaustive policy. Service time $s$ and switch time $r$ are assumed to be $s = r =1$. ${\cal T}_1 = (4,5,6,7,13)$. ${\cal T}_2 = (1,2,9,10,11)$. Schedule order for this instance is $\Big( (a,2), (b,2) , (c,1) , (d,1), (e,1), (f,1) , (g,2), (h,2), (i,2) , (j,1) \Big)$.} \label{figure:schedule}
\end{figure}

\subsection{The Intersection Coordination Algorithm} \label{subsection:policy}
We propose an event-triggered coordination algorithm that plans the motions of all vehicles that enter the control region as described in Section~\ref{section:problem}. More precisely, the control algorithm computes trajectory $x_{i,k}$ for each vehicle $(i,k)$ in the control region, such that trajectory $x_{i,k}$ is dynamically feasible and no two vehicles collide. The coordination algorithm is event-triggered in the sense that each trajectory $x_{i,k}$ is updated whenever a new vehicle arrives at the control region.

The core procedure embedded in this coordination algorithm is presented in Algorithm~\ref{algorithm:main}.
Each time a new customer arrives in queue $k$, this procedure is triggered. The procedure computes and returns the trajectories $x_{i,1}$ for all $i \in \{i_1,i_2, \dots, i_{n_1}\}$ and $x_{j,2}$ and all $j \in \{j_1,j_2,\dots, j_{n_2}\}$, each time a new vehicle arrives at the control region. 

Before presenting the coordination algorithm, let us introduce the following motion planning procedure. 
The ${\tt MotionSynthesize}$ procedure generates a trajectory for each vehicle given the time this vehicle must reach the intersection region and the trajectory of the vehicle directly in front of it. Furthermore, the trajectory is designed such that the vehicle stays as closely as possible to the intersection region at all times. 
Denote $t_0'$ as the arrival time of this vehicle, and $t_f'$ as the time to reach the intersection region. Suppose the trajectory of the vehicle in front is denoted by $y : [t_0,t_f] \to \reals$, where $t_0$ is its arrival time in the control region and $t_f$ is the time it is scheduled to reach the intersection region.
Then, the ${\tt MotionSynthesize}$ procedure computes a new trajectory such that the two vehicles do not collide and the new trajectory reaches the intersection at time $t_f'$ with full speed $v_m$, \ie, 
\begin{align*}
&\hspace{-0.5in}{\tt MotionSynthesize}(z_{i,k}(t_0'),t_0',t_f',y) := \hspace{1in}\\[1ex]
\argmin_{x \colon [t_0',t_f'] \to \reals} &\quad \int_{t_0}^{t_f} \vert x(t) \vert dt \\
\mathrm{subject}\,\,\mathrm{to} &\quad \ddot{x}(t) = u(t), \mbox{ for all } t \in [t_0', t_f'];\\
&\quad 0 \le \dot{x}(t) \le v_m, \mbox{ for all } t\in [t_0', t_f'];  \\
& \quad \vert u(t) \vert \le a_m, \mbox{ for all } t\in [t_0', t_f'];  \\
& \quad \vert x(t) - y(t) \vert \ge l, \mbox{ for all }t \in [t_0', t_f]; \\
& \quad x(t_0') = x_{i,k}(t_0');  \quad \dot{x}(t_0') = \dot{x}_{i,k}(t_0');\\
& \quad x(t_f') = 0; \quad \dot{x}(t_f') = v_m,
\end{align*}
where $z_{i,k}(t_0')$ is the state of vehicle $(i,k)$ at time $t_0'$, \ie, $z_{i,k}(t_0') := \big(x_{i,k}(t_0'),\dot{x}_{i,k}(t_0')\big)$, $v_m$ and $a_m$ are the maximum velocity and the maximum acceleration of the vehicles, $l$ is the length of the vehicle, and $L$ is the control region length.

Each time a new customer arrives in the control region from lane $k$, Algorithm~\ref{algorithm:main} is triggered. First, the algorithm executes the ${\tt AddToQueue}(k)$ procedure, which adding into queue $k$ of polling system ${\cal P}$ one customer that represents the newly arrived customer (Line~\ref{line:add_to_queue}). Then, the algorithm simulates the polling system forward in time assuming no new customers will arrive, with fixed service time $s = l/v_m$ and fixed set-up time $r = w/v_m$ (Line~\ref{line:simulate}). Note that the service time $s$ is precisely the amount of time from when a vehicle's front bumper enters the intersection region to when its rear bumper leaves the control region assuming the vehicle is traveling at $v_m$. The set-up time on the other hand is the amount of time from when the vehicle's rear bumper leaves the control region to when its rear bumper exits the intersection region. The result of the simulation is a sequence of times, namely ${\cal T}_1$ and ${\cal T}_2$.
Let us denote the $i$th element of ${\cal T}_k$ by $\tau_{i,k}$.
Finally, the algorithm generates a trajectory using the ${\tt MotionSynthesize}$ procedure for each vehicle such that the $i$th vehicle in lane $k$ is scheduled to reach the intersection at time $\tau_{i,k} + L/v_m$ while avoiding collision with the vehicle in front (Lines~\ref{line:motion_synthesize_start}-\ref{line:motion_synthesize_end}). To be precise, define $x_{0,k}(t) = l$ for all $t \in \reals_{\ge 0}$ and $k \in \{1,2\}$; then, ${\tt MotionSynthesize}$ computes $x_{i,k}$ for all $i \neq 0$.\footnote{Since the vehicle closest to the intersection region in each lane does not have a vehicle directly in front, the safety constraint in the ${\tt MotionSynthesize}$ procedure, \ie, $|x(t)-y(t)| \geq l$, can be ignored.}

\begin{algorithm}[t] \small 
${\cal P}.{\tt AddToQueue}(k)$\; \label{line:add_to_queue}
$({\cal T}_1, {\cal T}_2) \leftarrow {\cal P}.{\tt Simulate}()$\; \label{line:simulate}
\For{$k = 1,2$}{\label{line:motion_synthesize_start}
	\For{$i = 1,2,\dots,n_k$}{
		$\tau_{i,k} \leftarrow {\cal T}_k(i)$\;
		$x_{i,k} \leftarrow {\tt MotionSynthesize}(x_{i-1,k}, \tau_{i,k}+  \frac{L}{v_m} )$\; \label{line:motion_synthesize}
	}
}\label{line:motion_synthesize_end}
\caption{\small ${\tt NewArrival}(k)$ procedure is triggered when a new vehicle arrives in the control region from lane $k$.}\label{algorithm:main}
\end{algorithm}
We emphasize that the coordination algorithm depends on the polling policy that governs the polling system ${\cal P}$. In the next section, we show that safety is guaranteed for a wide range of polling policies and that performance can be bounded with respect to the performance of the polling system ${\cal P}$, which in turn depends on the said polling policy.

\begin{figure}
\centerline{\includegraphics[width=0.5\textwidth]{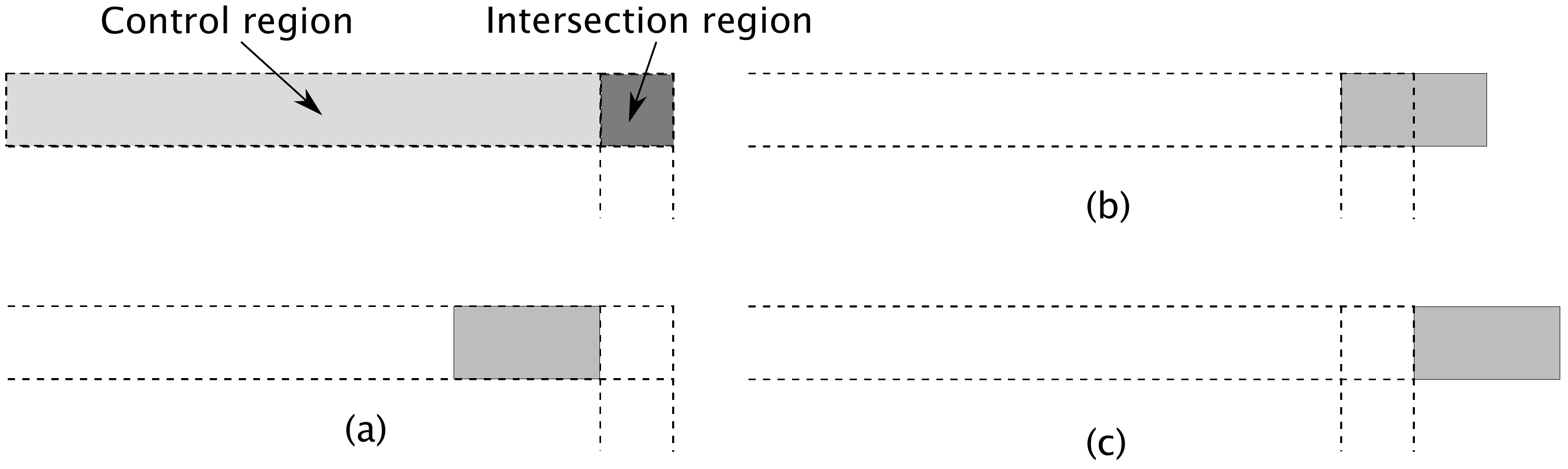}}
\caption{We show the three time instances that are described in the text. In (a), the front bumper of the vehicle is entering the intersection region. In (b), the rear bumper is leaving the control region and entering the intersection region. In (c), the rear bumper is leaving the intersection region. In the figure above, the control region (shaded light grey) and the intersection region (shaded in dark grey) are shown.} 
\label{figure:intersection_crossing}
\end{figure}

\section{ANALYSIS} \label{section:analysis}

In this section, we show that the coordination algorithm presented in Section~\ref{subsection:policy} has two important properties. Firstly, the algorithm is safe in the sense that no collisions occur in the control region. (See Definition~\ref{definition:safety}.) Second, the coordination algorithm provides good performance. Recall that the coordination algorithm simulates a polling system. We show that the additional delay incurred in traversing the control region is no more than the delay that would incur in servicing a customer in the corresponding polling system.

These guarantees hold under certain assumptions.
In what follows, we present a class of stochastic process models for vehicle arrival times and two important assumptions. Subsequently, we state and prove our theoretical results. 

\subsection{Vehicle arrival time model} \label{subsection:arrivalModel}

First, we model the arrival times $\{t_{i,k} : i \in \naturals\}$ as a hard-core stochastic point process on the non-negative real line~\cite{Stoyan:1995ve} such that the $t_{i+1,k} - t_{i,k} \ge l/v_m$ for all $i \in \naturals$ and $k \in \{1,2\}$.\footnote{Roughly speaking, hard-core point processes are those that ensure the distance between its points is lower bounded by a certain number. See~\cite{Stoyan:1995ve}.} The inequality guarantees that no two vehicles are in collision at the time of arrival. A widely studied hard-core stochastic point process is the Mat\'ern process~\cite{Teichmann:2012vj,Baccelli:2012jma,Stoyan:1985vwa}. First, a marked Poisson point process is generated on the non-negative real line, where each point $t_{i}$ is marked with an independent uniformly random real number between zero and one denoted by $m(t_i)$.
Subsequently, any point $t_i$ that satisfies the following is deleted: there exists another $t_j$ such that the distance between $t_i$ and $t_j$ is no more than $d$ and $m(t_j) > m(t_i)$. Clearly, no two points have distance less than $d$. Hence, the Mat\'ern process is a hard-core point process. 
{\color{red} }

In addition, we model the effect of overcrowding in the control region to rule out cases that trivially contradict safety. More precisely, we tacitly assume that, if vehicle $i$ arrives in lane $k$ at time $t_{i,k}$ and at this time there is no input that saves it from hitting the car in front, then the same vehicle chooses not to enter the intersection.\footnote{For practical purposes, we envision a system with a fork right at the entrance of the control region, and when it is clearly unsafe to enter the control region (\ie, it is impossible to avoid a collision), then the same vehicle takes the exit at the fork and does {\em not} enter the control region.} 
Let us emphasize that this assumption does not guarantee safety immediately. That there exists some path that does not lead to a collision does not immediately imply that there is a path that is both safe and provides performance guarantees. In other words, this assumption does not render the problem trivial. 
Without such an assumption, however, the system is (trivially) unsafe. 
We emphasize that it is impossible to guarantee safety without this assumption. 
That is, without it, one is faced with trivial cases where any coordination algorithm is unsafe with a non-zero probability. However, we show through computational experiments that the chance that a vehicle must choose not to enter the intersection is small, when the road length $L$ is very large or the intensity of the arrival times is very small (\ie, the light load case). We conjecture that our analysis applies in these limiting cases, even without this assumption. 

\newpage

\subsection{Assumptions on polling policies and road length}
Our first assumption is rather technical. 
Recall that the behavior of the coordination algorithm depends on the policy of the polling system it simulates. We introduce a regularity assumption on this policy. Roughly speaking, we assume that the policy does {\em not} respond to the arrival of a customer in queue $1$ by favoring the servicing of customers in queue $2$, and {\em vice versa}. The assumption is formalized as follows.
\begin{assumption}[Regular polling policies] \label{assumption:regular_policy}
Suppose customer $A$ arrives at time $t_0'$ in queue 1.
Let $\Big( (i_1,k_1), \ldots, (i_n,k_n) \Big)$ be the service order of customers computed at the previous call to Algorithm 1 before the arrival of customer $A$.
At time $t_0'$, we add customer $A$ to queue 1 and afterwards simulate the polling system according to ${\cal P}$ , i.e., ${\cal P}.{\tt AddToQueue}(1)$; $({\cal T}_1',{\cal T}_2') \leftarrow {\cal P}.{\tt Simulate}()$ to obtain $({\cal T}_1', {\cal T}_2')$.
Then, we say that the polling policy ${\cal P}$ is {\em regular} if customer $A$ is simply inserted into the service order of customers, without otherwise disturbing the service order of customers, \ie, the updated service order is $\Big( (i_l,k_l), \ldots, (i_{m-1},k_{m-1}), (A,1), (i_m,k_m), \ldots (i_n,k_n) \Big)$, where $1 \leq l \leq m \leq n$.
The same result also holds for queue 1, whenever a new customer is added to queue 2. 
\end{assumption}

Let us note that this assumption is satisfied by many widely-studied polling policies, as stated in the proposition below. 
\vspace{-0.1in}
\begin{proposition}
The exhaustive, gated, and the $k$-limited polling policies (see Subsection~\ref{subsection:queue}) satisfy Assumption~\ref{assumption:regular_policy}.
\end{proposition}
\vspace{-0.1in}
\begin{proof} (Sketch) This result can be verified easily for each of the three policies. We sketch the proof only for the exhaustive policy. Without loss of generality, suppose a new customer arrives in queue 1. First, suppose the server is currently servicing a customer from queue 1. This would cause the queue 2 customers to be serviced after the newly arrived customer, while not affecting the service order of customers in queue 1. Second, suppose the server is servicing a customer from queue 2 at the arrival time of the new customer. Then, the the newly arrived customer is simply inserted at the end of the previous service order, not affecting the schedule order of the customers in the polling system. Hence, Assumption~\ref{assumption:regular_policy} holds for the exhaustive policy. 
Similarly, it can be shown that Assumption~\ref{assumption:regular_policy} is satisfied by gated and $k$-limited policies, merely by utilizing the same proof technique.
\end{proof}

Our second assumption is that the length of the control region is bounded from below by a certain distance $L^*$. Note that this minimum road length $L^*$ is a function of the dynamics of the vehicles, but not of the geometry of the vehicles.
\begin{assumption}[The length of the control region] \label{assumption:control_region_length}
The length $L$ of each of the control regions is bounded from below as follows:
$
L \,\,\ge\,\, 2 v_m^2 / a_m := L^*.
$
\end{assumption}
As we will show later, this assumption guarantees that the control region is long enough to allow the safe coordination of vehicles while guaranteeing good performance.

\subsection{Main theoretical results: Safety, performance and stability}

Our first two main theoretical results are the safety and performance guarantees, presented below in Theorems~\ref{theorem:safety} and \ref{theorem:performance}. These results are enabled by the following lemma. 

\begin{lemma} \label{lemma:main}
Suppose Assumptions \ref{assumption:regular_policy} and \ref{assumption:control_region_length} hold. Then, each time a new vehicle arrives and  Algorithm~\ref{algorithm:main} is called, every call to the ${\tt MotionSynthesize}$ procedure (Line~\ref{line:motion_synthesize} of Algorithm~\ref{algorithm:main}) yields a feasible optimization problem.
\end{lemma}

Before providing the proof of this lemma, let us point out two important corollaries, which are our main results.

First, the algorithm guarantees safety surely. In other words, it is guaranteed that no collisions occur as new customers arrive at the control region.
\vspace{-0.1in}
\begin{theorem}[Safety] \label{theorem:safety}
Suppose Assumptions \ref{assumption:regular_policy} and \ref{assumption:control_region_length} hold. Then, the control region is safe at all times $t \ge 0$ in the sense of Definition~\ref{definition:safety}.
\end{theorem}
\vspace{-0.1in}

Second, the delay experienced by the vehicles is bounded by the delay of the corresponding polling system.

\vspace{-0.1in}
\begin{theorem}[Performance] \label{theorem:performance}
Suppose Assumptions \ref{assumption:regular_policy} and \ref{assumption:control_region_length} hold. Recall that the delay incurred in transitioning through the control region for vehicle $i$ in lane $k$ is denoted by $D_{i,k}$. (See Section~\ref{section:problem}.) Let $W_{i,k}$ denote the wait time of the corresponding customer added to the polling system (see Algorithm~\ref{algorithm:main}, Line~\ref{line:add_to_queue}) when vehicle $i$ arrives in the control region from lane $k$. Then, we have the following:
$$
D_{i,k} \,\,\leq\,\, W_{i,k}, \qquad \mbox{almost surely.}
$$
\end{theorem}

In essence, Theorem~\ref{theorem:performance} states that the differential constraints that bind the customers as they traverse the control region are irrelevant. The delay is no more than the delay that is incurred in the corresponding polling system. 
Hence, one can employ any polling policy in the coordination algorithm and guarantee the same performance provided by this policy. 
We emphasize that the expected wait time is known for many polling policies, in particular the exhaustive and gated policies, when the arrival process is Poisson~\cite{Takagi:1998vi}.

Our third main result is a stability criterion, presented in Theorem~\ref{theorem:stability}.
Let us define $A_{k}(t)$ as the total number of vehicles that have arrived in lane $k$ during the time interval $[0,t]$, according to the arrival process $\{t_{i,k} : i \in \naturals\}$. Next, let us define $\theta^{L,{\cal P}}_k (t)$ as the number of lane $k$ vehicles that are thinned due to the overcrowding assumption during the time interval $[0,t]$. Note that $\theta^{L,{\cal P}}_k(t)$ changes with the road length $L$ and the polling policy ${\cal P}$ chosen for the algorithm; however, when the polling policy is clear from context we simply write $\theta^L_k(t)$. 
We are ready to present the notion of stability.

\begin{definition}{\bf (Stability)}  \label{definition:stability}
We say that the coordination algorithm with polling policy ${\cal P}$ and arrival process $\{t_{i,k}\}$ is stable if 
$$
\lim_{L \to \infty} \lim_{t \to \infty} \EE[\frac{\theta^{L,{\cal P}}_1(t) + \theta^{L,{\cal P}}_2(t)}{t}] = 0.
$$
If the coordination algorithm is not stable, in the above sense, then we say that the algorithm is unstable.
\end{definition}
Note that the stability of the coordination algorithm is dependent only on the arrival process and the polling policy and does not depend on the length of the control region. Intuitively, this definition of stability means that as we make the control region longer and longer, the intensity of the thinned vehicles tends to 0.
Next, we define $S^{L,{\cal P}}_k(t) = A_k(t) - \theta^{L,{\cal P}}_k(t)$. This is the number of lane $k$ vehicles that enter the control region and are assigned a schedule time by the coordination algorithm during the time interval $[0,t]$.
Again, when clear from context, the superscript ${\cal P}$ denoting the polling policy will be dropped. 
We present the following lemma.

\begin{lemma}
For any control region length $L > 0$, and any polling policy ${\cal P}$, we have the following:
$$
\lim_{t \to \infty} \EE[\frac{S^{L,{\cal P}}_1(t)+S^{L,{\cal P}}_2(t)}{t}] \leq \frac{1}{s},
$$
where $s$ is the service time of a vehicle.
\end{lemma}

Now, we present a sufficient condition for the {\em instability} of the coordination algorithm.
By ergodicity of the arrival processes, we have that for any $t$, $\EE[A_k(t)/t] = \lambda_k$, which is the intensity of the arrival process in lane $k$.

\begin{theorem}[Sufficient condition for instability] \label{theorem:stability}
Let $A_k(t)$ count the arriving vehicles in lane $k$, and $\lambda_k$ be the intensity of the arrival process in lane $k$. The coordination algorithm is unstable, if
$$
\lambda_1 + \lambda_2 > \frac{1}{s}.
$$
\end{theorem}

\begin{proof} For any $L > 0$ and $t > 0$, we have
\begin{align*}
	\theta^L_1(t) + \theta^L_2(t) = A_1(t) + A_2(t) - \left(S^L_1(t) + S^L_2(t) \right).
\end{align*}
Dividing by $t$, taking expectation, and then taking the limit, we obtain
\begin{align*}
\lim_{t \to \infty} \EE[\frac{\theta^L_1(t) + \theta^L_2(t)}{t}] &= \lambda_1 + \lambda_2 - \lim_{t \to \infty} \EE[\frac{S^L_1(t) + S^L_2(t)}{t}] \\
&\geq \lambda_1 + \lambda_2 - \frac{1}{s} > 0.
\end{align*}
Since this holds for any $L > 0$, the system is unstable. 
\end{proof}

Furthermore, we conjecture the following.
\begin{conjecture}[Sufficient condition for stability] Let ${\cal P}$ be any regular polling policy. If
\begin{align}
	\lambda_1 + \lambda_2 < \frac{1}{s}, \label{sufficient_stability}
\end{align}
then the coordination algorithm with policy ${\cal P}$ is stable.
\end{conjecture}
This sufficient condition for stability is inspired by the similar results in the polling system literature~\cite{Takagi:1998vi}. 
The restriction to regular polling policies in the above conjecture is crucial. 
For instance, consider the polling policy that only allows one lane access and completely stops the other lane; this policy is not a regular polling policy.
Even when the arrival rates $\lambda_1$ and $\lambda_2$ are chosen to satisfy Inequality~\eqref{sufficient_stability}, the intensity of thinned vehicles will always be greater than some $\epsilon > 0$ for any $L$. Thus, this (non-regular) policy is unstable.

\section{Computational Experiments}\label{section:simulations}

In this section, we evaluate the proposed coordination algorithm in simulation studies.
First, we describe the simulation setup and the computational methods used in this study in Section~\ref{section:approximate}. 
Subsequently, we describe the results of computational experiments for light-load, medium-load and heavy-load cases in Section~\ref{section:platooning}, where we observe the platooning phenomenon. 
Next, we study the effects of the overcrowding assumption in Section~\ref{section:overcrowding}. 
Lastly, we compare the performance of our algorithm with an intersection managed by a traditional traffic light with human drivers in Section~\ref{section:trafficLight}. 

\notraj{
\begin{figure*}
\centerline{\includegraphics[clip=true, trim=2.8in 0.02in 2.65in 0.33in, width=1.09\textwidth]{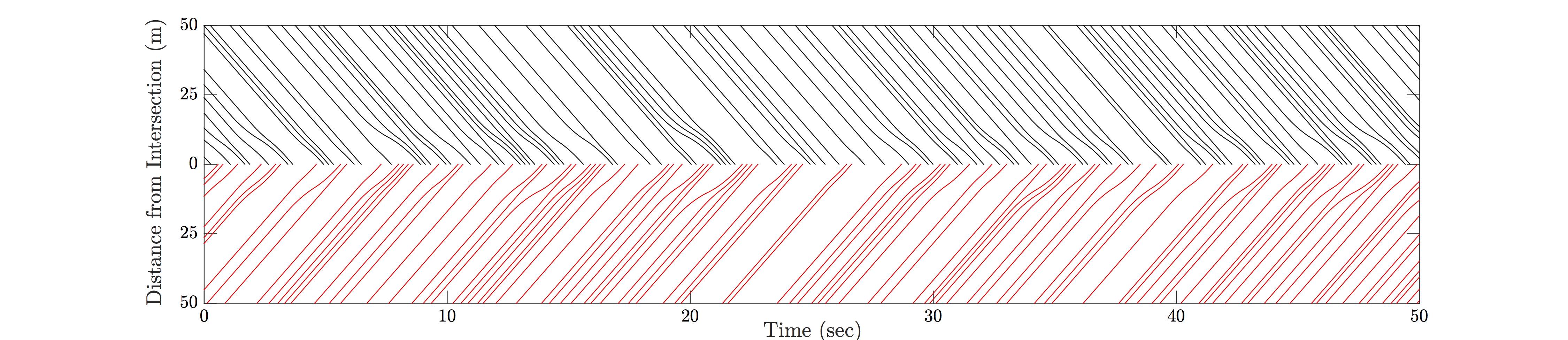}}
\centerline{\includegraphics[clip=true, trim=2.8in 0.02in 2.65in 0.33in, width=1.1\textwidth]{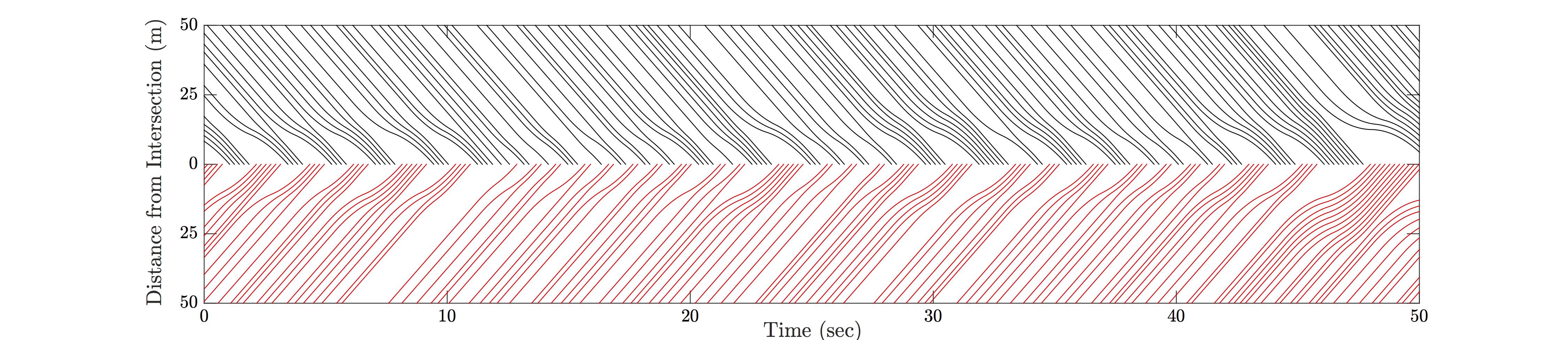}}
\centerline{\includegraphics[clip=true, trim=1.3in .05in 1.3in .15in, height=.19\textheight]{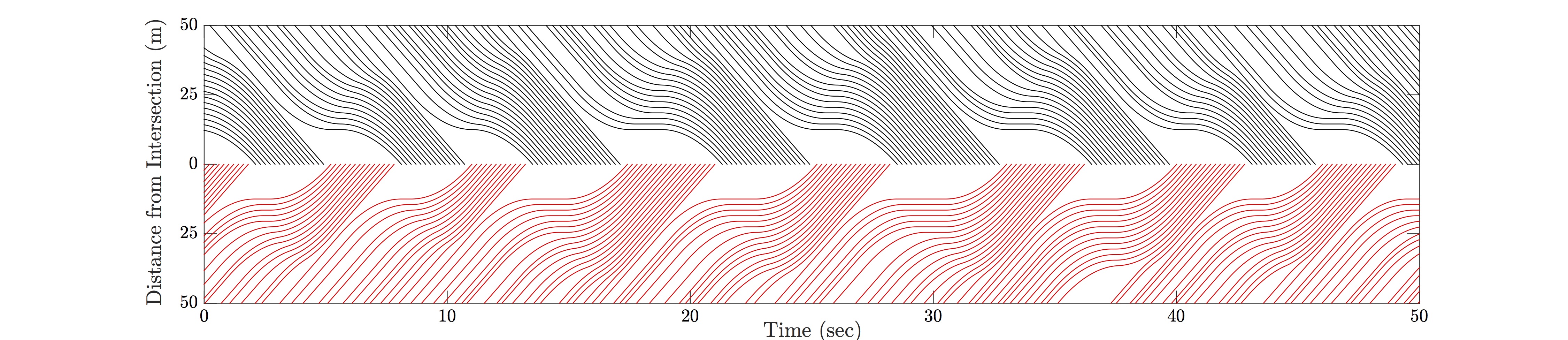}}
\caption{Trajectories of vehicles controlled by Algorithm~\ref{algorithm:main} with the exhaustive polling policy. The top, middle, and bottom figures are representatives for the light, medium, and heavy load cases, respectively. The top half of each figure shows the trajectories of vehicles in Lane 1, and the bottom half shows the trajectories of vehicles from Lane 2. Each figure shows a small window of time after the system reaches steady state.}\label{figure:trajectories}
\end{figure*}

\begin{figure*}
\centerline{\includegraphics[clip=true,trim=2.8in 0.02in 2.65in 0.33in, width=1.1\textwidth]{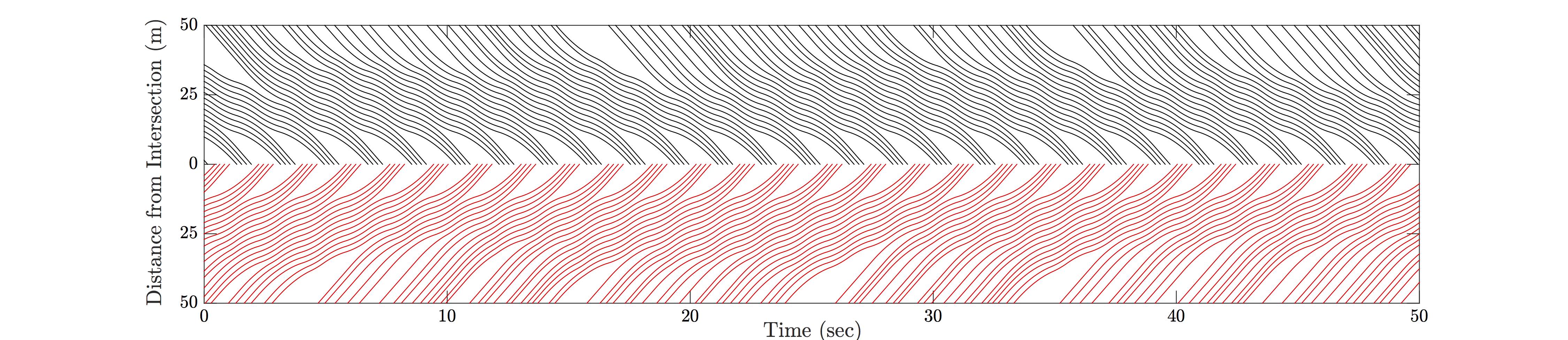}}
\centerline{\includegraphics[clip=true,trim=2.8in 0.02in 2.65in 0.33in, width=1.1\textwidth]{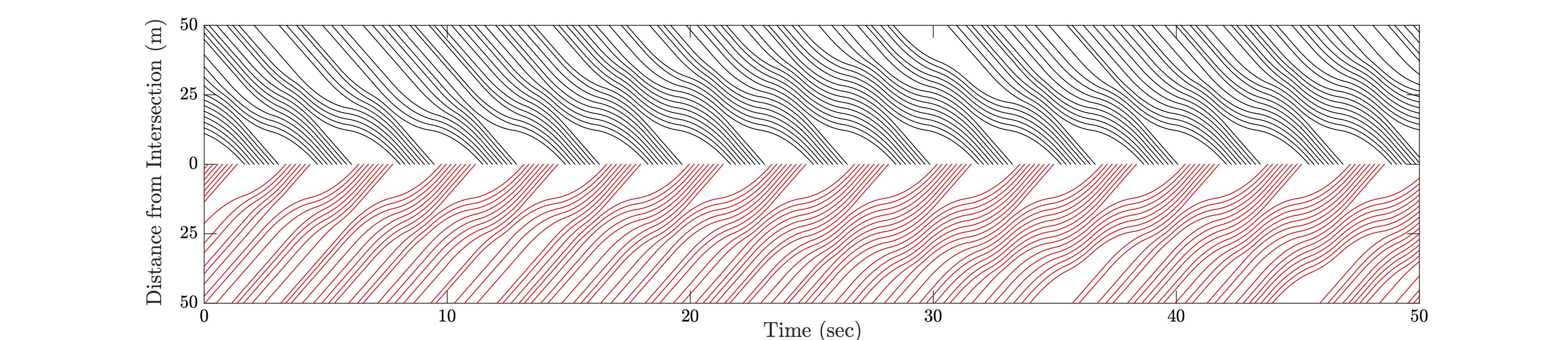}}
\caption{Trajectories of vehicles controlled by Algorithm~\ref{algorithm:main} with the $k$-limited polling policy, where $k=4$ in the upper plot and $k = 8$ in the lower plot. The top half of each figure shows the trajectories of vehicles in Lane 1, and the bottom half shows the trajectories of vehicles from Lane 2. Each figure shows a small window of time after the system reaches steady state.}
\label{figure:ktraj}
\end{figure*}
}
\begin{figure}
	\centerline{\includegraphics[clip=true, trim=.22in .05in .51in .26in, width=0.5\textwidth]{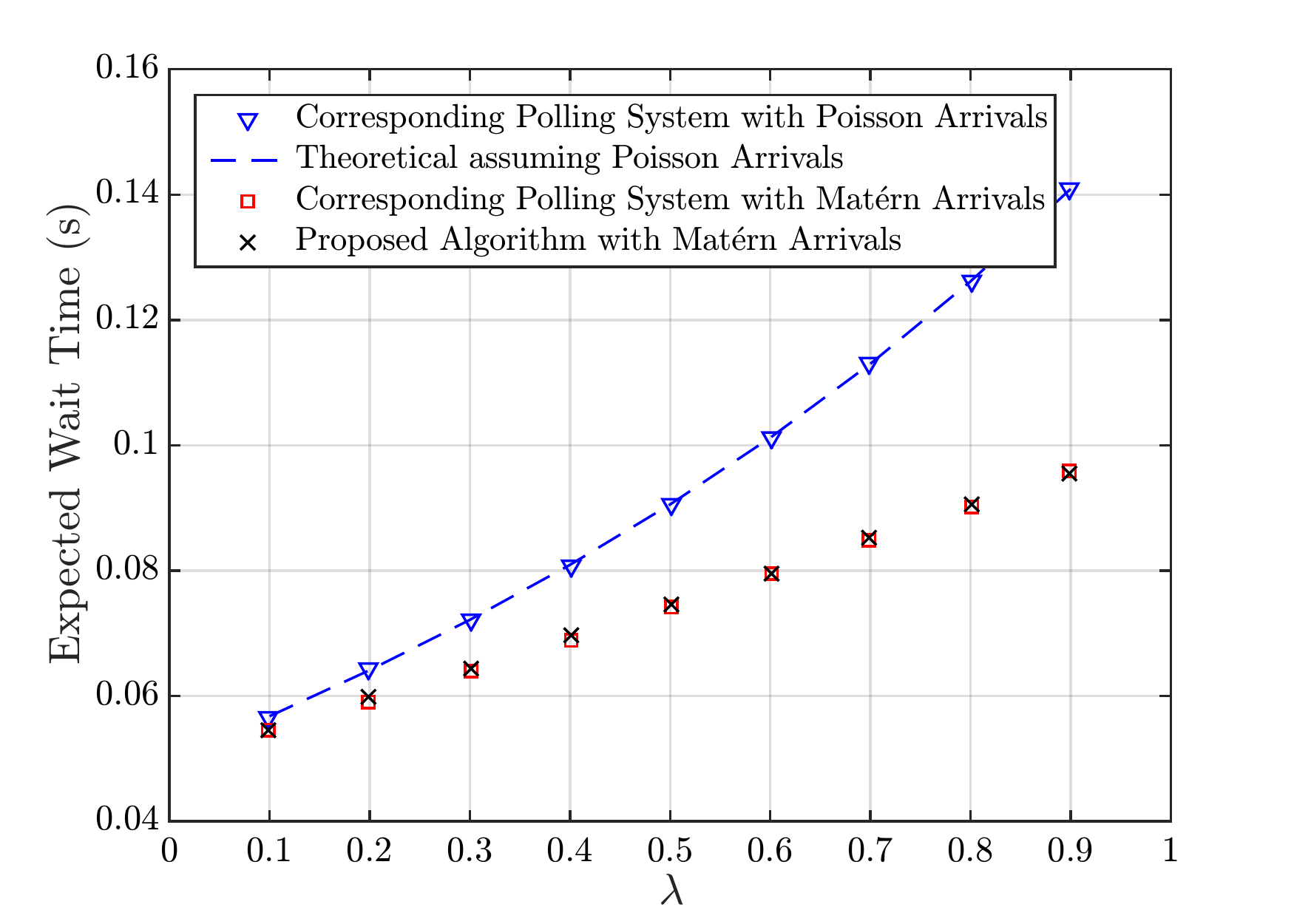}}
	\caption{Comparison of the expected wait time of Algorithm~\ref{algorithm:main} and the corresponding polling system according to a Mat\'ern arrival process with parameter $\lambda$. The expected wait time of a polling system with Poisson arrival process with intensity $\lambda$ is computed via simulations and also using polling systems theory. All polling policies in this figure use exhaustive polling policy.} \label{figure:performance}

\end{figure}

\subsection{The simulation environment and computational methods}
\label{section:approximate}
We used Matlab to generate the simulation results presented in this section, solving the linear programs with Gurobi~\cite{gurobi}.
Incoming vehicle arrival times were generated by a Mat\'ern process~\cite{Stoyan:1985vwa}, independent in each lane, but symmetric, \ie, equal arrival intensities.
The vehicle length, width, maximum velocity, and maximum acceleration were taken to be 2 meters, 1 meter, 10 meters/sec, and 4 meters/sec$^{2}$, respectively.

Vehicle trajectories were obtained by discretizing the ${\tt MotionSynthesize}$ procedure, leading to a linear program:
\begin{align*}
\max &\sum_{i=0}^N x_i \\
x_{i+1} &= x_i + (v_i + v_{i+1}) \cdot \Delta t /2,\text{ for all } i; \\
v_{i+1} &= v_i + u_i \cdot \Delta t, \text{ for all } i; \\
0 &\leq v_i \leq v_m, \text{ for all } i; \\
-a_m &\leq u_i \leq a_m, \text{ for all }i; \\
x_i &\leq y(t_0 + i \cdot \Delta t) - l, \text{ for all } i; \\
x_0 &= -L; \,\,x_N = 0; \\
v_0 &= v_m; \,\, v_n = v_m,
\end{align*}
where $\Delta t$ is the constant step size of the time history, $t_0$ is the arrival time of some vehicle that triggers the algorithm, $[x_0,\dots x_N,v_0,\dots,v_N,u_0,\dots,u_{N-1}]^T$ is the state vector, representing the discretized position, velocity and acceleration of the vehicle at time $t_0 + i \cdot\Delta t$, and $y$ is the trajectory of the vehicle directly in front.
Note that the trajectory $y$ is evaluated at time $t_0 + i \cdot \Delta t$ by a quadratic interpolation scheme, assuming a linear velocity history $\dot{y}$ over the relevant time interval.

Note that the running time of our algorithm scales polynomially with increasing number of vehicles, since we solve this linear program at most once for each vehicle that is in the system, when a new vehicle arrives in the system. 

We initialize the polling system with deterministic service and switching times, as mentioned in section~\ref{subsection:policy}.
We use the ``wait-and-see'' rule~\cite{Takagi:1998vi} to govern the switching dynamics. 

Expected delays were estimated by running a simulation for a long time and then averaging the delay incurred by each vehicle in the system.
We find that $N = 800$ yields sufficiently smooth trajectories. However, in general, the required degree of discretization increases with longer control regions.

\subsection{The effect of load conditions and platooning} \label{section:platooning}

In Figure~\ref{figure:trajectories}, we show trajectories of vehicles under an exhaustive policy with increasing load. 
In Figure~\ref{figure:ktraj}, we show trajectories of vehicles under a $k$-limited polling policy. 
In none of these simulations, the ${\tt MotionSynthesize}$ procedure reported infeasible, which is in line with Lemma~\ref{lemma:main}, hence with Theorems~\ref{theorem:safety} and \ref{theorem:performance}. 

First, notice that the {\em ``platooning behavior''} emerges naturally, particularly when the load becomes substantial. Specifically, vehicles {\em slow down} to form a cluster with other vehicles. This cluster of vehicles speeds up and passes through the intersection at maximum speed. In this way, the intersection, which is the shared resource, is utilized as efficiently as possible. The vehicles that slow down do so at the latest time possible, so that more vehicles can be packed into the control region. 
This behavior is observed easily in Figure~\ref{figure:trajectories} as a result of utilizing the exhaustive policy. It is visible in the utilization of the $k$-limited policy as well, as shown in Figure~\ref{figure:ktraj}. This platooning behavior emerges naturally, as it also arises in polling systems, which we discussed in Section~\ref{subsection:queue}.

\newpage 
Next, we compute performance bounds on Algorithm 1 under the exhaustive policy, assuming vehicles arrive according to a Mat\'ern process.
Define the Poisson process with parameter $\lambda$ as the Poisson point process in the line that has intensity $\lambda$. 
Define the Mat\'ern process with parameter $\lambda$ as the Mat\'ern process obtained by thinning a Poisson process with parameter $\lambda$.
Note that the intensity of the Mat\'ern process with parameter $\lambda$ is
\begin{align*}
\lambda_I (\lambda) = \frac{1-\exp(-2 \lambda b)}{2b},
\end{align*}
where $b$ is the service time~\cite{Stoyan:1985vwa}.
In Figure~\ref{figure:performance}, we compare the performance of our proposed algorithm with the performance of the corresponding polling systems, which is in line with our result in Theorem~\ref{theorem:performance}.
Note that the average delay incurred in Algorithm 1 for Mat\'ern arrivals of intensity $\lambda_I(\lambda)$ seems to be bounded from above by the average wait time of an exhaustive polling system with Poisson arrivals of intensity $\lambda$.
We observe this behavior consistently in simulation studies, and we conjecture it holds for all regular polling policies.

\subsection{Unsafe arrivals in terms of road length and arrival rate} \label{section:overcrowding}

Recall the overcrowding model from Section~\ref{subsection:arrivalModel}. We remove any incoming vehicle before it ever enters the control region, if the same vehicle has no safe trajectory. We say that this vehicle is diverted away from the control region. We call this phenomenon diversion.
The result is that the the stochastic process is ``thinned'' by the removal of unsafe vehicles.

In this section, we study the effect of road length and arrival rate on unsafe arrivals. Our key observation is the following. In the computational experiments we present here, we find that the fraction/intensity of the vehicles that are diverted converges to zero exponentially fast with increasing road length as well as with decreasing vehicle arrival rate. In fact, we observe that this overcrowding phenomenon is extremely rare, when the road length is reasonably large (\eg, twice the limit given by Assumption~\ref{assumption:control_region_length}) and the arrival process slightly below the limit of instability (\eg, 90\% of instability limit).

Before presenting our simulation results, note that the choice of polling policy affects overcrowding. For the remainder of the simulations carried out in this section, we employed the exhaustive policy, although any regular polling policy can be used and should show similar results.
First, let us consider how often a vehicle may be removed. In Figure~\ref{figure:thinnedVSlambda}, the percentage of diverted vehicles is plotted against arrival process intensity.  These quantities are each per lane, \ie, arrival process intensity in each lane, and percentage of vehicles per lane.
We see that until approximately 2.15 vehicles per second, virtually no thinning occurs. These results were obtained by running our proposed algorithm for 50,000 seconds of simulation time, and comparing the number of thinned vehicles per lane to the total number of vehicles that approached that lane.
Note that this simulation uses the shortest possible control region length $L^*$, as determined by Assumption~\ref{assumption:control_region_length}. For longer control regions, the onset is pushed even closer to the point of instability. 
Next, we consider how often diversion occurs as a function of control region length. In Figure~\ref{figure:thinnedL}, we plot the log of the intensity of thinned vehicles as a function of control region length.
We fix an arrival rate of 2.45 vehicles per second in each lane (while instability limit is at 2.5 vehicles per second). Each data point represents 50,000 to 100,000 seconds of simulation time, depending on how quickly the intensity converged.
The log of the diversion intensity and road length seem to have a linear correspondence.
Thus, the number of diversions decreases exponentially with increasing road length.
%

\begin{figure}[t]
	\includegraphics[clip=true,trim=.22in .04in .51in .26in,width=0.5\textwidth]{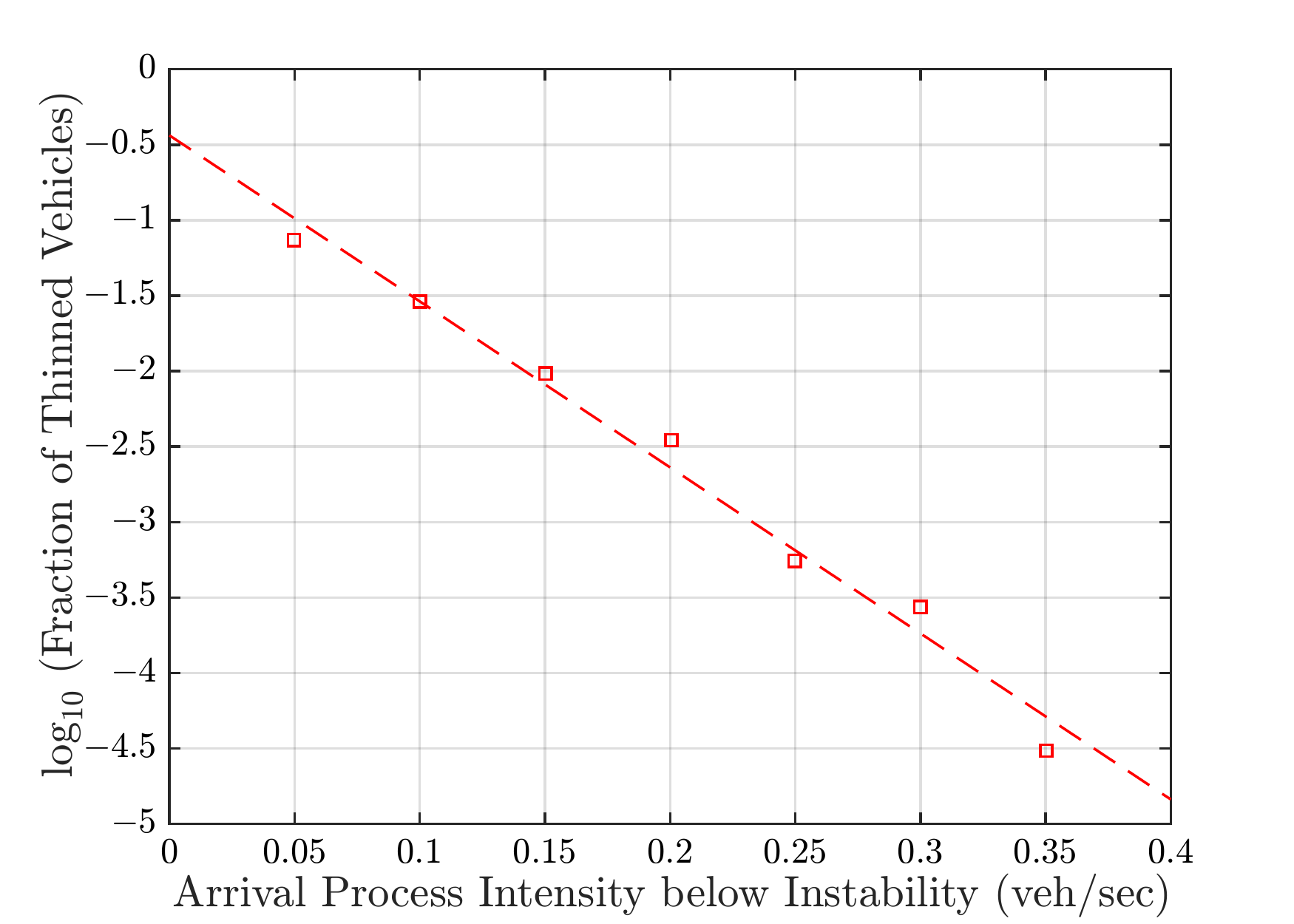}
	\caption{Logarithm of the fraction of thinned vehicles is plotted as a function of arrival process intensity. Algorithm~\ref{algorithm:main} is simulated with exhaustive polling policy. The length of the control region is set to the minimum control region length $L^*$ (See Assumption~\ref{assumption:control_region_length}) to guarantee safety of Algorithm~\ref{algorithm:main}. } \label{figure:thinnedVSlambda}
\end{figure}

\subsection{Comparison with a traditional traffic light}
\label{section:trafficLight}

In this section, we compare the performance of the proposed algorithm with a similar traffic light scenario.

Consider two lanes of vehicles approaching an intersection of the same dimensions as the system described in Section~\ref{section:problem}. Geometries and dynamics of the vehicles are the same as in the earlier problem definition. 
The positions of the vehicles are populated in the real line according to a Mat\'ern process, and all vehicles start with the maximum velocity $v_m$.
The traffic light cycles through three phases for each lane in the following order: green, yellow, red, yellow, and so on.

In the green phase, vehicles quickly accelerate to maximum speed, attempting to pass through the intersection.
During the red phase, vehicles decelerate quickly to a full stop, approaching as close to the intersection as safely possible. 
During the yellow phase, vehicles do either one of two things.
If the vehicle can come to a full stop before entering the intersection, the vehicle decelerates to a stop as close to the intersection as safely possible; otherwise, the vehicle accelerates quickly and passes through the intersection. This implies that a yellow phase which follows a red phase is simply an extension of the previous red phase.
Also, the traffic cycle is staggered so that while one lane is in its green phase, the other is in its red phase, and {\em vice versa}. Although both lanes are in the yellow phase simultaneously, due to the staggering of the green and red phases, only one lane permits its vehicles to traverse the intersection, provided that they cannot safely come to a stop before entering the intersection region.

At each time step, an aggressive control action is determined locally for each vehicle, inspired by the following assumptions.
To keep in line with the traffic light scenario, we assume that in determining a vehicle's control action, the vehicle can see the the phase of its lane, no matter how far away the vehicle may be from the intersection; however, the vehicle does not know how much longer the current phase will last.
Also, we assume that the a vehicle has perfect knowledge of the instantaneous velocity of the vehicle directly in front and also of the distance to the vehicle directly in front.
Although the vehicle does not know the future trajectory of the vehicle in front, the vehicle does know what the instantaneous acceleration of the vehicle in front is. 
Although knowledge of the instantaneous acceleration of the vehicle in front is not realistic, this assumption actually improves the performance of the traffic light, making a more competitive scenario with which we benchmark our algorithm.
\begin{figure}[t]
	\includegraphics[clip=true,trim=0.22in .05in .51in .26in,width=0.5\textwidth]{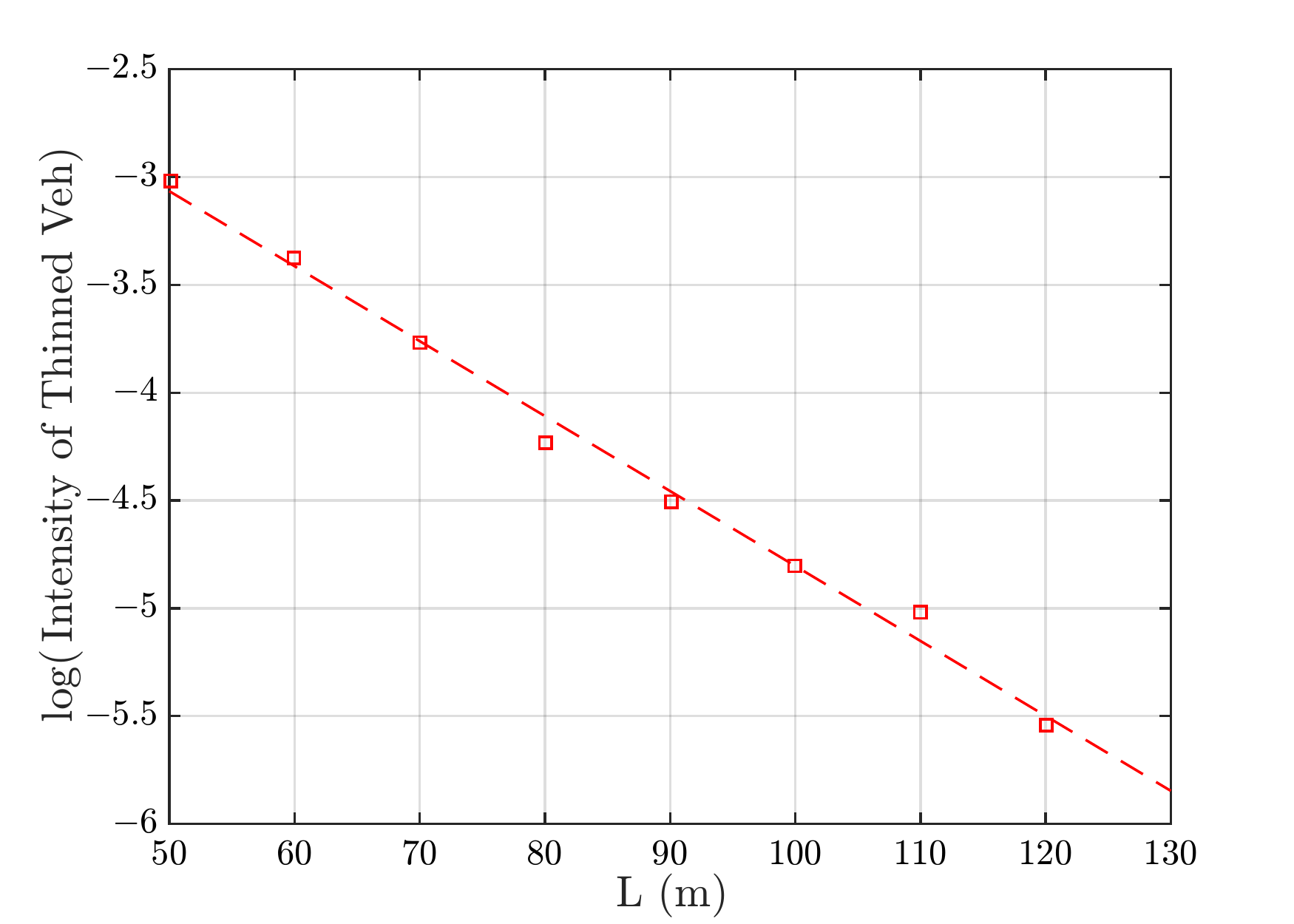}
	\caption{Logarithm of the intensity of thinned vehicles is plotted as a function of control region length $L$. Algorithm~\ref{algorithm:main} was run with exhaustive polling policy and arrival intensity of 2.45 vehicles per second (close to 2.5 vehicles per second instability limit) in each lane according to Mat\'ern process.}\label{figure:thinnedL}
\end{figure}

With these assumptions in mind, we design the most aggressive control action for each vehicle such that safety with the vehicle in front is ensured at the next time step.
A vehicle decides its acceleration as to position itself as close as possible to the vehicle directly in front, while also preserving safety. The manner in which we do this is by finding, at each time step, the maximum acceleration over the next time interval that will guarantee the following condition is satisfied at the next time step:
$
v_{i}^2/(2a_m) \leq d - l + v_{i-1}^2/(2a_m),
$
where $d$ is the distance between the bumpers of the two vehicles, and $v_i$ and $v_{i-1}$ are the velocities of the vehicle and the vehicle directly in front at the next time step, respectively.
The vehicle has knowledge of the right-hand side of the inequality, \ie, the distance that the vehicle in front will travel if it slams on the breaks beginning at the next time step. The left-hand side is the distance that the current vehicle will travel if it slams on the breaks beginning at the next time step. We choose maximum $a$ such that $v_i := v_{current} + a\Delta t$ will satisfy the above inequality.

We define delay for a vehicle approaching the traffic light similarly to our definition of delay earlier. (See Definition~\ref{definition:delay}.)
Delay is the additional time required for a vehicle to traverse the intersection due to the presence of other vehicles on the road and the traffic light, \ie, the difference in time between a vehicle fully traversing through the system and a vehicle continuing at maximum velocity through the system.

The traffic light scenario described provides a lower bound on the average delay of a traditional traffic light intersection.

The results of the comparison study are provided in Figure~\ref{figure:trafficComp}. Algorithm 1 with the exhaustive policy outperforms this traffic light scenario by at least two orders of magnitude in terms of delay.
\begin{figure}
\centerline{\includegraphics[clip=true,trim=0.22in .05in .51in .26in,width=0.5\textwidth]{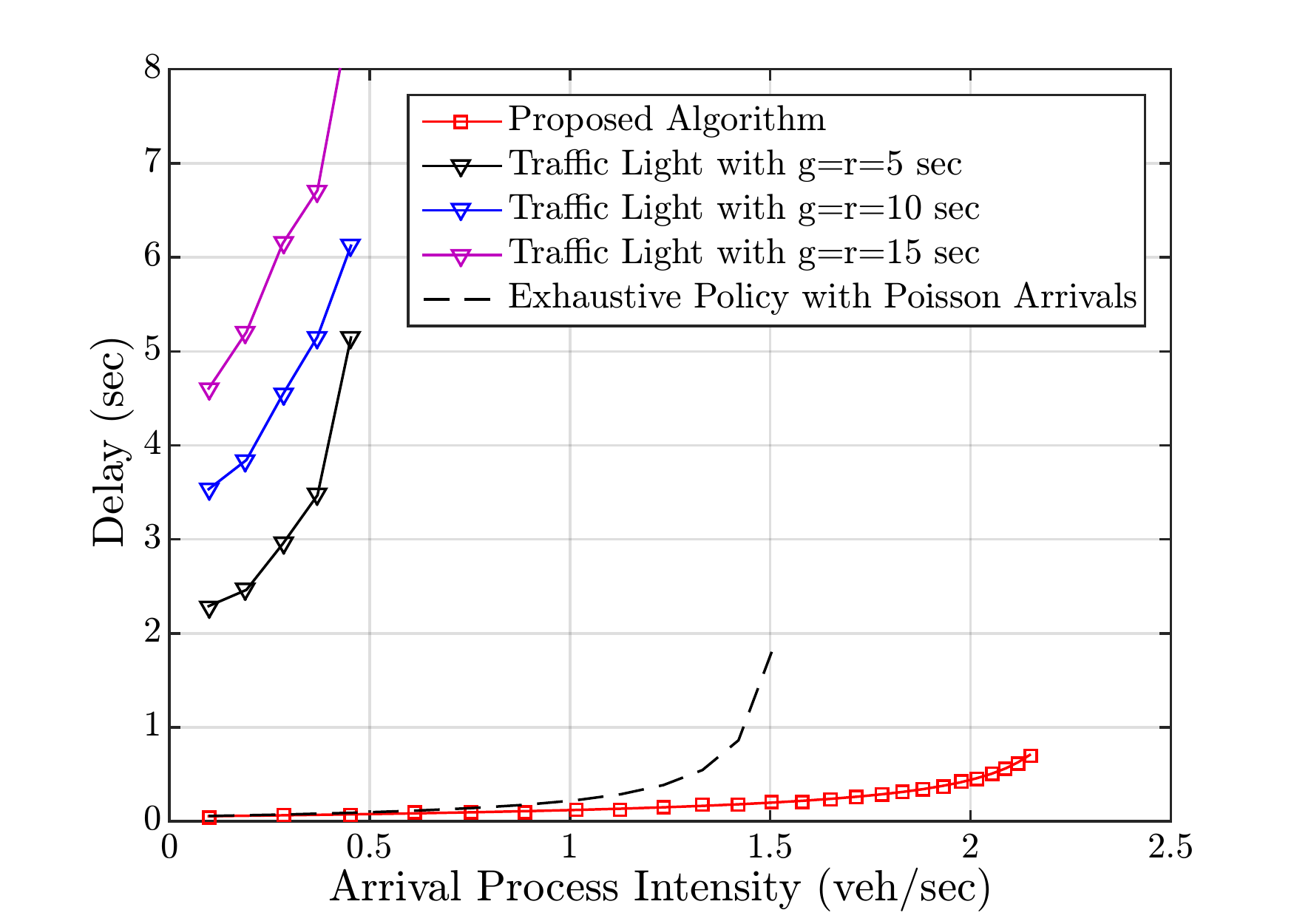}}
\caption{Delay of traffic intersection under Algorithm~\ref{algorithm:main} with exhaustive policy compared with delay of a traditional traffic intersection with equal green and red light phases of 5 seconds, 10 seconds, and 15 seconds. Expected wait time of polling system with exhaustive policy and Poisson arrivals is also plotted.} \label{figure:trafficComp}
\end{figure}
Beyond an arrival rate of $0.5$ vehicles per second in each lane, the delay of the traffic light scenario increases drastically, even though the point of instability of an intersection of these two flows is 2.5 vehicles per second. 

The durations of the red and green light phases are equal to each other, while the duration of the yellow light phase is always a constant, depending only on the geometry and the dynamics of the vehicles.
The duration of the yellow light phase is the minimum amount of time to guarantee safety between the two lanes of intersecting traffic, \ie, suppose a vehicle sees a red light and cannot come to a stop before the intersection region, then it must be allowed to safely traverse the intersection region.
For example, suppose a vehicle is traveling at full speed $v_m$ at the instant the traffic light switches from the green phase to yellow phase. The distance required for the vehicle to come to a full stop is $v_m^2/(2a_m)$.
However, suppose the remaining road length in front of the vehicle is slightly less than this quantity. Then, the vehicle will continue with full speed in order to traverse the intersection region. The time it takes to exit the intersection region is slightly less than $v_m/(2a_m) + (l+w)/v_m$.
Thus a necessary condition for safety is that the duration of the yellow phase must be at least as long. 
Furthermore, one can also show that for any other initial configuration of position and velocity at the end of the green phase, the time it takes for such a problematic vehicle to clear the intersection is no greater than $v_m/(2a_m) + (l+w)/v_m$. 
We set the duration of the yellow phase to this quantity.


\section{CONCLUSIONS} \label{section:conclusions}
In this paper, we considered the problem of coordinating the motion of vehicles through an intersection with no traffic lights. 
We proposed a coordination algorithm that provides provable guarantees on both safety and performance in all-autonomous traffic intersections. The proposed algorithm, at its core, schedules vehicles to use the intersection region according to a polling policy, which can be selected from a wide variety of policies for traditional polling systems. 
%
Provable performance bounds were established for the average delay of the proposed algorithm by considering its corresponding polling policy.
It was also shown that the proposed algorithm is computationally efficient. 
In simulation studies, the proposed system was compared to a traditional intersection system with a red-yellow-green traffic light; it was shown that the proposed system achieves delays that are at least two orders of magnitude smaller when compared to the traditional system.

\bibliographystyle{IEEEtran}
\bibliography{bib1,bib2,bib3}

\clearpage

\section*{Appendix}
\addcontentsline{toc}{section}{Appendices}
\renewcommand{\thesubsection}{\Alph{subsection}}
We devote this appendix to the proof of Lemma~\ref{lemma:main}. 
In what follows, we present a number of intermediate results leading to the proof of Lemma~\ref{lemma:main}. Then, we provide proofs of these intermediate results separately.
We restate the lemma here for convenience.

\medskip
\noindent {\bf Lemma~\ref{lemma:main}} 
{\em Suppose Assumptions \ref{assumption:regular_policy} and \ref{assumption:control_region_length} hold. Then, each time a new vehicle arrives and Algorithm~\ref{algorithm:main} is called, every call to the ${\tt MotionSynthesize}$ procedure (Line~\ref{line:motion_synthesize} of Algorithm~\ref{algorithm:main}) yields a feasible optimization problem.}
\medskip

First, we roughly outline the result that we are after.
Suppose a new vehicle, which we denote vehicle $A$, enters the control region at time $t_0'$. Without loss of generality, let us assume vehicle $A$ enters lane 2.
Since Algorithm~\ref{algorithm:main} is called (\ie, the vehicle is not removed before entering the control region), there exists a safe trajectory for vehicle $A$.
Every call to the ${\tt MotionSynthesize}$ procedure before time $t_0'$ yielded a feasible optimization problem in the same way.
Now to prove the lemma, we show that at time $t_0'$ each call to the ${\tt MotionSynthesize}$ procedure yields a feasible solution.

To make this precise, we introduce the following notation.
We denote $t_0$ as the latest time that Algorithm~\ref{algorithm:main} was called before the arrival of vehicle $A$, \ie, $t_0 = \max \{ t < t_0' \colon \text{Algorithm~\ref{algorithm:main} is called at time }t \}$.
Recall that the $i$th vehicle in lane $k$ is denoted as vehicle $(i,k)$.
We denote $\tau_{i,k}$ as the schedule time of vehicle $(i,k)$ assigned at time $t_0$.
Also, we denote $x_{i,k}$ as the trajectory of vehicle $(i,k)$ assigned by the ${\tt MotionSynthesize}$ procedure at time $t_0$.
Similarly, we denote $\tau_{i,k}'$ and $x_{i,k}'$ as the schedule time and trajectory assigned at time $t_0'$.
Lastly, we denote $\tau_A$ as the schedule time of newly arrived vehicle $A$ assigned at time $t_0'$.
See Table~\ref{table:notation} for a summary of this notation.

\begin{table}[ht]
\centering
\begin{tabular}{c c c}
\hline\hline
\noalign{\vskip 3pt} 
call to Algorithm~\ref{algorithm:main} 	&	$t_0$	&	$t_0'$ \\ [1ex]
\hline
\noalign{\vskip 3pt} 
schedule time	& 	$\tau_{i,k}$ 	& $\tau_{i,k}'$ \\ [1ex]
\hline
\noalign{\vskip 3pt} 
schedule time of vehicle $A$	& 	- 	& $\tau_A$ \\ [1ex]
\hline
\noalign{\vskip 3pt} 
assigned trajectory	& $x_{i,k}$	& $x_{i,k}'$ \\ [1ex]
\hline
\noalign{\vskip 3pt} 
scheduled to arrive \\ at intersection region	& \centering$\tau_{i,k}+L/v_m$	& $\tau_{i,k}'+L/v_m$ \\ [1ex]
\hline\hline
\vspace{1pt}
\end{tabular}
\caption{Notation pertaining to vehicles in the system at times $t_0$ and $t_0'$.}\label{table:notation}
\end{table}
\vspace{-1pc}
We use the boldface notation as follows. We denote a state $\mathbf{z} := (p,v)$. We denote a state history $z := (x,\dot{x})$.
Next, let $\mathscr{C}(\mathbf{z}_0,\tilde{t}_0,\tilde{t}_f,x')$ denote the set of trajectories with initial state $\mathbf{z}_0 := (p_0,v_0)$ at time $\tilde{t}_0$ that reach the intersection region with full velocity at time $\tilde{t}_f$ and are safe with $x'$.
See Figure~\ref{figure:Ci}.
More precisely, we define the {\em set of feasible trajectories} as
\begin{align*}
\mathscr{C}(\mathbf{z}_0,\tilde{t}_0,\tilde{t}_f,&x') := \big\{ x \colon [\tilde{t}_0,\tilde{t}_f] \to \reals \,\, \big\vert \,\,\, \exists u \in \mathcal{U}, \\
& \ddot{x}(t) = u(t),\, |u(t)| \leq a_m,\, \forall t \in [\tilde{t}_0,\tilde{t}_f]; \\
& \dot{x}(t) \in [0,v_m],\, \forall t \in [\tilde{t}_0,\tilde{t}_f];\\
& z(\tilde{t}_0) = \mathbf{z}_0;\, z(\tilde{t}_f) = (0,v_m); \\
& |x(t) | \geq l + |x'(t)|,\, \forall t \in [\tilde{t}_0,\tilde{t}_f] \cap \mathscr{D}(x')\big\},
\end{align*}
where $u \in \mathcal{U}$ is a measurable function, $u \colon [\tilde{t}_0,\tilde{t}_f] \rightarrow \reals$, and $\mathscr{D}(x')$ is the domain of trajectory $x'$.
Also, we write $\mathscr{C}_{i,k}'$ for $\mathscr{C}(z_{i,k}(t_0'),t_0',\tau_{i,k}'+L/v_m,x_{i-1,k}')$.
\begin{figure}[t]
\centerline{\includegraphics[clip=true, trim=2in 6.6in 4.2in 3.3in, width=0.5\textwidth]{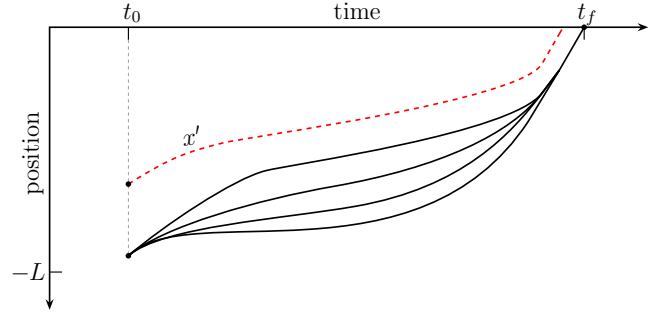}}
\caption{Sample trajectories from set $\mathscr{C}(\mathbf{z}_0,t_0,t_f,x')$ are shown in solid black. Trajectory $x'$ is shown as the dashed red line. Note that all sample trajectories have the same initial and terminal states, same terminal time, and are safe with trajectory $x'$.} \label{figure:Ci}
\end{figure}
The ${\tt MotionSynthesize}$ procedure at time $t_0'$ for vehicle $(i,k)$ is simply
\begin{align} \label{eq:procedure}
	x_{i,k}' \leftarrow \min_{x \in \mathscr{C}_{i,k}'} \int_{t_0'}^{\tau_{i,k}'+L/v_m} |x(t)| dt.
\end{align}
Hence, Lemma~\ref{lemma:main} basically states that the optimization problem in Equation~\eqref{eq:procedure} has a feasible solution each time the ${\tt MotionSynthesize}$ procedure is called. 
The following lemma restates Lemma~\ref{lemma:main} using the new notation.

\begin{lemma} \label{lemma:main2} 
Suppose Assumptions \ref{assumption:regular_policy} and \ref{assumption:control_region_length} hold. Suppose a new vehicle arrives at time $t_0'$ and Algorithm~\ref{algorithm:main} is called. Then, for every vehicle $(i,k)$ in the system at time $t_0'$, the optimization problem in Equation~\eqref{eq:procedure} has a feasible solution.
\end{lemma}

To prove Lemma~\ref{lemma:main2} (or equivalently Lemma~\ref{lemma:main}), we consider two cases separately: {\em (i)} vehicles that are scheduled to arrive at the intersection before vehicle $A$, and {\em (ii)} vehicles that are scheduled to arrive at the intersection after vehicle $A$. 
We show that the optimization problem in Equation~\eqref{eq:procedure} yields a feasible solution for all these vehicles. 

In Section~\ref{section:setup} of this appendix, we focus on the first case, and we prove the following lemma.
\begin{lemma} \label{lemma:committed}
Suppose that at time $t_0'$ the schedule of vehicle $(i,k)$ is earlier than that of vehicle $A$, \ie, $\tau_{i,k}' < \tau_A$.
Then, $\mathscr{C}_{i,k}'$ is non-empty and the minimum in Equation~\eqref{eq:procedure} is attained, for vehicle $(i,k)$. Moreover, the trajectory $x_{i,k}'$ is simply a truncation of $x_{i,k}$, \ie, $x_{i,k}' = x_{i,k}|_{\{ t \colon t\geq t_0' \} }$.
\end{lemma}
This lemma states that vehicles scheduled to arrive at the intersection before vehicle $A$ simply continue on their previous trajectories. That is, the call to the ${\tt MotionSynthesize}$ procedure at time $t_0'$ (triggered by the arrival of vehicle $A$) does not change the trajectories of this set of vehicles; for these vehicles, the ${\tt MotionSynthesize}$ procedure returns the same set of trajectories at time $t_0'$ it returned in the previous call to the procedure at time $t_0$.
In Section~\ref{section:setup} of this appendix, we prove this result by induction. We use Bellman's optimality principle to show that the previous trajectory is still optimal.

Next, we verify Equation~\eqref{eq:procedure} for vehicles scheduled later than vehicle $A$.
We prove the following lemma.
\begin{lemma} \label{lemma:uncommittedOutline}
Suppose that at time $t_0'$ the schedule of vehicle $(i,k)$ is later than that of vehicle $A$, \ie, $\tau_{i,k}' > \tau_A$. 
Then, $\mathscr{C}_{i,k}'$ is non-empty and the minimum is attained in Equation~\eqref{eq:procedure}.
\end{lemma}
To formalize our argument for this lemma, we introduce the following notation. 
Firstly, since all regular polling policies (See Assumption~\ref{assumption:regular_policy} for a definition) are {\it first-come first-serve} (FCFS), all vehicles scheduled after vehicle $A$ must be in lane 1, \ie, not in the lane vehicle $A$ entered.
Thus, for notational simplicity, for the remainder of the proof, we safely drop the $k$ subscript, \ie, denote $\tau_{i,k}'$ by $\tau_i'$, and $x_{i,k}$ by $x_i$, {\em et cetera}.

Secondly, for any vehicle $i$ scheduled after vehicle $A$, we define $\nu_i := i - \min  \{\,  i \, | \,\, \tau_{i,1}' > \tau_A \}  + 1$.
Note that this vehicle $i$ is the $\nu_i^\mathrm{th}$ vehicle scheduled after vehicle $A$.
See Figure~\ref{figure:inter3}.
\begin{figure}[b]
\centerline{\includegraphics[clip=true, trim=3in 8in 0.5in 2in, width=0.5\textwidth]{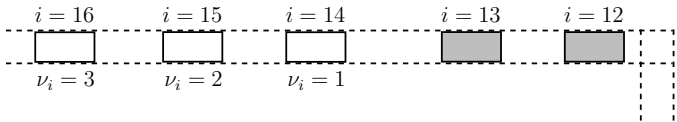}}
\caption{Lane 1 is shown in full. Vehicles scheduled before vehicle $A$ are shaded grey. Vehicles scheduled after vehicle $A$ have white interiors. For each vehicle scheduled after vehicle $A$, its value of $\nu_i$ is displayed directly below.} \label{figure:inter3}
\end{figure}
For each vehicle $i$ scheduled after vehicle $A$, we define:
\begin{align*}
	\mathcal{F}_i &:= \{\, (p_0,v_0) \in [-L,0] \times [0,v_m]\, | \\
			&p_0 + v_0^2/(2a_m) \leq - v_m^2/(2a_m) - (\nu_i-1)l \,\}.
\end{align*}
See Figure~\ref{figure:Fi}.
Furthermore, we denote the right-hand side boundary of the set $\mathcal{F}_i$ by $\partial \mathcal{F}_i$, \ie, 
\begin{align*}
	\partial \mathcal{F}_i &:= \{\, (p_0,v_0) \in [-L,0] \times [0,v_m] \, | \\
	&p_0 + v_0^2/(2a_m) =  - v_m^2/(2a_m) - (\nu_i-1)l \,\}.
\end{align*}

Let us provide some insight into sets $\partial{\cal F}_i$ and ${\cal F}_i$. 
Notice that following the boundary $\partial {\cal F}_i$, vehicle $i$ can decelerate as rapidly as possible (with deceleration $a_m$) and come to a full stop exactly $v_m^2/(2 a_m) + (\nu_i -1) l$ further away from the intersection region. Suppose $\nu_i = 1$, \ie, vehicle $i$ is the first vehicle to go through the intersection after vehicle $A$. Then, following $\partial {\cal F}_i$, vehicle $i$ can come to a full stop exactly $v_m^2/(2 a_m)$ away from the intersection region. Notice that $v_m^2/(2 a_m)$ is exactly the distance it takes vehicle $i$ to accelerate with maximum acceleration ($a_m$) and reach maximum speed ($v_m$) right when it arrives at the intersection region. When $\nu_i > 1$, vehicle $i$ can follow $\partial {\cal F}_i$ in the same manner and stop at a point that has an additional distance of $\nu_i-1$ vehicle lengths (\ie, $(\nu_i-1)l$) to the intersection region.
Hence, $k$ vehicles with $\nu_i \in \{1, 2, \dots, k\}$, by following their respective $\partial {\cal F}_i$, come to a full stop, one behind another, such that the foremost vehicle has distance $v_m^2/(2a_m)$ from the intersection region.
These vehicles can move together in this formation with maximum acceleration ($a_m$) to reach the intersection at maximum speed ($v_m$). 

The set ${\cal F}_i$ is critical in proving Lemma~\ref{lemma:uncommittedOutline}, and it provides us with insights into how the algorithm guarantees safety.

\begin{figure}[t]
	\centerline{\includegraphics[clip=true, trim=2.15in 6.5in 4.1in 3.4in, width=0.5\textwidth]{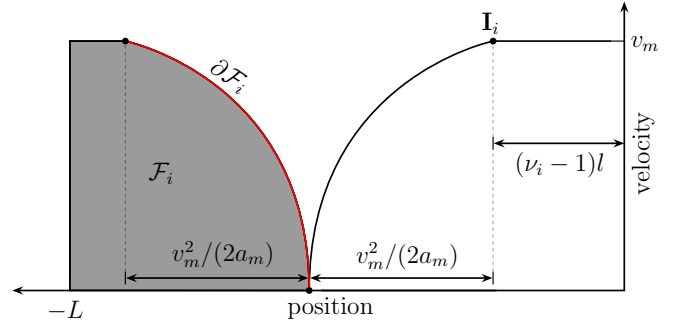}}
\caption{Sets $\mathcal{F}_i$ and $\partial \mathcal{F}_i$ are shown. Set $\mathcal{F}_i$ is shaded in grey. Set $\partial\mathcal{F}_i$, the right-hand side boundary of $\mathcal{F}_i$, is traced in red. State ${\bf I}_i$ has coordinates $(-(\nu_i-l)l,v_m)$. Set $\mathcal{F}_i$ can be described as the set of all states from which state ${\bf I}_i$ can be reached at a time indefinitely far into the future. } \label{figure:Fi}
\end{figure}

We prove Lemma~\ref{lemma:uncommittedOutline} in three steps.
First, we show that the state $z_i(t_0')$ of the $\nu_i^\mathrm{th}$ vehicle after vehicle $A$ at time $t_0'$ is inside the set ${\cal F}_i$. See Figure~\ref{figure:Fi}.
Second, we show that each such vehicle has a {\em feasible} trajectory satisfying a particular property: the $\nu_i^\mathrm{th}$ vehicle continues as long as possible along its last assigned trajectory $x_i$. See Figure~\ref{figure:Ei}.
Finally, we show that the ${\tt MotionSynthesize}$ procedure actually returns such a trajectory.
This last step is more technical, requiring results from optimal control theory showing that the optimization problem in Equation~\eqref{eq:procedure} is well-posed.
These steps are formalized in Lemmas~\ref{lemma:uncommittedInFi}, \ref{lemma:EiNotEmpty}, and \ref{lemma:well-defined}, respectively. We state these lemmas below, and provide insights into their proofs.

First, we show that at time $t_0'$ the state of the $\nu_i^\mathrm{th}$ vehicle after vehicle $A$ is in ${\cal F}_i$. These vehicles are able to come to a full stop far enough away from the intersection region, such that they can attain full speed before crossing it. 
\begin{lemma} \label{lemma:uncommittedInFi}
For any vehicle $i$ scheduled after vehicle $A$ at time $t_0'$, \ie, $\tau_i' > \tau_A$,
we have $z_i(t_0') \in \mathcal{F}_i$.
\end{lemma}
Let us provide some intuition into the proof of this Lemma.
First, we define state $\mathbf{I}_i := (\, -(\nu_i-1) l , v_m)$. See Figure~\ref{figure:Fi}.
By the construction of the ${\tt MotionSynthesize}$ procedure, the $\nu_i^\mathrm{th}$ vehicle after vehicle $A$ arrives at state ${\bf I}_i$ and then continues at maximum speed to the intersection region.
(One can observe this phenomena in the heavy load case of Figure~\ref{figure:trajectories}. The reason is that the ${\tt MotionSynthesize}$ procedure, while constraining the terminal state and time, also requires safety.)
Due to the regular polling policy condition (Assumption~\ref{assumption:regular_policy}) and the minimum road length condition (Assumption~\ref{assumption:control_region_length}), it turns out that the time between the arrival of vehicle $A$ and the arrival of vehicle $i$ at state $\mathbf{I}_i$ (according to previous trajectory $x_i$) is greater than or equal to $2 v_m/ a_m$, \ie, 
\begin{align*}
\tau_i + L/v_m - (\nu_i-1)s  - t_0' \geq 2v_m / a_m.
\end{align*}
We then conclude by showing that the state $z_i$ of the $\nu_i^\mathrm{th}$ vehicle at time $t_0'$ must be in $\mathcal{F}_i$, \ie, $z_i(t_0') \in \mathcal{F}_i$, via a straightforward application of Pontryagin's Minimum Principle.
Note that if $z_i(t_0') \not\in \mathcal{F}_i$, then there can be only two outcomes: either vehicle $i$ is not able to reach state $I_i$, or vehicle $i$ can reach state $I_i$ but only in an amount of time strictly less than $2v_m/a_m$, clearly violating the above inequality.
Detailed proof is given in Section B of this appendix.

Second, we show that not only does there exist a {\it feasible} trajectory for the $\nu_i^\mathrm{th}$ vehicle after vehicle $A$, \ie, $\mathscr{C}_i' \not= \emptyset$, but we also show that vehicle $i$ continues along its last assigned trajectory $x_i$ as long as possible. See Figure~\ref{figure:Ei}.
We will write this result rigorously as $\mathscr{E}_i' \not= \emptyset$, and we now proceed to make this notion precise.
Let $\mathscr{E}(y,\tilde{t}_0,\tilde{t}_f,x')$ be the set of all extensions of trajectory $y$ starting at time $\tilde{t}_0$ that reach the intersection region at time $\tilde{t}_f$ and are safe with trajectory $x'$. More precisely, we define
\begin{align*}
\mathscr{E}(y,\tilde{t}_0,\tilde{t}_f,x') := \big\{ & x \in \mathscr{C}\big(\, \big(y(\tilde{t}_0),\dot{y}(\tilde{t}_0)\big),\tilde{t}_0,\tilde{t}_f,x'\big) \,\big\vert \\
		& x(t) = y(t), \forall t \in [\tilde{t}_0,\tilde{t}_f] \cap \mathscr{D}(y) \big\}.
\end{align*}
We denote $t_c[x_i]$ as the time that trajectory $x_i$ leaves the set $\partial\mathcal{F}_i$, \ie,
$ 
	t_c[x_i] := \sup \{t \,| \,\, z_i(t) \in \partial \mathcal{F}_i \}.
$ 
For shorthand, we write $\mathscr{E}_i'$ for $\mathscr{E}(x_i |_{[t_0,t_c[x_i]]},t_0',\tau_i' + L/v_m,x'_{i-1})$.
Note that $\mathscr{E}_i' \subseteq \mathscr{C}_i'$.
We now formalize our second intermediate result. 

\begin{figure}[t]
	\centerline{\includegraphics[clip=true, trim=2.15in 6.5in 4.1in 3.4in, width=0.5\textwidth]{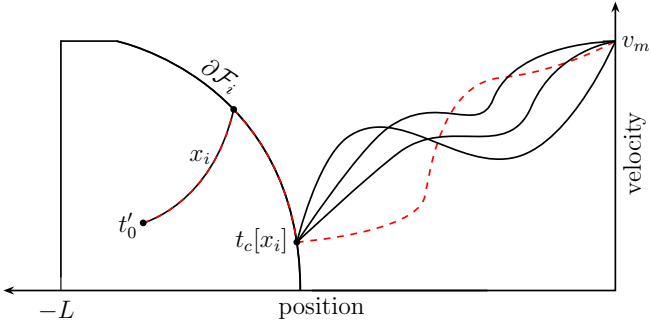}}
	\caption{Trajectory $x_i$ is depicted in position-velocity space as the dashed red line; it begins at time $t_0'$, exits set $\mathcal{F}_i$ at time $t_c[x_i]$, and continues to the state $(0,v_m)$. Sample trajectories from set $\mathscr{E}_i'$ follow trajectory also begin at time $t_0'$, replicate trajectory $x_i$ until time $t_c[x_i]$, and then choose any (feasible and safe) path to state $(0,v_m)$. Sample trajectories of $\mathscr{E}_i'$ are shown in black.} \label{figure:Ei}
\end{figure}

\begin{lemma} \label{lemma:EiNotEmpty}
Consider vehicle $i$ scheduled after vehicle $A$ at time $t_0'$, \ie, $\tau_i' > \tau_A$. Suppose $z_i(t_0') \in \mathcal{F}_i$ and $x_{i-1}' \in \mathscr{E}_{i-1}'$.
Then, $\mathscr{E}_i' \not= \emptyset$.
\end{lemma}
First, note that Lemma~\ref{lemma:EiNotEmpty} states that there exists a feasible trajectory for the $\nu_i^\mathrm{th}$ vehicle, \ie, $\mathscr{C}_i \not= \emptyset$.
Moreover, there exists a particular trajectory for the $\nu_i^\mathrm{th}$ vehicle that follows its last assigned trajectory $x_i$ at least until it leaves set $\mathcal{F}_i$, \ie, at time $t_c[x_i]$, as depicted in Figure~\ref{figure:Ei}.

The proof of Lemma~\ref{lemma:EiNotEmpty} is simply the construction of a trajectory in the set $\mathscr{E}_i'$, and it is given in Section C of this appendix.
This construction is quite elaborate in the general case. 
However, we first present a simple case that demonstrates the basic building blocks of the construction. In this simple case, we concern ourselves with the case $\nu_i = 1$. We construct a feasible trajectory in $\mathscr{E}_i'$ for this vehicle as follows: vehicle $i$ follows its last assigned trajectory $x_i$ until it reaches the boundary $\partial \mathcal{F}_i$, decelerates along $\partial \mathcal{F}_i$, and then finally accelerates to full speed before crossing the intersection region. The transition point from deceleration to acceleration is dependent solely on the terminal time $\tau_i' + L/v_m$ at which vehicle $i$ arrive at the intersection region.
We then argue that this trajectory for vehicle $i$ must be safe with the newly updated trajectory $x_{i-1}'$ of the vehicle directly in front.
One can see this from noting that the trajectory constructed above for vehicle $i$ is clearly safe with the last assigned trajectory $x_{i-1}$ for vehicle $i-1$, and then noting that $x_{i-1}$ and $x_{i-1}'$ are essentially identical, by the earlier stated Lemma~\ref{lemma:committed}.
The general case ($\nu_i > 1$) is more complicated in its construction due to the fact that the updated trajectory $x_{i-1}'$ of the vehicle directly in front changes; thus, additional care must be taken in constructing a feasible trajectory in $\mathscr{E}_i'$ for vehicle $i$ that is safe with the updated trajectory $x_{i-1}'$.

Now, we formalize our final intermediate result. 
\begin{lemma} \label{lemma:well-defined}
Suppose vehicle $i$ is scheduled after vehicle $A$ at time $t_0'$, and $\mathscr{E}_i' \not= \emptyset$.
Then, the ${\tt MotionSynthesize}$ procedure described in Equation~\eqref{eq:procedure} returns a solution from the set $\mathscr{E}_i' \subseteq \mathscr{C}_i'$.
\end{lemma}
This lemma states that if $\mathscr{E}_i'$ is non-empty, then the ${\tt MotionSynthesize}$ procedure admits a trajectory in $\mathscr{E}_i'$. 
This rather straightforward lemma is quite technical in nature, while adding no further insight. 
The details of the proof of Lemma~\ref{lemma:well-defined} is given in Section D of this appendix.

Now, let us give a brief overview.
From Lemmas~\ref{lemma:committed} and~\ref{lemma:uncommittedOutline}, the proof of Lemma~\ref{lemma:main} follows.
Recall that Lemma~\ref{lemma:committed} verifies the feasibility of the optimization problem in Equation~\eqref{eq:procedure} for all vehicles with updated schedules {\it earlier} than vehicle $A$, \ie, $\tau_{i,k}' < \tau_A$. 
Similarly, Lemma~\ref{lemma:uncommittedOutline} shows the feasibility of Equation~\eqref{eq:procedure} for all vehicles with updated schedules {\it later} than vehicle $A$, \ie, $\tau_{i,k}' > \tau_A$.
We prove this lemma by induction, using our three intermediate results: Lemmas~\ref{lemma:uncommittedInFi},~\ref{lemma:EiNotEmpty}, and~\ref{lemma:well-defined}.

Consider the first vehicle with updated schedule after the arrival of vehicle $A$, \ie, $\nu_i = 1$.
By Lemmas~\ref{lemma:uncommittedInFi} and~\ref{lemma:EiNotEmpty}, $\mathscr{E}_i'$ is non-empty.
Next, by Lemma~\ref{lemma:well-defined}, trajectory $x_i'$, as defined by Equation~\eqref{eq:procedure}, exists {\it and} is contained in $\mathscr{E}_i'$.
Now, consider the case $\nu_i > 1$, and assume (by induction) that trajectory $x_{i-1}'$, as defined by Equation~\eqref{eq:procedure}, exists and is contained in $\mathscr{E}_{i-1}'$.
Then, by Lemmas~\ref{lemma:uncommittedInFi} and~\ref{lemma:EiNotEmpty}, the set $\mathscr{E}_i'$ is non-empty.
By Lemma~\ref{lemma:well-defined}, the ${\tt MotionSynthesize}$ procedure admits a solution $x_i'$ contained in $\mathscr{E}_i'$.
Thus, we have shown, for any vehicle $i$ scheduled after vehicle $A$, \ie, $\tau_i' > \tau_A$, the ${\tt MotionSynthesize}$ procedure always admits a trajectory in $\mathscr{E}_i'$. Lemma~\ref{lemma:uncommittedOutline} follows directly from noting $\mathscr{C}_i' \subseteq \mathscr{E}_i'$.
%

To be completely rigorous, one has to verify that the optimization problem in Equation~\eqref{eq:procedure} has a feasible solution for newly arrived vehicle $A$.
However, this is not difficult, since vehicle $A$ in lane 2 follows a vehicle whose updated trajectory is simply its last assigned trajectory.
Therefore, by using almost identical arguments as for vehicle $i$ in lane 1 satisfying $\nu_i = 1$, one can show that vehicle $A$ also has a feasible trajectory with terminal time $\tau_A + L/v_m$. 

\newpage
\subsection{Proof of Lemma~\ref{lemma:committed}} \label{section:setup}

Recall that $\tau_{i,k}$ is the schedule time for vehicle $i$ in lane $k$; hence, the same vehicle was scheduled to arrive at the intersection region at time $\tau_{i,k}+L/v_m$, when the algorithm was run at time $t_0$. Recall that the schedule time for the same vehicle is $\tau_{i,k}'$, when the algorithm is run again at time $t_0'$. 
Consider vehicles scheduled before vehicle $A$, according to the updated schedule times.
%
First, note that by the regular polling policy condition (See Assumption~\ref{assumption:regular_policy}), any vehicle $(i,k)$ scheduled before the newly arrived vehicle $A$ does not change its schedule time, \ie, $\tau_{i,k}' < \tau_A$ implies $\tau_{i,k}' = \tau_{i,k}$.
For shorthand, we denote $\tau_{i,k}+L/v_m$ by $t_f$.
Note that for all such vehicles, we have $\tau_{i,k}+L/v_m = \tau_{i,k}'+L/v_m = t_f$.
Now, we prove (by induction) that $x_{i,k}' = x_{i,k}|_{ [ t_0', t_f] }$.

Out of all the vehicles scheduled before vehicle $A$, let us consider the first vehicle in lane $k$ at time $t_0'$, \ie, $i = \min I_k(t_0')$.
We denote ${\tt MS}$ as short-hand for ${\tt MotionSynthesize}$ in this section. Note that trajectory $x_{i,k}$ satisfies
\begin{align*}
	x_{i,k} &=  {\tt MS}(z_{i,k}(t_0),t_0,t_f,x_{i-1,k}).
\end{align*}
Also note that by definition, we have
\begin{align*}
	x_{i,k}' &:=  {\tt MS}(z_{i,k}(t_0'),t_0',t_f,\emptyset).
\end{align*}
See Figure~\ref{figure:sectionA}. By Bellman's Principle of Optimality,
\begin{align*}
	x_{i,k}|_{[t_0',t_f]} &= {\tt MS}(z_{i,k}(t_0'),t_0',t_f,x_{i-1,k}) \\
				&= {\tt MS}(z_{i,k}(t_0'),t_0',t_f,x_{i-1,k}|_{[t_0',t_f]}) \\
				&= {\tt MS}(z_{i,k}(t_0'),t_0',t_f,\emptyset).
\end{align*}
Therefore, we have $x_{i,k}' = x_{i,k}|_{[t_0',t_f]}$.
Now, we consider any vehicle $(i,k)$ that is scheduled before vehicle $A$, \ie, $\tau_{i,k}' < \tau_A$.
By induction, we assume vehicle $(i-1,k)$ continues on its previous trajectory, \ie, $x_{i-1,k}' = x_{i-1,k}|_{ \{ t \colon t\geq t_0' \} }$.
We claim that vehicle $(i,k)$ also continues on its previous trajectory, \ie, $x_{i,k}' = x_{i,k}|_{ [ t_0' ,t_f]}$.
Recall
\begin{align*}
	x_{i,k}' := {\tt MS}(z_{i,k}(t_0'),t_0',t_f,x_{i-1,k}').
\end{align*}
Note that
\begin{align*}
	x_{i,k}|_{[t_0',t_f]} &= {\tt MS}(z_{i,k}(t_0'),t_0',t_f,x_{i-1,k}) \\
					&= {\tt MS}(z_{i,k}(t_0'),t_0',t_f,x_{i-1,k}|_{[t_0',t_f]}) \\
					&= {\tt MS}(z_{i,k}(t_0'),t_0',t_f,x_{i-1,k}').
\end{align*}
Hence, $x_{i,k}' = x_{i,k}|_{ [ t_0' ,t_f] }$.
The proof of this lemma is complete.

\begin{figure}
	\centerline{\includegraphics[clip=true, trim=2in 6.6in 4.2in 3.3in, width=0.5\textwidth]{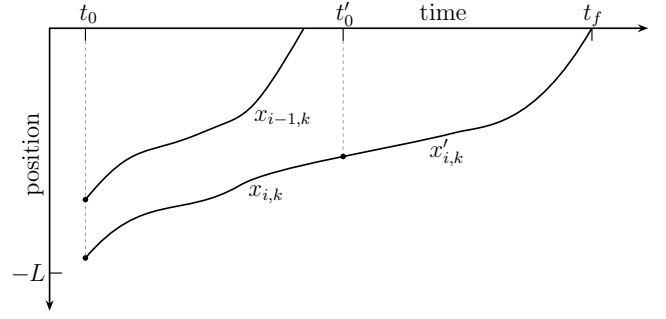}}
	\caption{Trajectories $x_{i,k}$ and $x_{i-1,k}$ that had been previously generated at time $t_0$. Newly generated trajectory $x_{i,k}'$ is also shown, assuming that vehicle $(i,k)$ is the first vehicle in lane $k$ at time $t_0'$. Note that in this scenario, updated trajectory $x_{i,k}'$ is simply a restriction of previous trajectory $x_{i,k}$.} \label{figure:sectionA}
\end{figure}

\newpage
\subsection{Proof of Lemma~\ref{lemma:uncommittedInFi}} \label{section:B}
In this section, we show that if at time $t_0'$ vehicle $i$ is scheduled to arrive at the intersection region after vehicle $A$, then the state of vehicle $i$ is currently contained in $\mathcal{F}_i$, \ie, $z_i(t_0') \in \mathcal{F}_i$.
In summary, we argue as follows.
Recall that for any vehicle $i$ scheduled to arrive at the intersection region after vehicle $A$, we defined $\nu_i$ such that vehicle $i$ is the $\nu_i$th vehicle after vehicle $A$. 
Let us define state $\mathbf{I}_i$ as $\mathbf{I}_i := (\,-(\nu_i-1) l , v_m)$. See Figure~\ref{figure:Fi}.
By the construction of the ${\tt MotionSynthesize}$ procedure, any such vehicle $i$ arrives at state ${\bf I}_i$ and then continues at maximum speed to the intersection region.
Due to the regular polling policy condition (Assumption~\ref{assumption:regular_policy}) and the minimum road length condition (Assumption~\ref{assumption:control_region_length}), it turns out that the time between the arrival of vehicle $A$ and the arrival of vehicle $i$ at state $\mathbf{I}_i$ (according to previous trajectory $x_i$) is greater than or equal to $2 v_m/ a_m$.
We then conclude by showing that $z_i(t_0') \in \mathcal{F}_i$ holds via an optimal control argument.

Now, we make the previous arguments precise.
First, we show that vehicle $i$ arrives at state $\mathbf{I}_i$ at time $\tau_i + L/v_m - (\nu_i-1)s$: 
\begin{proposition} \label{proposition:intersection}
Suppose vehicle $i$ is scheduled after vehicle $A$ at time $t_0'$, \ie, $\tau_i' < \tau_A$.  Let $x_i$ be the previous trajectory of vehicle $i$, let $z_i$ be the state history, \ie, $z_i = (x_i,\dot{x}_i)$, and let $\tau_i$ be its previous schedule time.
Then, $z_i(\tau_i+L/v_m-(\nu_i-1)s) = (\, -(\nu_i-1)l,v_m) := \mathbf{I}_i$.
\end{proposition}
\begin{figure}
	\centerline{\includegraphics[clip=true, trim=2in 6.6in 4.2in 3.3in, width=0.5\textwidth]{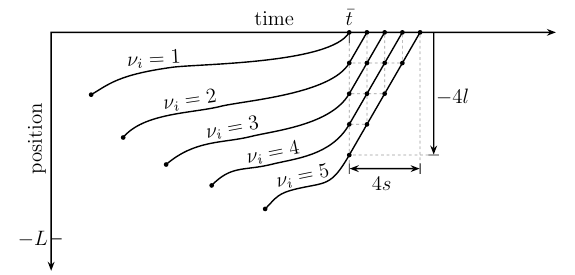}}
	\caption{Sample trajectories for the first five vehicles scheduled after newly arrived vehicle $A$. Time $\bar{t} := \tau_i + L/v_m - (\nu_i-1)s$ is exactly the terminal time of the first vehicle scheduled after vehicle $A$, \ie, $\nu_i = 1$. Note how these vehicles move as a platoon, as the first vehicle begins to cross the intersection region.} \label{figure:sectionBterminalTime}
\end{figure}

\begin{proof}
See Figure~\ref{figure:sectionBterminalTime}. Essentially, this proposition states that all vehicles scheduled after vehicle $A$ move as a platoon when the first vehicle ($\nu_i = 1$) begins to cross the intersection region.
This phenomenon can be observed clearly when the algorithm is run in the heavy load regime. See Figure~\ref{figure:trajectories}.
We prove this by induction on $\nu_i$.
First, consider the case $\nu_i = 1$. Then, $z_{i}(\tau_i + L/v_m) = (0,v_m)$ immediately holds from the construction of the ${\tt MotionSynthesize}$ procedure.
Next, we consider the more general case $\nu_i > 1$.
We assume $z_{i-1}(\tau_{i-1}+L/v_m-(\nu_{i-1}-1)s) = ( \,-(\nu_{i-1}-1)l,v_m)$.
Denote 
\begin{align*}
	\bar{t} := \tau_{i-1}+L/v_m-(\nu_{i-1}-1)s = \tau_i + L/v_m - (\nu_i-1)s.
\end{align*}
By safety of previous trajectories $x_i$ and $x_{i-1}$, we have 
\begin{align*}
	x_i(\bar{t}) \leq x_{i-1}(\bar{t}) - l = -(\nu_{i-1}-1)l - l =  -(\nu_i-1)l.
\end{align*}
We claim equality also holds.
To see this, suppose $x_i(\bar{t}) < -(\nu_i-1)l$.
Then, the time that vehicle $i$ arrives at the intersection region (according to trajectory $x_i$) is strictly later than $\tau_i + L/v_m$, a contradiction.
Thus, $x_i(\bar{t}) = -(\nu_i-1)l$.

Similarly, we show that $\dot{x}_i(\bar{t}) = v_m$.
By construction of the ${\tt MotionSynthesize}$ procedure, $\dot{x}_i(\bar{t}) \leq v_m$.
We claim equality. 
To see this, suppose $\dot{x}_i(\bar{t}) < v_m$.
Then, vehicle $i$ arrives at the intersection region (according to $x_i$) strictly later than time $\tau_i + L/v_m$, a contradiction.
Thus, $\dot{x}_i(\bar{t}) = v_m$.
\end{proof}
Although not necessary for the development of the proof of this lemma, it is an immediate corollary of Proposition~\ref{proposition:intersection} that vehicle $i$ travels with maximum speed from state $\mathbf{I}_i$ to the intersection region.

Next, we show that the time between the arrival of vehicle $A$ and the arrival of vehicle $i$ at state $\mathbf{I}_i$ is at least $2v_m/a_m$, \ie, Inequality~\eqref{equation:tauigeo}.
Let
\begin{gather} \label{order:t0prime}
\begin{split}
\Big( (i_l,k_l),\ldots,(i_{m-1},k_{m-1}),(B,k_B), \\ 
(A,2), \\ 
(C,k_C),(i_{m},k_{m}),\ldots,(i_n,k_n)\Big)
\end{split}
\end{gather}
be the service order of vehicles computed at time $t_0'$, and let
\begin{gather} \label{order:t0}
\begin{split}
\Big( (i_1,k_1),\ldots,(i_{m-1},k_{m-1}),(B,k_B), \\ 
(C,k_C),(i_{m},k_{m}),\ldots,(i_n,k_n) \Big)
\end{split}
\end{gather}
be the service order of vehicles computed at time $t_0$. (See Section~\ref{subsection:simulateQueue} and also Assumption~\ref{assumption:regular_policy}.)
First, note that all vehicles scheduled after vehicle $A$ are in lane 1, \ie, $k_C = k_m = \cdots = k_n = 1$, since a regular polling policy must be {\em first-come first-serve} (FCFS).
From the previous schedule order~\eqref{order:t0}, we can express the previous schedule time $\tau_i$ of vehicle $i$ as
\begin{align} \label{equation:tauioldtauc}
	\tau_i = \tau_C + (\nu_i - 1) s.
\end{align}
Define $\theta_{BC}$ to be equal to 1 if vehicles $B$ and $C$ are in opposite lanes, \ie, $k_B \not= k_C$, and 0, if they are in the same lanes, \ie, $k_B = k_C$.
Then, the previous schedule time of vehicle $C$ can be expressed as
\begin{align} \label{equation:taucold}
	\tau_C = \tau_B + s + r \theta_{BC},
\end{align}
where $r$ we defined as the switch time between queues.
Combining Equations~\eqref{equation:tauioldtauc} and~\eqref{equation:taucold}, we obtain
\begin{align} \label{equation:tauitaub}
	\tau_i = \tau_B + s + r\theta_{BC} + (\nu_i - 1) s.
\end{align}
Next, we establish $\tau_B + s \geq t_0'$ in the following proposition.
\begin{proposition} \label{proposition:taubt0}
Recall that $t_0'$ is the arrival time of newly arrived vehicle $A$ in lane 2.
Let~\eqref{order:t0prime} and~\eqref{order:t0} the schedule orders at times $t_0'$ and $t_0$, respectively, be generated according to a regular polling policy (Assumption~\ref{assumption:regular_policy}).
Then, $\tau_B + s \geq t_0'$ holds.
\end{proposition}
\begin{proof}
Suppose vehicles $B$ and $A$ are in the same lane.
Then, $\tau_A = \tau_B + s$.
Since $\tau_A \geq t_0'$ always holds, we have established this proposition. Now suppose, vehicles $B$ and $A$ are in opposite lanes.
Then, applying the previous argument only assures us that vehicle $A$ arrives before the server finishes servicing vehicle $B$ and and switching lanes, \ie, $t_0' \geq \tau_A = \tau_B + s + r$.
Thus, we argue differently.
First, note that vehicles $B$ and $C$ must be in the same lane.
Suppose that vehicle $A$ arrives strictly later than the service completion of vehicle $B$, \ie, $t_0' > \tau_B + s$.
Then, we claim that vehicle $C$ is serviced earlier than vehicle $A$, contradicting the updated schedule order~\eqref{order:t0prime}.
Consider what happens at the service completion of vehicle $B$, \ie, at time $\tau_B + s$.
Since vehicle $A$ has not arrived yet, queue 2 is empty.
Thus, according to the regular polling policy condition (Assumption~\ref{assumption:regular_policy}), vehicle $C$ will be serviced immediately after vehicle $B$, \ie, $\tau_C = \tau_B  + s$.
A nontrivial amount of time $\tau_B + s - t_0' > 0$ has elapsed.
Note that any polling policy (not just regular polling policies) cannot interrupt the service of a customer in order to service another customer.
Thus, vehicle $A$ will be serviced after vehicle $C$, contradicting the updated schedule order~\eqref{order:t0prime}.
Therefore, $\tau_B + s\geq t_0'$.
\end{proof}
From Equation~\eqref{equation:tauitaub} and Proposition~\ref{proposition:taubt0}, we establish
\begin{align} \label{equation:tauit0}
	\tau_i \geq t_0' + (\nu_i - 1) s.
\end{align}
Hence, the time that vehicle $i$ reaches the intersection region (according to previous trajectory $x_i$) is at time
\begin{align}
	\tau_i + L/v_m	&\geq t_0' + (\nu_i - 1) s + L / v_m \nonumber \\
				&\geq t_0' + (\nu_i - 1) s + 2 v_m/a_m, \label{equation:tauivmam}
\end{align}
applying Inequality~\eqref{equation:tauit0}, first, and then the minimum control region length condition (Assumption~\ref{assumption:control_region_length}).
Rearranging Inequality~\eqref{equation:tauivmam}, we obtain
\begin{align} \label{equation:tauigeo}
\tau_i + L/v_m - (\nu_i-1)s  - t_0' \geq 2v_m / a_m.
\end{align}
Recall that Proposition~\ref{proposition:intersection} states that vehicle $i$ arrives at state $\mathbf{I}_i$ (according to trajectory $x_i$) at time $\tau_i + L/v_m - (\nu_i-1)s$. Thus, Inequality~\eqref{equation:tauigeo} simply says that the time from the arrival of vehicle $A$ until the arrival time of vehicle $i$ at state $\mathbf{I}_i$ is greater than or equal to $2 v_m/ a_m$.

We conclude the proof of this lemma by showing that the state of vehicle $i$ at time $t_0'$ must be in $\mathcal{F}_i$, \ie, $z_i(t_0') \in \mathcal{F}_i$.
Consider the set of all states that can reach state $\mathbf{I}_i$, \ie,
\begin{align*}
	\mathbf{\Omega}(\mathbf{I}_i) := \{ \mathbf{z} \in [-L,0] \times [0,v_m] \colon \exists\, t_f \,\mathrm{s.t.} \, \mathscr{C}(\mathbf{z},t_0',t_f,\emptyset) \not= \emptyset \}.
\end{align*}
See Figure~\ref{figure:omega}.
First, note that from any state $\mathbf{z} \in \mathcal{F}_i$, we can accelerate to maximum speed and then continue at maximum speed until arriving at state $\mathbf{I}_i$; hence, $\mathcal{F}_i \subseteq {\bf \Omega}({\bf I}_i)$.
\begin{figure}
	\centerline{\includegraphics[clip=true, trim=2.15in 6.5in 4.17in 3.4in, width=0.5\textwidth]{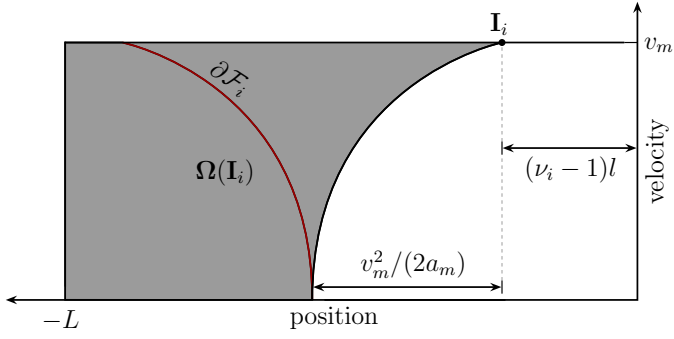}}
	\caption{State $\mathbf{I}_i := (-(\nu_i-1) l , v_m)$ is shown. Set ${\bf \Omega} ({\bf I}_i)$, shaded in grey, depicts the set of all states from which state ${\bf I}_i$ can be reached. Note that $\mathcal{F}_i$ lies entirely within ${\bf \Omega} ({\bf I}_i)$. } \label{figure:omega}
\end{figure}
Now, suppose the state of vehicle $i$ is not in $\mathcal{F}_i$ at time $t_0'$, \ie, $z_i(t_0') \in \mathbf{\Omega}(\mathbf{I}_i) \setminus \mathcal{F}_i$.
Note that by construction, any state $(p_0,v_0) \in \mathbf{\Omega}(\mathbf{I}_i) \setminus \mathcal{F}_i$ satisfies $|p_0| < (v_0^2 + v_m^2)/(2a_m)$.
Essentially, this means that a {\it feasible} trajectory initialized in $\mathbf{\Omega}(\mathbf{I}_i) \setminus \mathcal{F}_i$ at time $t_0'$ cannot possibly come to a full stop after time $t_0'$, since this trajectory cannot satisfy the terminal constraints.
We make this precise in the following proposition.
\begin{proposition} \label{proposition:minVelocity}
Let a vehicle be initialized with velocity $v_0$, \ie, $\dot{x}(0) = v_0$. Let $d$ be the distance it must travel, \ie, $x(0) = 0$ and $x(\tilde{t}_f) = d$. Suppose the following constraints are satisfied: $d < (v_0^2+v_m^2)/(2a_m)$ and $\dot{x}(\tilde{t}_f) = v_m$.
Then, the time $\max \tilde{t}_f$ of the longest duration feasible trajectory satisfies $ \max \tilde{t}_f < 2v_m/a_m$.
\end{proposition}

\begin{proof}
We solve the following free terminal time optimal control problem:
\begin{subequations} \label{oc:original}
	\begin{align} 
		\max_{x : [0,\tilde{t}_f] \to \reals} &\quad \int_{0}^{\tilde{t}_f} dt \\
		\mathrm{subject}\,\,\mathrm{to} &\quad \ddot{x}(t) = u(t), \mbox{ for all } t \in [0, \tilde{t}_f];\\
		&\quad 0 \le \dot{x}(t) \le v_m, \mbox{ for all } t\in [0, \tilde{t}_f]; \label{oc:original_velocity_constraint} \\ 
		& \quad \vert u(t) \vert \le a_m, \mbox{ for all } t\in [0, \tilde{t}_f];  \\
		& \quad x(0) = 0;  \quad \dot{x}(0) = v_0;\\
		& \quad x(\tilde{t}_f) = d; \quad \dot{x}(\tilde{t}_f) = v_m.
	\end{align}
\end{subequations}
Denote $v_l$ as the lowest velocity that the vehicle achieves during its trajectory, \ie, $v_l = \inf \{ \dot{x}(t) \colon \forall t \in [0,\tilde{t}_f], \forall x \in \mathscr{C}((0,v_m),0,\tilde{t}_f,\emptyset) \}$.
Suppose $v_l \leq 0$.
Note that the distance required to decelerate from the current velocity $v_0$ to a full stop and then accelerate from a full stop to full speed is precisely $(v_0^2+v_m^2)/(2a_m)$. But since $d < (v_0^2+v_m^2)/(2a_m)$, we have a contradiction. Thus, $v_l > 0$, and, hence,~\eqref{oc:original} does not grow unboundedly.
Due to the difficulty of handling state constraints directly, we relax this problem to the following optimal control problem, for which a solution is readily available:
\begin{subequations} \label{oc:relaxed}
	\begin{align} 
		\max_{x : [0,\tilde{t}_f] \to \reals} &\quad \int_{0}^{\tilde{t}_f} dt \label{oc:relaxedCost}\\
		\mathrm{subject}\,\,\mathrm{to} &\quad \ddot{x}(t) = u(t), \mbox{ for all } t \in [0, \tilde{t}_f]; \label{oc:relaxed1}\\
		& \quad |u(t)| \le a_m, \mbox{ for all } t\in [0, \tilde{t}_f];  \label{oc:relaxed2}\\
		& \quad x(0) = 0;  \quad \dot{x}(0) = v_0; \label{oc:relaxed3}\\
		& \quad x(\tilde{t}_f) = \xi_0; \quad \dot{x}(\tilde{t}_f) = v_m. \label{oc:relaxed4}
	\end{align}
\end{subequations}
By Pontryagin's Minimum Principle, the control law for (\ref{oc:relaxed}) is simply bang-bang control. Furthermore, the optimal trajectory $\tilde{x}$ which maximizes the terminal time $\tilde{t}_f$ can be divided into two distinct phases.
In the first phase, the vehicle decelerates to the lowest speed $v_l$ with maximum deceleration, while in the second phase, vehicle accelerates to full speed, with maximum acceleration.
This trajectory is summarized in Table~\ref{table:maxTime}.
\begin{table}[ht]
\centering
\begin{tabular}{c c c c c}
\hline\hline
\noalign{\vskip 3pt} 
Phase & $\Delta t$ & $\Delta x$ & $\Delta v$ & $\ddot{x}$ \\ [0.5ex]
\hline
\noalign{\vskip 3pt} 
1	&	$\frac{v_0-v_l}{a_m}$	&	$\frac{v_0^2 - v_l^2}{2a_m}$ & $v_l - v_0$ & $-a_m$\\ [1ex]
2 	& 	$\frac{v_m-v_l}{a_m}$ 	& $\frac{v_m^2-v_l^2}{2a_m}$ & $v_m-v_l$ & $a_m$\\ [1ex]
\hline
\noalign{\vskip 3pt} 
Total & 	$\frac{v_m+v_0-2v_l}{a_m}$	& $\frac{v_m^2+v_0^2-2v_l^2}{2a_m}$ & $v_m-v_0$ \\ [1ex]
\hline\hline
\vspace{1pt}
\end{tabular}
\caption{Latest Schedule Time Trajectory}\label{table:maxTime}
\end{table}
Note that the velocity of the vehicle throughout the trajectory satisfies the velocity constraint~\eqref{oc:original_velocity_constraint}; therefore, we have also solved~\eqref{oc:original}.
Finally, note that 
\begin{align*}
	\max \tilde{t}_f = \frac{v_m+v_0-2v_l}{a_m} < \frac{v_m + v_0}{a_m} \leq 2\frac{v_m}{a_m},
\end{align*}
where the first inequality holds from $v_l > 0$ and the second inequality from $v_0 \leq v_m$.
\end{proof}
Supposing $z_i(t_0') \in \mathbf{\Omega}(\mathbf{I}_i) \setminus \mathcal{F}_i$ and applying Proposition~\ref{proposition:minVelocity}, the time between the arrival of vehicle $A$ and the arrival of vehicle $i$ at state $\mathbf{I}_i$ is strictly less than $2 v_m / a_m$, contradicting Inequality~\eqref{equation:tauigeo}.
Hence, $z_i(t_0') \in \mathbf{\Omega}(\mathbf{I}_i) \cap \mathcal{F}_i = \mathcal{F}_i$.
This completes the proof of Lemma~\ref{lemma:uncommittedInFi}.

\newpage
\subsection{Proof of Lemma~\ref{lemma:EiNotEmpty}} \label{section:C}

In this section, we show that the set $\mathscr{E}_i'$ is non-empty.
Recall that a trajectory $\tilde{x}_i \in \mathscr{E}_i'$ has the following properties:
\begin{align}
	&\tilde{x}_i(\tau_i' + L/v_m) = 0; \\
	&\dot{\tilde{x}}_i(\tau_i' + L/v_m) = v_m; \\
	&\tilde{x}_i (t) = x_i(t), \text{ for all } t \in [t_0',t_c[x_i]]; \\
	&\tilde{x}_i (t) \leq x_{i-1}'(t) - l, \text{ for all } t \geq t_0'.
\end{align}
Any feasible trajectory $x$ can be uniquely described by the triple $(\velpath{x},\wait{x},t_0)$, where $\velpath{x}$ maps the position of trajectory $x$ to the velocity at that position, $\wait{x}$ maps the position of trajectory $x$ to the amount of time spent at that position, and $t_0$ is the initial time of the trajectory. Note that $\wait{x}(p)$ is 0 for any position $p$ with nonzero velocity.
Because a vehicle might fully stop, the mapping of its trajectory from the position space to the time domain is not a function in general. Thus, we will make frequent use of the arrival time of a trajectory $x$ at position $p$, \ie, $\arrive{x}(p)$, and the departure time of a trajectory $x$ from position $p$, \ie, $\depart{x}(p)$. Formerly, we define:
\begin{align*}
	\arrive{x}(p) &:= \inf \{  t \colon x(t) = p  \} \\
	\depart{x}(p) &:= \sup \{   t \colon x(t) = p  \}.
\end{align*}

First, consider the case $\nu_i = 1$.
We introduce the following proposition. Note Figure~\ref{figure:boundedPaths}.
\begin{proposition} \label{proposition:boundedPaths}
Let $\mathbf{z}_0 := (p_0,v_0)$. 
Suppose there exist a ``high-performance" trajectory $\hat{x} \in \mathscr{C}(\mathbf{z}_0,0,\hat{t},\emptyset)$ and a ``safe" trajectory $\bar{x} \in \mathscr{C}(\mathbf{z}_0,0,\bar{t},\emptyset)$, satisfying
\begin{align} \label{equation:slow_fast}
	\bar{t} \geq \hat{t}.
\end{align}
Suppose the following inequalities hold for all $p \in [p_0, 0]$:
\begin{align}
	\velpath{\bar{x}}(p) \leq \velpath{\hat{x}}(p)     \label{equation:p1} \\
	\wait{\bar{x}}(p) \geq \wait{\hat{x}}(p). \label{equation:delta1}
\end{align}
Then, for any terminal time $\tilde{t}$ satisfying $\bar{t} \geq \tilde{t} \geq \hat{t}$, there exists a trajectory $\tilde{x} \in \mathscr{C}(\mathbf{z}_0,0,\tilde{t},\emptyset)$ satisfying the following, for all $p \in [p_0,0]$:
\begin{align}
	 \velpath{\bar{x}}(p) \leq \velpath{\tilde{x}} (p)\leq \velpath{\hat{x}} (p)  \label{equation:p2} \\
	 \wait{\bar{x}}(p) \geq \wait{\tilde{x}}(p) \geq \wait{\hat{x}}(p). \label{equation:delta2}
\end{align}
\end{proposition}
\begin{figure}
	\centerline{\includegraphics[clip=true, trim=2.1in 6.5in 4.1in 3.4in, width=0.5\textwidth]{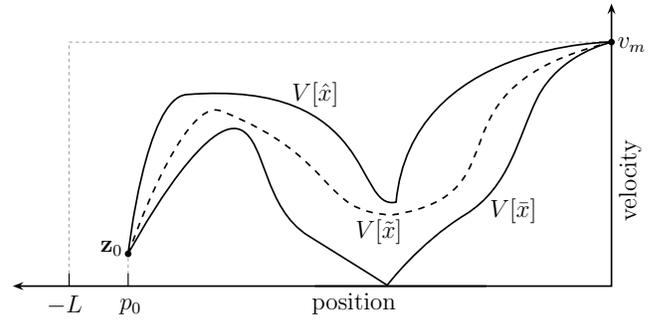}}
	\caption{ Notation for Proposition~\ref{proposition:boundedPaths}. ``High-performance" velocity path $\velpath{\hat{x}}$ and ``safe" velocity path $\velpath{\bar{x}}$ are depicted in solid black. Bounded velocity path $\velpath{\tilde{x}}$ is depicted in dashed black. Note that all trajectories have the same initial state $\mathbf{z}_0$. Note that terminal time of $\tilde{x}$ is also bounded by the terminal times of the ``high-performance" and ``safe" trajectories.} \label{figure:boundedPaths}
\end{figure}
We define $x := \mathscr{T}(\mathbf{z}_0,u)$ as the unique solution of the linear ODE $\ddot{x} = u$ with initial state $\mathbf{z}_0$.
Let $u_i$ be a control action that generates trajectory $x_i$, \ie, $u_i := \ddot{x}_i$.
Denote $t_f$ and $t_f'$ as the times that vehicle $i$ reaches the intersection region according to trajectories $x_i$ and $x_i'$, respectively, \ie, $t_f := \tau_i + L/v_m$ and $t_f' := \tau_i' + L/v_m$.
Let $v_i^c := \dot{x}_i(t_c[x_i])$.
Consider trajectory $\bar{x}_i := \mathscr{T}(z_i(t_0'),\bar{u}_i)$ (depicted in Figure~\ref{figure:uii}), where $\bar{u}_i$:
\begin{align} \label{eqn:ubarbar1}
\bar{u}_i(t) =
	\begin{cases}
		u_i(t)  &\colon t \in [t_0',t_c[x_i]) \\
		-a_m &\colon t \in [t_c[x_i],t_s) \\
		0	&\colon t \in [t_s, t_s + \sigma) \\
		a_m &\colon t \in [t_s + \sigma,t_s + \sigma +v_m/a_m],
	\end{cases}
\end{align}
where
\begin{align}
	t_s &:= t_c[x_i] + v_i^c/a_m, \\ 
	\sigma &:= \max \{0, t_f'-t_s-v_m/a_m \}.
\end{align}
\begin{figure}
	\centerline{\includegraphics[clip=true, trim=2.1in 6.5in 4.1in 3.4in, width=0.5\textwidth]{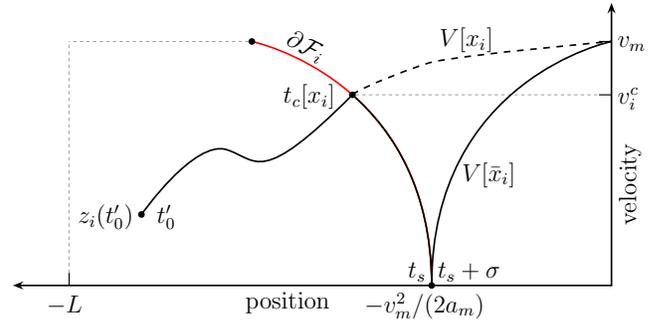}}
	\caption{Using terminology from Proposition~\ref{proposition:boundedPaths}, ``safe" velocity path $\velpath{\bar{x}_i}$ as generated by control $\bar{u}_i$ defined in Equation~\eqref{eqn:ubarbar1} is depicted in solid black. Recall that time $t_c[x_i]$ is defined as the time that trajectory $x_i$ exits set $\mathcal{F}_i$. ``High-performance" velocity path $\velpath{x_i}$ replicates $\velpath{\bar{x}_i}$ from $t_0'$ to $t_c[x_i]$ and then afterwards, follows the dashed black line. Note that during ``safe" trajectory $\bar{x}_i$, the vehicle halts for an amount of time $\sigma$ at a distance $v_m^2 / (2 a_m)$ from the intersection region.} \label{figure:uii}
\end{figure}
Borrowing the language from Proposition~\ref{proposition:boundedPaths}, the ``high-performance" trajectory is $x_i$, and the ``safe" trajectory is $\bar{x}_i$.
We show the existence of a feasible trajectory with terminal time $t_f'$.
First, note by construction of $\sigma$, the time that $\bar{x}_i$ reaches the intersection region is no earlier than $t_f'$.
Hence, $t_s + \sigma + v_m/a_m  \geq t_f' \geq t_f$.
Second, note $\wait{x_i}(-v_m^2/(2a_m)) \leq \sigma = \wait{\bar{x}_i}(-v_m^2/(2a_m))$; otherwise, $x_i$ would arrive at the intersection region strictly later than $t_f'$, \ie, $t_f > t_f'$, which is a contradiction.
Hence, Inequality~\eqref{equation:delta1} holds.
Lastly, directly from the construction of $\bar{u}_i$, Inequality~\eqref{equation:p1} holds.
Thus, by Proposition~\ref{proposition:boundedPaths}, there exists trajectory $\tilde{x}_i \in \mathscr{C}(z_i(t_0'),t_0',t_f',\emptyset)$ satisfying Inequalities~\eqref{equation:p2} and~\eqref{equation:delta2}.

Now, we state a proposition, depicted in Figure~\ref{figure:aboveBelow}:
\begin{proposition} \label{proposition:boundedTimes}
Suppose there exist a ``fast" trajectory $\hat{x} \in \mathscr{C}(\mathbf{\hat{z}}_0,t_0,\hat{t},\emptyset)$ and a ``slow" trajectory $\bar{x} \in \mathscr{C}(\mathbf{\bar{z}}_0,t_0,\bar{t},\emptyset)$, such that for all $p \in (p_0,p_f)$:
\begin{align}
	\velpath{\bar{x}}(p) 		&\leq 	\velpath{\hat{x}}(p) 	\label{equation:pp}\\
	\wait{\bar{x}}(p)			&\geq 	\wait{\hat{x}}(p) 	\label{equation:deldel} \\
	\depart{\bar{x}}(p_0) 		&\geq  	\depart{\hat{x}}(p_0). \label{equation:dd} 
\end{align}
Then, for all $t \in  [\arrive{\hat{x}}(p_0), \depart{\bar{x}}(p_f)]$, we have
\begin{align}
	\bar{x}(t) \leq \hat{x}(t). \label{eq:prop12}
\end{align}
\end{proposition}
\begin{figure}
	\centerline{\includegraphics[clip=true, trim=2.1in 6.5in 4.1in 3.4in, width=0.5\textwidth]{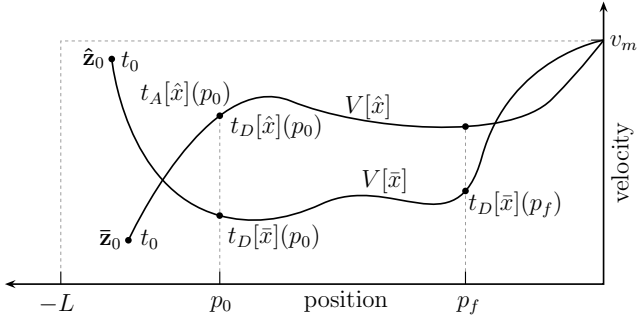}}
	\caption{Notation for Proposition~\ref{proposition:boundedTimes}. ``Fast" velocity path $\velpath{\hat{x}}$ and ``slow" velocity path $\velpath{\bar{x}}$ are shown. Note that at each position $p \in (p_0,p_f)$, the velocity of the ``fast" trajectory is greater than that of the ``slow" trajectory. See Inequality~\eqref{equation:pp}. Also, Inequality~\eqref{equation:dd} states that $\hat{x}$ departs position $p_0$ earlier than $\bar{x}$. The claim of this proposition, expressed in Inequality~\eqref{eq:prop12}, states that $\hat{x}$ is always ahead, or closer to the intersection, than $\bar{x}$.} \label{figure:aboveBelow}
\end{figure}
First, note that Inequalities~\eqref{equation:pp} and~\eqref{equation:deldel} are satisfied for trajectories $\bar{x}_i$, $\tilde{x}_i$, and $x_i$, \ie, 
\begin{align*}
		\velpath{\bar{x}_i}	&\leq \velpath{\tilde{x}_i} 	\leq \velpath{x_i} \\
		\wait{\bar{x}_i} 		&\geq \wait{\tilde{x}_i} 	\geq  \wait{x_i}.
\end{align*}
Define $p_i' := x_i(t_0')$.
Note that Inequality~\eqref{equation:dd} also holds for trajectories $\bar{x}_i$, $\tilde{x}_i$, and $x_i$, \ie,
\begin{align*}
	\depart{\bar{x}_i}(p_i')  \geq \depart{\tilde{x}_i}(p_i') \geq \depart{x_i}(p_i').
\end{align*}
Then, by Proposition~\ref{proposition:boundedTimes}, for all $t \geq t_0'$:
\begin{align*}
	\bar{x}_i(t) \leq \tilde{x}_i(t) \leq x_i(t).
\end{align*}
Hence, $\tilde{x}_i(t) = x_i(t)$ for all $t \in [t_0',t_c[x_i]]$.
It is left to show that trajectories $\tilde{x}_i$ and $x_{i-1}'$ are safe for all time.
Since $\nu_i = 1$, we have $x_{i-1}' = x_{i-1}|_{ \{ t \colon t\geq t_0' \} }$, by Proposition~\ref{lemma:committed}.
Thus, for any $t \geq t_0'$:
\begin{align*}
	\tilde{x}(t) \leq x_i(t) \leq x_{i-1}(t) - l = x_{i-1}'(t) - l.
\end{align*}
Hence, $\tilde{x}_i \in \mathscr{E}_i'$. \\

Now, we consider vehicle $i$ satisfying $\nu_i > 1$.
Recall that at the start of the induction step, the updated trajectory $x_{i-1}'$ of the vehicle $i-1$ and the previous trajectory $x_i$ of vehicle $i$ are known; and we would like to construct the updated trajectory $x_i'$ of vehicle $i$.
We consider three cases.
In each of those cases, we construct a ``high-performance" trajectory $\hat{x}_i$ and a ``safe" trajectory $\bar{x}_i$, and then apply the following proposition to establish that there exists a trajectory $\tilde{x}_i$ for vehicle $i$ that has both performance and safety guarantees, \ie, $\mathscr{E}_i'$ is non-empty.
Consider Figure~\ref{figure:sectionCmainprop}.
\begin{figure}
	\centerline{\includegraphics[clip=true, trim=2.1in 6.5in 4.1in 3.4in, width=0.5\textwidth]{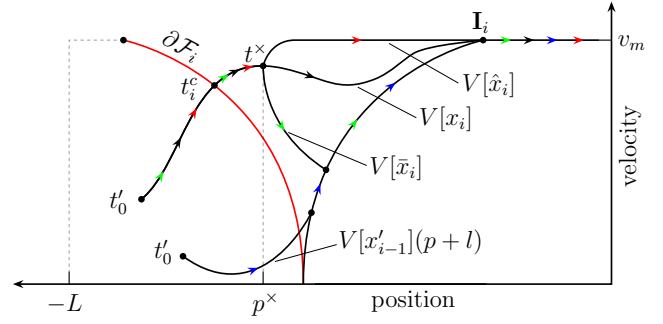}}
	\caption{Notation for Proposition~\ref{proposition:sectionCmainprop}. Previous trajectory $x_i$ is identified by following the black arrows, ``high-performance" trajectory $\hat{x}_i$ by red arrows, ``safe" trajectory by green arrows, and updated trajectory $x_{i-1}'$ of vehicle in front by blue arrows. Both ``high-performance" trajectory $\hat{x}_i$ and safe trajectory $\bar{x}$ replicate $x_i$ until time $t^\times$. After time $t^\times$, ``high-performance" velocity path $\velpath{\hat{x}_i}$ always has a higher velocity than the velocity path $\velpath{\bar{x}_i}$ of the ``safe" trajectory. See Inequality~\eqref{equation:nivv}. Time $t_i^c$ is the time that trajectory $x_i$ leaves set $\mathcal{F}_i$. See notation in~\eqref{equation:notation13}. Trajectory $\bar{x}_i$ is safe with updated trajectory $x_{i-1}'$ of the vehicle directly in front. Note that we shift velocity path of vehicle $i-1$ to the left by $l$, the vehicle length. We adopt this visualization since it is more natural to discuss safety of vehicle $i$ with respect to the rear bumper of vehicle $i-1$.} \label{figure:sectionCmainprop}
\end{figure}
\begin{proposition} \label{proposition:sectionCmainprop}
Let time $t^\times$ be such that $t^\times \geq t_c[x_i]$.
Suppose there exists a ``high-performance" trajectory $\hat{x}_i \in \mathscr{E}(x_i|_{[t_0',t^\times]},t_0',\hat{t},\emptyset)$ for some $\hat{t}$ satisfying
\begin{align} \label{equation:terminaltime}
	t_f' \geq \hat{t} ,
\end{align}
and a ``safe" trajectory $\bar{x}_i \in \mathscr{E}(\hat{x}_i|_{[t_0',t^\times]},t_0',\bar{t},x_{i-1}')$.
Define $p^\times := \hat{x}_i(t^\times)$.
Suppose the following are satisfied for all $p \in [p^\times,0]$:
\begin{align}
	\velpath{\bar{x}_i}(p) 	&\leq \velpath{\hat{x}_i}(p) \label{equation:nivv} \\
	\wait{\bar{x}_i}(p) 		&\geq \wait{\hat{x}_i}(p) \label{equation:nidd}
\end{align}
Suppose the following are satisfied for all $p \in (p^\times,-l)$:
\begin{align}
	\velpath{\bar{x}_i}(p) 	&\geq \velpath{x_{i-1}'}(p+l) \label{equation:nivvold} \\
	\wait{\bar{x}_i}(p) 		&\leq \wait{x_{i-1}'}(p+l). \label{equation:niddold}
\end{align}
Then, $\mathscr{E}_i'$ is non-empty.
\end{proposition}
We use the following notation throughout the three cases:
\begin{align} \label{equation:notation13}
\begin{split}
	t_i^c &:= t_c[x_i]  ;  \\	
	p_i^c &:= x_i(t_i^c) ;  \\
	v_i^c &:= \dot{x}_{i}(t_{i}^c) ;
\end{split}
\begin{split}
	t_{i-1}^c &:= t_c[x_{i-1}] ; \\
	p_{i-1}^c &:= x_{i-1}(t_{i-1}^c) - l; \\
	v_{i-1}^c &:= \dot{x}_{i-1}(t_{i-1}^c).
\end{split}
\end{align}
\begin{figure}
	\includegraphics[clip=true, trim=2.1in 6.5in 4.1in 3.4in, width=0.5\textwidth]{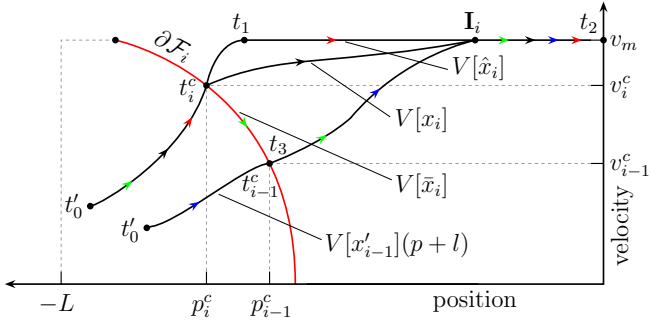}
	\caption{Depiction of Case 1 in Section C. Previous trajectory $x_i$ is identified by following the black arrows, ``high-performance" trajectory $\hat{x}_i$ by red arrows, ``safe" trajectory by green arrows, and updated trajectory $x_{i-1}'$ of vehicle in front by blue arrows. Both ``high-performance" trajectory $\hat{x}_i$ and safe trajectory $\bar{x}$ replicate $x_i$ until time $t_i^c$, which denotes the time that $x_i$ leaves the set $\mathcal{F}_i$. Inequality~\eqref{equation:case1position} is depicted clearly. Also, Inequality~\eqref{equation:case1time} states that previous trajectory $x_i$ leaves set $\mathcal{F}_{i}$, denoted time $t_i^c$, after previous trajectory $x_{i-1}$ of vehicle $i-1$ leaves set $\mathcal{F}_{i-1}$, denoted time $t_{i-1}^c$. In ``high-performance" trajectory $\hat{x}$, time $t_1$ denotes the time that the vehicle reaches maximum speed, \ie, $t_1$ is the earliest time $t$ such that $\dot{\hat{x}}_i(t) = v_m$. See definition of control $\hat{u}_i$ in Equation~\eqref{eqn:hat1}. Time $t_2$ is the time that ``high-performance" trajectory $\hat{x}_i$ reaches the intersection region, \ie, state $(0,v_m)$. Time $t_3$ is the time that ``safe" trajectory $\bar{x}_i$ arrives at state $(p_{i-1}^c,v_{i-1}^c)$. See definition of control $\bar{u}_i$ in Equation~\eqref{eqn:ubar1}. Note that ``safe" trajectory $\bar{x}_i$ arrives at this state later than when previous trajectory $x_{i-1}$ departs from this state, \ie, $t_3 \geq t_{i-1}^c$.} \label{figure:case1}
\end{figure}
{\it Case 1.}
See Figure~\ref{figure:case1}.
Suppose the following:
\begin{align}
	t_i^c &\geq t_{i-1}^c ; \label{equation:case1time} \\
	 p_i^c &\leq p_{i-1}^c. \label{equation:case1position}
\end{align}
First, we construct the ``high-performance" trajectory $\hat{x}_i := \mathscr{T}(z_i(t_0'),\hat{u}_i)$, where control action $\hat{u}_i$ is defined as follows:
\begin{align} \label{eqn:hat1}
\hat{u}_i(t) = 
	\begin{cases}
		u_i(t) &\colon t \in [t_0',t_i^c) \\
		a_m &\colon t \in [t_i^c,t_1)\\
		0 &\colon t \in [t_1,t_2],
	\end{cases}
\end{align}
where
\begin{align*}
	t_1 &:= t_i^c + (v_m-v_i^c)/a_m, \\
	t_2 &:= t_1 + \frac{p_i^c-\frac{v_m^2 - (v_i^c)^2}{2a_m}}{v_m}.
\end{align*}
Using $x_i$ as the ``slow" trajectory and $\hat{x}_i$ as the ``fast" trajectory, Inequalities~\eqref{equation:pp},~\eqref{equation:deldel}, and~\eqref{equation:dd} are satisfied; thus, by Proposition~\ref{proposition:boundedTimes}, we have $x_i(t) \leq \hat{x}_i(t)$ for all time $t$. 
Since trajectory $x_i$ reaches the intersection region at time $t_f$, which satisfies $t_f' > t_f$, then trajectory $\hat{x}_i$ must reach the intersection region before time $t_f'$, \ie, $ t_f' \geq t_2$; thus, Inequality~\eqref{equation:terminaltime} is met.

By Inequalities~\eqref{equation:case1time} and~\eqref{equation:case1position}, we have $t_3 \geq t_{i-1}^c$.
Now, we construct the ``safe" trajectory $\bar{x}_i$.
Define trajectory $\bar{x}_i := \mathscr{T}(z_i(t_0'),\bar{u}_i)$, where control action $\bar{u}_i$ is defined as follows:
\begin{align} \label{eqn:ubar1}
\bar{u}_i(t) = 
	\begin{cases}
		u_i(t) &\colon t \in [t_0',t_i^c) \\
		-a_m &\colon t \in [t_i^c,t_3) \\
		u_{i-1}'(t-t_3+t_{i-1}^c) &\colon t \in [t_3,t_4)\\
		0 &\colon t \in [t_4,t_4+s],
	\end{cases}
\end{align}
where
\begin{align}
	t_3 &:= t_i^c + (v_i^c - v_{i-1}^c)/a_m  \label{equation:case1_t3} \\
	t_4 &:= t_{i-1}'-t_{i-1}^c+t_3 \\
	t_{i-1}' &:= \tau_{i-1}' + L/v_m.
\end{align}
To borrow the language of Proposition~\ref{proposition:sectionCmainprop}, $t_i^c$ is $t^\times$, $\hat{x}_i$ is the ``high-performance" trajectory, and $\bar{x}_i$ is the ``safe" trajectory.
First, note $\hat{x}(t) = x_i(t) = \bar{x}_i(t)$ for all $t \in [t_0',t_i^c]$.
Second, by construction, Inequalities~\eqref{equation:nivv},~\eqref{equation:nidd},~\eqref{equation:nivvold} and~\eqref{equation:niddold} hold.
Lastly, note that $\bar{x}_i$ and $x_{i-1}'$ are safe for all time. 
For any $t \in [t_0',t_{i-1}^c]$: 
\begin{align*}
	\bar{x}_i(t) = x_i(t) \leq x_{i-1}(t)- l = x_{i-1}'(t) - l.
\end{align*}
For any $t \in [t_{i-1}^c,t_3]$:
\begin{align*}
	\bar{x}_i(t) \leq p_{i-1}^c \leq x_{i-1}'(t_{i-1}^c) - l.
\end{align*}
For any $t \geq t_3$: 
\begin{align*}
	\bar{x}_i(t) = x_{i-1}'(t - t_3 + t_{i-1}^c) - l \leq x_{i-1}'(t) - l,
\end{align*}
since $t_3 \geq t_{i-1}^c$.
By Proposition~\ref{proposition:sectionCmainprop}, $\mathscr{E}_i'$ is non-empty. \\

\begin{figure}
		\includegraphics[clip=true, trim=2.1in 6.5in 4.1in 3.4in, width=0.5\textwidth]{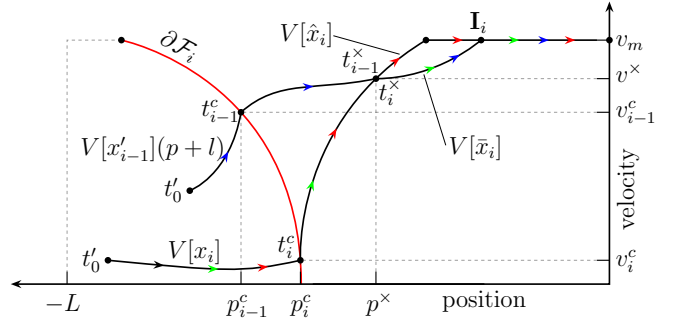}
		\caption{Depiction of Case 2 in Section C. Only the first segment of previous trajectory $x_i$ is depicted, by a black arrow. ``High-performance" trajectory $\hat{x}_i$ is depicted by the sequence of red arrows, ``safe" trajectory by green arrows, and updated trajectory $x_{i-1}'$ of vehicle in front by blue arrows. Inequality~\eqref{equation:case2position} is depicted clearly. Also, Inequality~\eqref{equation:case2time} states that previous trajectory $x_i$ leaves set $\mathcal{F}_{i}$, denoted time $t_i^c$, after previous trajectory $x_{i-1}$ of vehicle $i-1$ leaves set $\mathcal{F}_{i-1}$, denoted time $t_{i-1}^c$. Both ``high-performance" trajectory $\hat{x}_i$ and ``safe" trajectory $\bar{x}$ replicate $x_i$ until time $t_i^c$. Position $p^\times$ is the first position that ``high-performance" trajectory $\hat{x}_i$ and the rear bumper of the updated trajectory $x_{i-1}'$ intersect. See Definition~\eqref{eq:first}. Also, note that $\hat{x}_i$ departs position $p^\times$ at a later time than when the rear bumper of the updated trajectory $x_{i-1}'$ arrives at this position, \ie, $t_i^\times \geq t_{i-1}^\times$. Furthermore, ``high-performance" trajectory $\hat{x}_i$ and ``safe" trajectory $\bar{x}_i$ are equivalent up until time $t_i^\times$. See definition of control $\hat{u}_i$ in Equation~\eqref{eqn:ubar2}.} \label{figure:case2}
\end{figure}
{\it Case 2.}
See Figure~\ref{figure:case2}.
Suppose the following:
\begin{align}
	t_i^c &\geq t_{i-1}^c  \label{equation:case2time} \\
	p_{i}^c &> p_{i-1}^c. \label{equation:case2position}
\end{align}
In this case, we use the same ``high-performance" trajectory $\hat{x}_i$ as defined in Equation~\eqref{eqn:hat1}; and by the same argument made earlier, Inequality~\eqref{equation:terminaltime} holds.
Next, we construct the ``safe" trajectory $\bar{x}_i$.
By Inequality~\eqref{equation:case2position}, $\velpath{\hat{x}_i}(p)$ and $\velpath{x_{i-1}'}(p+l)$ must intersect after position $p_i^c$. 
Define 
\begin{align}
	p^\times &:= \inf \{ p \geq p_i^c \colon \velpath{\hat{x}_i}(p) = \velpath{x_{i-1}'}(p+l) \} \label{eq:first} \\
	v^\times &:= \velpath{\hat{x}_i}(p^\times) = \velpath{x_{i-1}'}(p^\times).
\end{align}
Define 
\begin{align*}
	t_i^\times 		&:= \depart{\hat{x}_i}(p^\times) , \\
	t_{i-1}^\times 	&:= \arrive{x_{i-1}'}(p^\times + l).
\end{align*}
By Inequalities~\eqref{equation:case2time} and~\eqref{equation:case2position}, we have $t_i^\times \geq t_{i-1}^\times$.
Define ``safe" trajectory $\bar{x}_i := \mathscr{T}(z_i(t_0'),\bar{u}_i)$, where control action $\bar{u}_i$ is defined as follows:
\begin{align} \label{eqn:ubar2}
\bar{u}_i(t) = 
	\begin{cases}
		\hat{u}_i(t) 					&\colon t \in [t_0',t_i^\times) \\
		u_{i-1}'(t-t_i^\times+t_{i-1}^\times) 	&\colon t \in [t_i^\times,t_5) \\
		0 							&\colon t \in [t_5,t_5+s],
	\end{cases}
\end{align}
where
\begin{align*}
	t_5 := t_{i-1}'+t_i^\times-t_{i-1}^\times.
\end{align*}
To borrow the language of Proposition~\ref{proposition:sectionCmainprop}, $t_i^\times$ is $t^\times$, $\hat{x}_i$ is the ``high-performance" trajectory, and $\bar{x}_i$ is the ``safe" trajectory.
First, note $\hat{x}_i(t) = x_i(t) = \bar{x}_i(t)$ for all $t \in [t_0',t_i^\times]$.
Second, by construction, Inequalities~\eqref{equation:nivv},~\eqref{equation:nidd},~\eqref{equation:nivvold} and~\eqref{equation:niddold} hold.
Lastly, note that $\bar{x}_i$ and $x_{i-1}'$ are safe for all time.
For any $t \in [t_0',t_{i-1}^c]$:
\begin{align*}
	\bar{x}_i(t) = x_i(t) \leq x_{i-1}(t) - l = x_{i-1}'(t) - l,
\end{align*}
since $t_i^c \geq t_{i-1}^c$. For any $t \in [t_{i-1}^c,t_i^\times]$:
\begin{align*}
	\bar{x}_i(t) \leq x_{i-1}'(t) - l,
\end{align*}
by Proposition~\ref{proposition:boundedTimes}. For any $t \geq t_i^\times$:
\begin{align*}
	\bar{x}_i(t) = x_{i-1}'(t-t_i^\times+t_{i-1}^\times) - l \leq x_{i-1}'(t) - l,
\end{align*}
since $t_i^\times \geq t_{i-1}^\times$.
By Proposition~\ref{proposition:sectionCmainprop}, $\mathscr{E}_i'$ is non-empty. \\

\begin{figure}
\centerline{\includegraphics[clip=true, trim=2.1in 6.5in 4.1in 3.4in, width=0.5\textwidth]{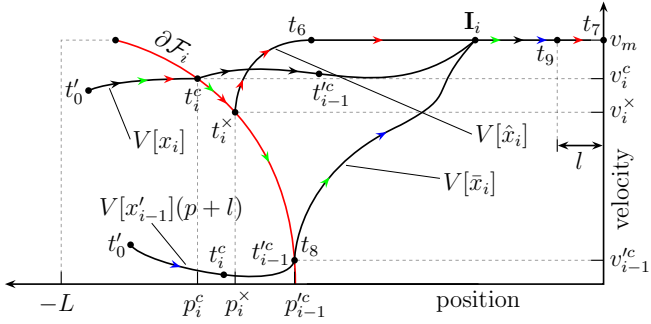}}
\caption{Depiction of Case 3 in Section C. Previous trajectory $x_i$ is depicted by a black arrow, ``high-performance" trajectory $\hat{x}_i$ is depicted by the sequence of red arrows, ``safe" trajectory by green arrows, and updated trajectory $x_{i-1}'$ of vehicle in front by blue arrows. Inequality~\eqref{equation:case3time} states that previous trajectory $x_i$ leaves set $\mathcal{F}_{i}$, denoted time $t_i^c$, earlier than when updated trajectory $x_{i-1}'$ of vehicle $i-1$ leaves set $\mathcal{F}_{i-1}$, denoted time $t_{i-1}'^c$. See definition in~\eqref{eq:defCase3}. We define time $t_i^\times$ in Equation~\eqref{eq:case3min}. ``High-performance" trajectory $\hat{x}_i$ and ``safe" trajectory $\bar{x}$ are equivalent up until time $t_i^\times$. After this time, $\hat{x}_i$ accelerates until reaching maximum speed at time $t_6$; while $\bar{x}_i$ continues to decelerate until its arrival at position $p_{i-1}'^c$ at time $t_8$. Note that $\bar{x}_i$ arrives at position $p_{i-1}'^c$ later than when the rear bumper of updated trajectory $x_{i-1}'$ of vehicle $i-1$ arrives at this same position, \ie, $t_8 \geq \arrive{(x_{i-1}'}(p_{i-1}'^c+l)$. Also, time $t_7$ is the time that ``high-performance" trajectory $\hat{x}_i$ reaches the intersection region, \ie, state $(0,v_m)$.} \label{figure:case3}
\end{figure}
{\it Case 3.}
See Figure~\ref{figure:case3}.
Suppose $t_{i-1}^c \geq t_i^c$. Since $t_c[x_{i-1}'] \geq t_{i-1}^c$ (because $x_{i-1}' \in \mathscr{E}_{i-1}'$ by Lemma~\ref{lemma:EiNotEmpty}), we have
\begin{align}
	t_i^c \leq t_c[x_{i-1}']. \label{equation:case3time}
\end{align}
Throughout this case, we use the following notation:
\begin{align} \label{eq:defCase3}
	\begin{cases}
	t_{i-1}'^c &:= t_c[x_{i-1}'] \\
	p_{i-1}'^c &:= x_{i-1}'(t_{i-1}'^c) - l \\
	v_{i-1}'^c &:= \dot{x}_{i-1}'(t_{i-1}'^c).
	\end{cases}
\end{align}
Define:
\begin{align}
	t_i^\times &:= \min \{ t_{i-1}'^c, t_i^c + v_i^c/a_m \} \label{eq:case3min}\\
	v_i^\times &:= 	v_i^c - a_m(t_i^\times-t_i^c). \label{equation:t7_linearconstraint}
\end{align}
Now, we construct the ``high-performance" trajectory $\hat{x}_i := \mathscr{T}(z_i(t_0'),\hat{u}_i)$, where control action $\hat{u}_i$ is defined as follows:
\begin{align*}
\hat{u}_i(t) = 
	\begin{cases}
		u_i(t) 	&\colon t \in [t_0',t_i^c) \\
		-a_m 	&\colon t \in [t_i^c,t_i^\times) \\
		a_m		&\colon t \in [t_i^\times,t_6) \\
		0 		&\colon t \in [t_6,t_7],
	\end{cases}
\end{align*}
where 
\begin{align*}
	t_6 		&:=	 t_i^\times + (v_m-v_i^\times)/a_m \\
	t_7 		&:= 	t_6 + \frac{ p_i^c - \frac{v_m^2 + (v_i^c)^2- 2(v_i^\times)^2}{2a_m}}{v_m}.
\end{align*}
``High-performance" trajectory $\hat{x}_i$ reaches the intersection region no later than $t_f'$.
To show this, we use Equation~\eqref{equation:t7_linearconstraint} to re-express $t_7$ as a function of $v_i^\times$ and some constants:
\begin{align*}
	t_7 = t_i^c + \frac{v_i^c - v_i^\times}{a_m} + \frac{v_m-v_i^\times}{a_m} + \frac{ p_i^c - \frac{v_m^2 + (v_i^c)^2- 2(v_i^\times)^2}{2a_m}}{v_m}.
\end{align*}
We upper bound $t_7$ as follows.
Note that $t_7$ is quadratic in $v_i^\times$ with positive curvature.
Also, $v_i^\times$ must range in $[v_{i-1}'^c,v_i^c]$.
Therefore, 
\begin{align*}
	t_7 = t_7(v_i^\times) \leq \max \{ t_7(v_{i-1}'^c) , t_7(v_i^c)  \}
\end{align*}
Note that $t_7(v_{i-1}'^c)$ is the time that the ``high-performance" trajectory $\hat{x}_i$ reaches the intersection region, with the constraint that vehicle $i$ arrives at $p_{i-1}'^c$ before updated trajectory $x_{i-1}'$ of vehicle $i-1$ departs $p_{i-1}'^c+l$.
Using $t_{i-1}' 	:= 	\tau_{i-1}' + L/v_m$:
\begin{align*}
	t_7(v_{i-1}'^c) \leq  t_{i-1}' + s = t_f'.
\end{align*}
Note that $t_7(v_i^c)$ is the time that ``high-performance" trajectory $\hat{x}_i$ reaches the intersection region if it copies exactly the previous trajectory $x_{i}$ of vehicle $i$ until position $p_i^c$ and then fully accelerates until the intersection region: 
\begin{align*}
	t_7(v_i^\times) \leq t_f \leq t_f'.
\end{align*}

Next, we construct ``safe" trajectory $\bar{x}_i$. Define trajectory $\bar{x}_i := \mathscr{T}(z_i(t_0'),\bar{u}_i)$, where control action $\bar{u}_i$ is defined as follows:
\begin{align} \label{eqn:ubar3}
\bar{u}_i(t) = 
	\begin{cases}
		\hat{u}_i(t)							&\colon t \in [t_0',t_i^\times) \\
		-a_m								&\colon t \in [t_i^\times,t_8) \\
		u_{i-1}'(t - t_8 + \arrive{x_{i-1}'}(p_{i-1}'^c+l))	&\colon t \in [t_8,t_9) \\
		0								&\colon t \in [t_9,t_9+s],
	\end{cases}
\end{align}
where
\begin{align*}
	t_8 		&:= t_i^\times + (v_i^\times-v_{i-1}'^c)/a_m \\
	t_9 		&:= t_8 + t_{i-1}' - \arrive{x_{i-1}'}(p_{i-1}'^c+l).
\end{align*}
First, note that $\bar{x}_i(t) = \hat{x}_i(t)$ for all $t \in [t_0',t_i^\times]$.
Second, by construction, Inequalities~\eqref{equation:nivv},~\eqref{equation:nidd},~\eqref{equation:nivvold} and~\eqref{equation:niddold} hold.
Lastly, $\bar{x}_i$ and $x_{i-1}'$ are safe for all time.
For any $t \in [t_0',t_{i-1}^c]$:
\begin{align*}
	\bar{x}_i(t) \leq x_i(t) \leq x_{i-1}(t) - l = x_{i-1}'(t) - l.
\end{align*}
For any $t \in [t_{i-1}^c,t_{i-1}'^c]$:
\begin{align*}
	\bar{x}_i(t) \leq x_{i-1}'(t) - l
\end{align*}
by Proposition~\ref{proposition:boundedTimes}.
For any $t \in [t_{i-1}'^c,t_8]$:
\begin{align*}
	\bar{x}_i(t) \leq p_{i-1}'^c \leq x_{i-1}'(t) - l.
\end{align*}
For any $t \geq t_8$:
\begin{align*}
	\bar{x}_i(t) = x_{i-1}'(t-t_8+\arrive{x_{i-1}'}(p_{i-1}'^c+l)) - l \leq x_{i-1}'(t) - l,
\end{align*}
since $t_8 \geq \arrive{(x_{i-1}'}(p_{i-1}'^c+l)$ by Inequality~\eqref{equation:case3time}.
By Proposition~\ref{proposition:sectionCmainprop}, $\mathscr{E}_i'$ is non-empty.


\subsection*{Proofs used in Section~\ref{section:C}}
{\it Proof of Proposition~\ref{proposition:boundedPaths}:}
Define:
\begin{align*}
	t_m = \int_{p_0}^{0} \frac{1}{\velpath{\bar{x}}(p)} dp + \sum_{\{ p \colon \wait{\hat{x}}(p) > 0\}} \wait{\hat{x}}(p).
\end{align*}
This is the terminal time of the trajectory uniquely generated by the triple $(\velpath{\bar{x}},\wait{\hat{x}},0)$, \ie, $\mathscr{G}(\velpath{\bar{x}},\wait{\hat{x}},0)$.
This trajectory copies the ``slow" trajectory $\bar{x}$ except at points of zero velocity; at which, the trajectory only stays as long as the ``fast" trajectory $\hat{x}$.
Note that this trajectory satisfies Inequalities~\eqref{equation:p2} and~\eqref{equation:delta2}.
Furthermore, for any terminal time $\tilde{t} \in [t_m,\bar{t}]$, we can just as easily create a trajectory satisfying those inequalities.
Define:
\begin{align*}
	\lambda &:= \frac{\tilde{t}-t_m}{\bar{t}-t_m} \\
	\delta_{\lambda} &:= \lambda(\wait{\bar{x}}-\wait{\hat{x}}) + \wait{\hat{x}}.
\end{align*}
Note that for any $\lambda \in [0,1]$, for all $p \in [p_0,0]$:
\begin{align*}
	\wait{\bar{x}}(p) \geq \delta_\lambda(p) \geq \wait{\hat{x}}(p).
\end{align*}
Therefore, the trajectory $\tilde{x}$ uniquely generated by the triple $(\velpath{\bar{x}},\delta_{\lambda},0)$, \ie, $\tilde{x} := (\velpath{\bar{x}},\delta_{\lambda},0)$, satisfies Inequalities~\eqref{equation:p2} and~\eqref{equation:delta2}.

Now, suppose $\tilde{t} < t_m$.
In this case, the velocity path $\velpath{\bar{x}}$ is too slow, and we must construct a different velocity path in order to meet the terminal time $\tilde{t}$. Consider the following velocity path.

\begin{figure}
\includegraphics[clip=true, trim=2.1in 6.5in 4.1in 3.4in, width=0.5\textwidth]{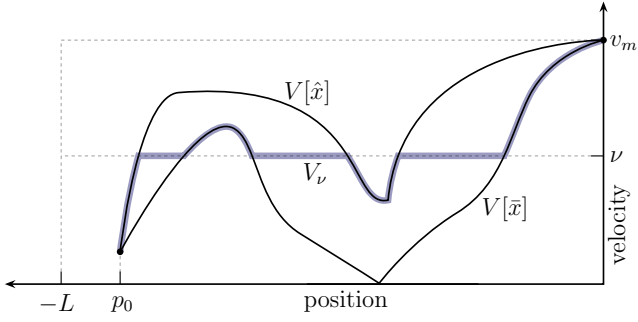}
		\caption{Velocity path $V_\nu$, shaded blue, as defined in Equation~\eqref{equation:vnu}, used in the proof of Proposition~\ref{proposition:boundedPaths}. We wish to find some velocity path $V_\nu$ (bounded between velocity paths $\velpath{\hat{x}}$ and $\velpath{\bar{x}}$) that has terminal time $\tilde{t}$ in between the terminal times of the ``fast" trajectory $\hat{x}$ and the ``slow" trajectory $\bar{x}$. When $\tilde{t}$ is strictly less than $t_m$, we can always find some $\nu > 0$ such that the trajectory $x_\nu$ constructed from velocity path $V_\nu$ has terminal time $\tilde{t}$. For terminal times $\tilde{t}$ greater than or equal to $t_m$, any bounded velocity path is fully stopped at some point along its trajectory, and we can choose the duration at this stopped location appropriately to verify the proposition.} \label{figure:vnu}
\end{figure}
For each $\nu \in [0,v_m]$, we define path $V_\nu \colon [p_0,0] \to [0,v_m]$ as follows:
\begin{align} \label{equation:vnu}
V_\nu(p) = 
	\begin{cases}
		\velpath{\bar{x}}(p) 	& \colon \velpath{\bar{x}}(p) > \nu \\
		\nu 				& \colon \velpath{\bar{x}}(p) \leq \nu \leq \velpath{\hat{x}}(p) \\
		\velpath{\hat{x}}(p) 	& \colon \velpath{\hat{x}}(p) < \nu.
	\end{cases}
\end{align}
See Figure~\ref{figure:vnu}.
Define $x_\nu := \mathscr{G}(V_\nu,\wait{\hat{x}},0)$.
For any $\nu \in [0,v_m]$, trajectory $x_\nu$ satisfies Inequalities~\eqref{equation:p2} and~\eqref{equation:delta2}.
We claim that for any $\tilde{t} \in [\hat{t}, t_m)$, there always exists $\tilde{\nu}$ such that $x_{\tilde{\nu}}$ has terminal time $\tilde{t}$.
Define the map $T \colon (0,v_m] \to \reals$ as follows:
\begin{align*}
	T(\nu) = \int_{p_0}^{0} \frac{1}{V_\nu(p)} dp + \sum_{\{ p \colon \wait{\hat{x}}(p) > 0\}} \wait{\hat{x}}(p).
\end{align*}
First, note that $T_\nu$ is exactly the terminal time of trajectory $x_\nu$. 
Second, $T(v_m) = \hat{t}$. 
Next, by the monotone convergence theorem, $\lim_{\nu\to 0^+} T(\nu) = t_m$.
Lastly, note that $T$ is continuous; therefore, for any $\tilde{t} \in [\hat{t},t_m)$, there exists $\tilde{\nu} \in (0,v_m]$ such that $T(\tilde{\nu}) = \tilde{t}$; and by definition of $T$, trajectory $x_{\tilde{\nu}} := \mathscr{G}(V_{\tilde{\nu}},\wait{\hat{x}},0)$ has terminal time $\tilde{t}$.

It is left to show the continuity of $T$.
Let $\epsilon > 0$ be given. Choose any $\nu \in (0,v_m]$.
We show there exists $\delta > 0$ such that for any $\omega \in (0,v_m]$ satisfying $|\nu - \omega| \leq \delta$:
\begin{align*}
	\Big\vert\int_{p_0}^{0} \frac{1}{V_\nu(p)} dp - \int_{p_0}^{0} \frac{1}{V_\omega(p)} dp\Big\vert \leq \epsilon.
\end{align*}
First, note that for any $p \in [p_0,0]$: 
\begin{align*}
	|V_\nu(p) - V_\omega(p)| \leq |\nu-\omega|.
\end{align*}
Define set $Z_\omega = \{ p \colon V_\nu(p) < \min \{ \nu,\omega\} \}$.
Next, note that for any $p \in Z_\omega$:
\begin{align*}
	V_\nu(p) = V_\omega(p).
\end{align*}
Furthermore, for any $p \in Z_\omega^\complement$:
\begin{align*}
	V_\nu(p) &\geq \min\{\nu,\omega\} \\
	V_\omega(p) &\geq \min \{ \nu,\omega\}.
\end{align*}
Then,
\begin{align*}
	&\quad\text{ } \Big\vert\int_{p_0}^{0} \frac{1}{V_\nu(p)} dp - \int_{p_0}^{0} \frac{1}{V_\omega(p)} dp \Big\vert \\
	&\leq \Big\vert \int_{Z_\omega} \frac{1}{V_\nu(p)} - \frac{1}{V_\omega(p)} dp \Big\vert + \Big\vert \int_{Z_\omega^\complement} \frac{1}{V_\nu(p)} - \frac{1}{V_\omega} dp \Big\vert \\
	&\leq \Big\vert\int_{Z_\omega^\complement} \frac{1}{V_\nu(p)} - \frac{1}{V_\omega(p)} dp \Big\vert \\
	& \leq \int_{Z_\omega^\complement} \frac{|V_\nu - V_\omega|} {V_\nu(p) V_\omega(p)} dp  
	\leq \frac{|\nu-\omega|}{ \min\{\nu^2,\omega^2 \}} \int_{Z_\omega^\complement} dp \\
	&\leq \frac{|p_0| \delta}{\min\{\nu^2,\omega^2 \}}  \leq \frac{4 |p_0| \delta}{ \nu^2} \leq \epsilon.
\end{align*}
The penultimate inequality holds, since $\omega \geq \nu/2$ for sufficiently small $\delta$, \ie, $\delta \leq \nu/2$.
The last inequality holds by choosing sufficiently small $\delta$, \ie, $\delta \leq \epsilon \nu^2 /(4|p_0|)$. \qed \\


{\it Proof of Proposition~\ref{proposition:sectionCmainprop}}
Suppose $t_f' \in [ \hat{t} ,\bar{t}]$.
From the construction of the ``high-performance" trajectory $\hat{x}$ and the ``safe" trajectory $\bar{x}$, there exists (by Proposition~\ref{proposition:boundedPaths}) a trajectory $\tilde{x}_i \in \mathscr{E}(\hat{x}_i|_{[t_0',t^\times]},t_0',t_f',\emptyset)$ satisfying for all $p \in [p^\times,0]$:
\begin{align}
	\velpath{\bar{x}_i}(p) \leq   \velpath{\tilde{x}_i} (p) \leq  \velpath{\hat{x}_i} (p) \label{equation:13-1}\\
	\wait{\bar{x}_i}(p) \geq  \wait{\tilde{x}_i}(p)  \geq  \wait{\hat{x}_i}(p). \label{equation:13-2}
\end{align}
Now, we show $\tilde{x}_i$ is safe with $x_{i-1}'$.
For all $t \in [t_0',\depart{\bar{x}_i}(t^\times)]$:
\begin{align*}
	\tilde{x}_i(t) = \bar{x}_i(t) \leq x_{i-1}'(t) - l,
\end{align*}
by construction of $\tilde{x}_i$.
To show that for all $t > \depart{\bar{x}_i}(t^\times)$,
\begin{align}
	\tilde{x}_i(t) \leq x_{i-1}'(t) - l, \label{equation:safetyEnd}
\end{align}
we argue as follows.
From Inequalities~\eqref{equation:nivvold},~\eqref{equation:niddold},~\eqref{equation:13-1}, and~\eqref{equation:13-2}, for all $p \in (p^\times,-l)$:
\begin{align}
	\velpath{\tilde{x}_i}(p) &\geq \velpath{x_{i-1}'} (p+l) \label{equation:13-3} \\
	\wait{\tilde{x}_i}(p) &\leq \wait{x_{i-1}'}(p+l). 	\label{equation:13-4}
\end{align}
Also, note:
\begin{align}
	\arrive{x_{i-1}'}(0) = \tau_{i-1}' + L/v_m = t_f' - s = \arrive{\tilde{x}_i}(l). \label{equation:13-5}
\end{align}
From Inequalities~\eqref{equation:13-3},~\eqref{equation:13-4},~\eqref{equation:13-5}, Inequality~\eqref{equation:safetyEnd} holds (by Proposition~\ref{proposition:boundedTimesFinal}) for all $t > \depart{\bar{x}_i}(t^\times)$. 
Hence, we have shown that $\tilde{x}_i \in \mathscr{E}_i'$.

Now, suppose $t_f' > \bar{t}$. 
We proceed by constructing a ``slower" trajectory $\bar{\bar{x}}_i$ than trajectory $\bar{x}_i$ in order to bound the terminal time $t_f'$, as in Equation~\eqref{equation:terminaltimebound}, and afterwards apply Proposition~\ref{proposition:boundedPaths} to show existence of a trajectory $\tilde{x}_i$ with terminal time $t_f'$.
We then show $\tilde{x}_i$ is safe with $x_{i-1}'$.

Borrowing the language of Proposition~\ref{proposition:boundedPaths}, $\bar{x}_i$ acts as the ``fast" trajectory, and $\bar{\bar{x}}_i$ acts as the ``slow" trajectory.
In order to construct $\bar{\bar{x}}_i$, we first find the earliest intersection point of velocity path $\velpath{\bar{x}_i}$ and the blue curve depicted in Figure~\ref{figure:ubb}.
Formerly, we define:
\begin{align*}
	p^+ &= \inf \Bigg\{  p \colon \velpath{\bar{x}_i}(p) = \sqrt{p + \frac{v_m^2}{2a_m} + (\nu_i-1)l } \Bigg\} \\
	v^+ &= \velpath{\bar{x}_i} (p^+).
\end{align*}
We use notation developed in Equations~\eqref{equation:notation13}.
Now, consider trajectory $\bar{\bar{x}}_i := \mathscr{T}(z_i(t_0'),\bar{\bar{u}}_i)$, where control action $\bar{\bar{u}}_i$ is defined as follows (depicted in Figure~\ref{figure:ubb}):
\begin{align} \label{eqn:ubarbar}
\bar{\bar{u}}_i(t) =
	\begin{cases}
		\bar{u}_i(t)  	&\colon t \in [t_0',t_i^c) \\
		-a_m 		&\colon t \in [t_i^c,t_s) \\
		0			&\colon t \in [t_s, t_s + \sigma) \\
		a_m  		&\colon t \in [t_s + \sigma, t^+] \\
		\bar{u}_i( t - t^+ + \arrive{\bar{x}_i}(p^+)) 	&\colon t \in [t^+,\bar{\bar{t}}],
	\end{cases}
\end{align}
where
\begin{align*}
	t_s 			&:= 	t_i^c + v_i^c/a_m \\
	t^+		 	&:= 	t_s + \sigma + v^+/a_m \\
	\bar{\bar{t}} 	&:= 	t^+ + \bar{t} - \arrive{\bar{x}_i}(p^+) \\
	\sigma 		&:= 	\max \big\{ 0 , t_f' - \big(t_s + v^+ / a_m + \bar{t} - \arrive{\bar{x}_i}(p^+) \big) \big\}.
\end{align*}
\begin{figure}
	\centerline{\includegraphics[clip=true, trim=2.1in 6.5in 4.1in 3.4in, width=0.5\textwidth]{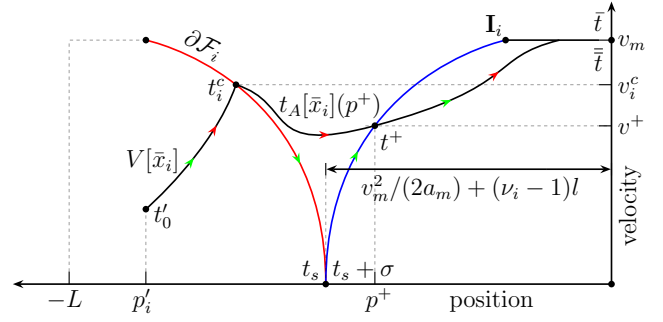}}
	\caption{``Slow" velocity path $\velpath{\bar{\bar{x}}_i}$ as generated by control $\bar{\bar{u}}_i$ defined in Equation~\eqref{eqn:ubarbar} is depicted with green arrows; ``fast" velocity path $\velpath{\bar{x}_i}$ is depicted with red arrows. Note that during ``slow" trajectory $\bar{\bar{x}}_i$ vehicle halts for an amount of time $\sigma$ at a distance $v_m^2 / (2 a_m) + (\nu_i-1)l$ from the intersection region. Trajectory $\bar{\bar{x}}_i$ arrives at this position at time $t_s$ and departs this position at time $t_s + \sigma$. Time $t^+$ is the time that $\bar{\bar{x}}_i$ arrives at state $(p^+,v^+)$. Trajectory $\bar{\bar{x}}_i$ reaches the intersection region, \ie, state $(0,v_m)$, at time $\bar{\bar{t}}$. ``Fast" trajectory $\bar{x}_i$ arrives at state $(p^+,v^+)$ at time $\arrive{\bar{x}_i}(p^+)$, and arrives at the intersection region, \ie, state $(0,v_m)$, at time $\bar{t}$.} \label{figure:ubb}
\end{figure}
By construction of $\sigma$:
\begin{align}
	\bar{\bar{t}} \geq t_f' \geq \bar{t}. \label{equation:terminaltimebound}
\end{align}
By construction of $\bar{\bar{u}}_i$, for all $p$:
\begin{align*}
	\velpath{\bar{\bar{x}}_i}(p) &\leq \velpath{\bar{x}_i}(p) \\
	\wait{\bar{\bar{x}}_i}(p) 	&\geq \wait{\bar{x}_i}(p).
\end{align*}
Since Inequalities~\eqref{equation:p1} and~\eqref{equation:delta1} hold for trajectories $\bar{\bar{x}}$ and $\bar{x}$, there exists (by Proposition~\ref{proposition:boundedPaths}) a trajectory $\tilde{x}_i \in \mathscr{C}(z_i(t_0'),t_0',t_f',\emptyset)$ satisfying Inequalities~\eqref{equation:pp} and~\eqref{equation:deldel} for trajectories $\bar{\bar{x}}_i$, $\tilde{x}_i$, and $\bar{x}_i$, \ie, for all $p$:
\begin{align*}
	\velpath{\bar{\bar{x}}_i}(p) &\leq 	\velpath{\tilde{x}_i}(p) \leq 		\velpath{\bar{x}_i}(p) \\
	\wait{\bar{\bar{x}}_i}(p) 	&\geq 	\wait{\tilde{x}_i}(p) 	 \geq 	\wait{\bar{x}_i}(p).
\end{align*}
Note that Inequality~\eqref{equation:dd} also holds for trajectories $\bar{\bar{x}}_i$, $\tilde{x}_i$, and $\bar{x}_i$, \ie,
\begin{align*}
	\depart{\bar{x}_i}(p_i')  \geq \depart{\tilde{x}_i}(p_i') \geq \depart{x_i}(p_i'),
\end{align*}
where $p_i' := x_i(t_0')$.
Since Inequalities~\eqref{equation:pp},~\eqref{equation:deldel}, and~\eqref{equation:dd} hold, then (by Proposition~\ref{proposition:boundedTimes}) for all $t \geq t_0'$:
\begin{align*}
	\bar{\bar{x}}_i(t) \leq \tilde{x}_i(t) \leq \bar{x}_i(t).
\end{align*}
Safety of $\tilde{x}_i$ and $x_{i-1}'$ immediately follows.
Also, $\tilde{x}_i(t) = x_i(t)$ for all $t \in [t_0',t_c[x_i]]$.
Thus, $\tilde{x}_i \in \mathscr{E}_i'$. \qed\\


{\it    Proof of Proposition~\ref{proposition:boundedTimes} }
\begin{figure}
	\centerline{\includegraphics[clip=true, trim=2in 6.6in 4.2in 3.3in, width=0.5\textwidth]{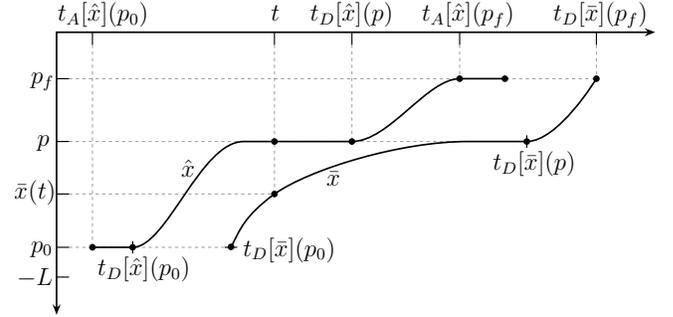}}
	\caption{Notation for proof of Proposition~\ref{proposition:boundedTimes} on a position-time plot. Note that trajectories $\hat{x}$ and $\bar{x}$ are both monotonically increasing functions. Note that Inequality~\eqref{equation:dd} is illustrated clearly: the departure time of $\bar{x}$ from position $p_0$ is later than the departure time of $\hat{x}$ from the same position, \ie, $\depart{\bar{x}}(p_0) \geq \depart{\hat{x}}(p_0)$. We show that $\bar{x}(t) \leq \hat{x}(t)$ for all time $t$ from the arrival of $\hat{x}$ at position $p_0$ until the departure of $\bar{x}$ from position $p_f$, \ie, for all $t \in [\arrive{\hat{x}}(p_0),\depart{\bar{x}}(p_f)]$.} \label{figure:prop12}
\end{figure}
See Figure~\ref{figure:prop12}.
Let
\begin{align*}
	t \in  [\arrive{\hat{x}}(p_0), \arrive{\hat{x}}(p_f)).
\end{align*}
We can find $p \in [p_0,p_f)$ such that $\hat{x}(t) = p$. 
By Inequalities~\eqref{equation:pp},~\eqref{equation:deldel}, and~\eqref{equation:dd}:
\begin{align*}
	t &\leq \depart{\hat{x}}(p) \\
	  &= \depart{\hat{x}}(p_0) + \int_{p_0}^{p} \frac{1}{\velpath{\hat{x}}(q)} dq + \sum_{ \{ q \in (p_0 , p] \colon \wait{\hat{x}}(q) > 0 \} } \wait{\hat{x}} (q) \\
	  &\leq
	  \depart{\bar{x}}(p_0) + \int_{p_0}^{p} \frac{1}{\velpath{\bar{x}}(q)} dq + \sum_{ \{ q \in (p_0 , p] \colon \wait{\bar{x}}(q) > 0 \} } \wait{\bar{x}} (q) \\ 
	  &= \depart{\bar{x}}(p).
\end{align*}
Since $t \leq \depart{\bar{x}}(p)$ and $\bar{x}$ is monotonically increasing:
\begin{align*}
	\bar{x}(t) \leq \bar{x}(\depart{\bar{x}}(p)) = p = \hat{x}(t).
\end{align*}
It remains to show the above inequality holds for
\begin{align*}
	t \in [ \arrive{\hat{x}}(p_f) , \depart{\bar{x}}(p_f) ].
\end{align*}
Note $\bar{x}(t) \leq p_f = \hat{x}(\arrive{\hat{x}}(p_f)) \leq \hat{x}(t)$. \qed


\begin{proposition} \label{proposition:boundedTimesFinal}
Suppose there exist a ``fast" trajectory $\hat{x} \in \mathscr{C}(\mathbf{\hat{z}}_0,t_0,\hat{t},\emptyset)$ and a ``slow" trajectory $\bar{x} \in \mathscr{C}(\mathbf{\bar{z}}_0,t_0,\bar{t},\emptyset)$, such that for all $p \in (p_0,p_f)$:
\begin{align}
	\velpath{\bar{x}}(p) 		&\leq 	\velpath{\hat{x}}(p) 	\label{equation:pp2}\\
	\wait{\bar{x}}(p)			&\geq 	\wait{\hat{x}}(p) 	\label{equation:deldel2} \\
	\arrive{\bar{x}}(p_f) 		&\leq  	\arrive{\hat{x}}(p_f). \label{equation:dd2} 
\end{align}
Then, for all $t \in  [\arrive{\bar{x}}(p_0), \depart{\hat{x}}(p_f)]$:
\begin{align}
	\bar{x}(t) \geq \hat{x}(t). \label{equation:idk}
\end{align}
\end{proposition}
\begin{figure}
	\centerline{\includegraphics[clip=true, trim=2in 6.6in 4.2in 3.3in, width=0.5\textwidth]{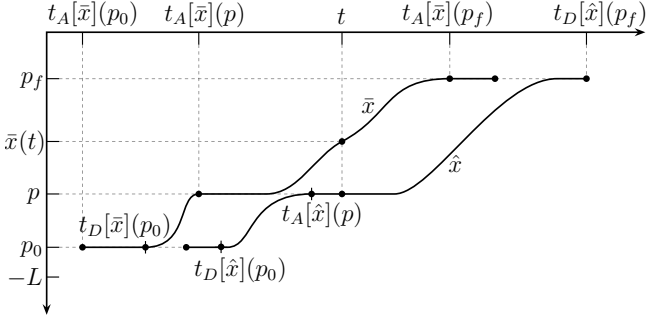}}
	\caption{Notation for proof of Proposition~\ref{proposition:boundedTimesFinal} on a position-time plot. Note that trajectories $\hat{x}$ and $\bar{x}$ are both monotonically increasing functions. Note that Inequality~\eqref{equation:dd2} is illustrated clearly: the arrival time of $\bar{x}$ at position $p_f$ is earlier than the arrival time of $\hat{x}$ at the same position. We show that $\bar{x}(t) \geq \hat{x}(t)$ for all time $t$ from the arrival of $\bar{x}$ at position $p_0$ until the departure of $\hat{x}$ from position $p_f$, \ie, for all $t \in [\arrive{\hat{x}}(p_0),\depart{\bar{x}}(p_f)]$. This proposition complements Proposition~\ref{proposition:boundedTimes}, in that it applies its argument in reverse.} \label{figure:prop14}
\end{figure}

\begin{proof}
Consider Figure~\ref{figure:prop14}.
Let $t \in  (\depart{\hat{x}}(p_0), \depart{\hat{x}}(p_f)]$.
Find $p \in (p_0,p_f]$ such that $\hat{x}(t) = p$.
By Inequalities~\eqref{equation:pp2},~\eqref{equation:deldel2}, and~\eqref{equation:dd2}:
\begin{align*}
	t &\geq \arrive{\hat{x}}(p) \\
	  &= \arrive{\hat{x}}(p_f) - \int_{p_0}^{p_f} \frac{1}{\velpath{\hat{x}}(q)} dq - \sum_{ \{ q \in [p, p_f) \colon \wait{\hat{x}}(q) > 0 \} } \wait{\hat{x}} (q) \\
	  &\geq
	  \arrive{\bar{x}}(p_f) - \int_{p_0}^{p_f} \frac{1}{\velpath{\bar{x}}(q)} dq - \sum_{ \{ q \in [p , p_f) \colon \wait{\bar{x}}(q) > 0 \} } \wait{\bar{x}} (q) \\ 
	  &= \arrive{\bar{x}}(p).
\end{align*}
Since $t \geq \arrive{\bar{x}}(p)$ and $\bar{x}$ is monotonically increasing:
\begin{align*}
	\bar{x}(t) \geq \bar{x}(\arrive{\bar{x}}(p)) = p = \hat{x}(t).
\end{align*}
It remains to show the above inequality holds for
\begin{align*}
	t \in [\arrive{\bar{x}}(p_0),\depart{\hat{x}}(p_0)].
\end{align*}
Note
\begin{align*}
	\bar{x}(t) \geq p_0 = \hat{x}(\depart{\hat{x}}(p_0)) \geq \hat{x}(t),
\end{align*}
where the last inequality holds since $\hat{x}$ is monotonically increasing.
\end{proof}

\newpage
\subsection{Proof of Lemma~\ref{lemma:well-defined}} \label{subsection:D}
In this section, we first show that {\it if} the set $\mathscr{E}_i'$ is non-empty, then the optimization problem in the ${\tt MotionSynthesize}$ procedure (see Equation~\eqref{eq:procedure}) attains its minimum. We use Filippov's theorem in this step. See Theorem~\ref{theorem:filippov}.

Next, we show that the minimizing trajectory, \ie, the updated trajectory $x_i'$, is contained in $\mathscr{E}_i'$.
More specifically, we show that the following inequalities hold for any $\tilde{x} \in \mathscr{E}_i'$ and all time $t$:
\begin{align}
	\tilde{x}(t) \leq x_i'(t) \leq x_i(t), \label{equation:d1}
\end{align}
where $x_i$ is the previous trajectory of vehicle $i$.
%
%
From the above inequalites, $x_i' \in \mathscr{E}_i'$ directly follows.

The first inequality holds by the following argument.
For all $x \in \mathscr{C}_i'$ and for all time $t$:
\begin{align}
	x(t) \leq x_i'(t), \label{equation:d2}
\end{align}
and then noting $\tilde{x} \in \mathscr{E}_i' \subseteq \mathscr{C}_i'$.
The second inequality holds from the following.
For all $x \in \mathscr{C}_i'$ and for all time $t$:
\begin{align}
	x(t) \leq x_i(t), \label{equation:d3}
\end{align}
and then noting $x_i' \in \mathscr{C}_i'$.
We show Inequality~\eqref{equation:d2} holds with a proof by contradiction.
Then, we show Inequality~\eqref{equation:d3} holds by induction on $\nu_i$.

To show the existence of a minimum in the ${\tt MotionSynthesize}$ procedure, we employ the following theorem from optimal control theory.
\begin{theorem} \textbf{Filippov~\cite{agrachev2001}} \label{theorem:filippov}
Let $\mathcal{U}$ be the set of all measurable functions from set $[t_0,t_f]$ to $U$, where $U \subseteq \reals^m$ is a compact set.
Let there exist a compact $K \subseteq \reals^n$ such that for all $\mathbf{q} \not\in K$ and for all $u \in U$:
\begin{align}
	f(\mathbf{q},u) = \mathbf{0}. \label{equation:filippov0}
\end{align}
For each $\mathbf{q} \in K$, let
\begin{align*}
	f(\mathbf{q},U) := \{ f(\mathbf{q},u) \colon u \in U\} \subseteq \reals^n
\end{align*}
be convex.
Then, for each initial state $\mathbf{q_0} \in K$ and any terminal time $t_f > 0$, the attainable set
\begin{align*}
	\mathcal{A}
	:= \{ q_u(t_f) \colon u \in \mathcal{U},\,\, q_u(t_0) = \mathbf{q}_0 \} \subseteq \reals^n
\end{align*}
is compact, where $q_u \colon [t_0,t_f] \to \reals^n$ satisfies $\dot{q}_u(\tau) = f(q_u(\tau),u(\tau))$ for all $\tau \in [t_0,t_f]$ and $q_u(t_0) = \mathbf{q_0}$.

Define $\mathcal{Q} := \{ q_u \colon u \in \mathcal{U}, \,\, q_u(0) = \mathbf{q_0}\}$ as the set of all state histories $q_u$ from initial state $\mathbf{q_0}$ generated by a control $u \in \mathcal{U}$.
Then, for any sequence $\{ q^n \}_{n=1}^{\infty}$ of state histories where $q^n \in \mathcal{Q}$, there exists a subsequence $\{q^{n_k} \}_{k=1}^\infty$ of state histories that converges uniformly to some $\tilde{q} \in \mathcal{Q}$. 
\end{theorem}

Consider the following system with state vector $q := [q_1, q_2, q_3]^T$:
\begin{align*}
\dot{q}(t) = 
	\begin{cases}
		[q_2(t),q_3(t),p(u(t))]^T	&\colon q \in K \\
		\mathbf{0}			&\colon q \not\in K,
	\end{cases}
\end{align*}
where $K := [-L(t_f'-t_0'),0] \times [-L,0] \times [0, v_m]$ and $p \colon U \to U$ is defined as
\begin{align*}
p(u) = 
	\begin{cases}
		\max \{ u, 0 \} 	&\colon q_3 = 0 \\
		\min \{ u, 0 \}	&\colon q_3 = v_m \\
		u			&\colon q_3 \in (0,v_m).
	\end{cases}
\end{align*}
First, note that $U = [-a_m,a_m]$ is compact.
Next, note that Equation~\eqref{equation:filippov0} is satisfied for all $q \not\in K$ and for all $u \in U$.
Lastly, note that $f(q,U)$ is convex for all $q \in K$.
Consider initial state (column vector) $\mathbf{q_0} := [0,z_i(t_0')]$.
By Theorem~\ref{theorem:filippov}, $\mathcal{A}$ is compact.

Define the following set of state histories:
\begin{align*}
	\mathcal{Q}[x_{i-1}'] &:= \{ q \in \mathcal{Q} \,\, | \,\, \forall t,  q_2(t) \leq x_{i-1}'(t) - l ; \\
	\quad& q_2(t_f') = 0 ; q_3(t_f')=v_m \}.
\end{align*}
Next, define the following attainable set:
\begin{align*}
	\mathcal{A}[x_{i-1}'] &:= \{ q(t_f') \colon q \in \mathcal{Q}[x_{i-1}'] \} \subseteq \reals^3.
\end{align*}
Now, $\mathcal{A}[x_{i-1}']$ is also compact. 
First, note that since $\mathcal{A}[x_{i-1}'] \subseteq \mathcal{A}$, then $\mathcal{A}[x_{i-1}']$ is bounded.
We, now show $\mathcal{A}[x_{i-1}']$ is also closed.

Let $\{ a^n \}$ by a convergent sequence with elements $a^n \in \mathcal{A}[x_{i-1}'] \subseteq \reals^3$.
We show that the limit of sequence $\{ a^n \}$ is also contained in $\mathcal{A}[x_{i-1}']$.
By definition of $\mathcal{A}[x_{i-1}']$, we can embed convergent sequence $\{ a^n \}$ into some sequence $\{ q^n \}$ of state histories with elements $q^n \in \mathcal{Q}[x_{i-1}']$.
Next, we invoke Theorem~\ref{theorem:filippov} to find a convergent subsequence $q^{n_k}$, which converges uniformly to some $\tilde{q} \in \mathcal{Q}$.
Next, note for all $t$ and for all $n$:
\begin{align*}
 	q^{n_k}_2(t) &\leq x_{i-1}'(t) - l \\
	\lim_{k\to\infty} q_2^{n_k}(t) &\leq  \lim_{k\to\infty} x_{i-1}'(t) - l \\
	\tilde{q}_2(t) &\leq x_{i-1}'(t) - l.
\end{align*}
Thus, $\tilde{q} \in \mathcal{Q}[x_{i-1}']$.
Moreover, by our embedding of $\{ a^n \}$ into $\{ q^n \}$, we have $\lim_{n\to\infty} a^n = \tilde{q}(t_f') \in \mathcal{A}[x_{i-1}']$.
Hence, $\mathcal{A}[x_{i-1}']$ is also closed.
Hence, we have shown that $\mathcal{A}[x_{i-1}']$ is compact.

Now, consider optimization problem
\begin{align} \label{opt:opt1}
	\min_{ \mathbf{a} \in \mathcal{A}[x_{i-1}'] } [-1,0,0]^T \mathbf{a}.
\end{align}
Optimization problem~\eqref{opt:opt1} is nontrivial, since $\mathscr{C}_i' \not= \emptyset$ implies $\mathcal{A}[x_{i-1}'] \not= \emptyset$.
Also, note that the cost function is linear and $\mathcal{A}[x_{i-1}']$ is compact; hence, a global minimum exists.
Consider optimization problem
\begin{align} \label{opt:opt2}
	\min_{q \in \mathcal{Q}[x_{i-1}'] } -q_1(t_f')
\end{align}
and optimization problem
\begin{align} \label{opt:opt3}
	\min_{x \in \mathscr{C}_i'} \int_{t_0'}^{t_f'} |x(t)| dt.
\end{align}
By construction, optimization problems~\eqref{opt:opt1} and~\eqref{opt:opt2} are equivalent.
Next, notice optimization problems~\eqref{opt:opt2} and~\eqref{opt:opt3} are equivalent.
For any $x \in \mathscr{C}_i'$, define $q_1 \colon [t_0',t_f'] \to \reals$ as $q_1(t) := \int_{t_0'}^{t} x(\tau) d\tau$.
Then, $[q_1,x,\dot{x}]^T \in \mathcal{Q}$ and $-q_1(t_f') = \int_{t_0'}^{t_f'} |x(t)| dt$.
Conversely, note that for any $q \in \mathcal{Q}$, we have position history $q_2 \in \mathscr{C}_i'$.
Hence, the solution to optimization problem~\eqref{opt:opt3} exists.

Thus far, we have shown $x_i'$, as defined by the ${\tt MotionSynthesize}$ procedure, exists.
It remains to show Inequalities~\eqref{equation:d2} and~\eqref{equation:d3} hold.
First we show Inequality~\eqref{equation:d2} with a proof by contradiction.
We use the following lemma:
\begin{lemma} \label{lemma:minimumTrajectory}
Let $x \in \mathscr{C}(\mathbf{z}_0,t_0,t_f,x_s)$ and $x' \in \mathscr{C}(\mathbf{z}_0,t_0,t_f',x_s)$ satisfy
\begin{align*}
	t_f' \geq t_f.
\end{align*}
Suppose there exists time $\bar{t}$ satisfying
\begin{align*}
	x(\bar{t}) < x'(\bar{t}).
\end{align*}
Then, there exists trajectory $y \in \mathscr{C}(\mathbf{z}_0,t_0,t_f,x_s)$ satisfying for all $t \in [0,t_f]$:
\begin{align}
	x(t) 		&\leq y(t) \label{equation:higher} \\
	x(\bar{t}) 	&< y(\bar{t}). \label{equation:highest}
\end{align}
\end{lemma}
Let $x \in \mathscr{C}_i'$.
Suppose there exists $t_a \in [t_0',t_f']$ such that $x(t_a) > x_i'(t_a)$.
By Lemma~\ref{lemma:minimumTrajectory}, there exists trajectory $y \in \mathscr{C}_i'$ satisfying for all time $t$:
\begin{align*}
	y(t) 		&\geq x_i'(t) \\
	y(t_a) 	&> x_i'(t_a).
\end{align*}
By continuity of $y$:
\begin{align*}
	\int_{t_0'}^{t_f'} |y(t)| dt < \int_{t_0'}^{t_f'} |x_i'(t)| dt,
\end{align*}
contradicting construction of $x_i'$.
Thus, Inequality~\eqref{equation:d2} holds.

We conclude by showing Inequality~\eqref{equation:d3} holds, by induction on $\nu_i$.
Let $x \in \mathscr{C}_i'$.
Suppose $\nu_i = 1$.
Then, by Lemma~\ref{lemma:committed}, $x_{i-1}' = x_{i-1}|_{t \geq t_0'}$, from which $x_i$ is safe with the same trajectory as $x$.
Thus, we are almost ready to apply Proposition~\ref{lemma:minimumTrajectory}.
Suppose there exists $t_a$ satisfying $x_i(t_a) < x(t_a)$.
Then, by Lemma~\ref{lemma:minimumTrajectory}, there exists trajectory $y \in \mathscr{C}(z_i(t_0'),t_0',t_f,x_{i-1})$ satisfying for all $t$:
\begin{align*}
	|y(t)| 		&\leq |x_i(t)| \\
	|y(t_a)| 	&< |x_i(t_a)|.
\end{align*}
By continuity of $y$:
\begin{align*}
	\int_{t_0'}^{t_f} |y(t)| dt < \int_{t_0'}^{t_f} |x_i(t)| dt,
\end{align*}
contradicting construction of $x_i$.
Thus, Inequality~\eqref{equation:d3} holds for $\nu_i = 1$.

Now, assume $\nu_i > 1$.
By induction, we assume for any $x \in \mathscr{C}_{i-1}'$ and for all $t$:
\begin{align*}
	x(t) \leq x_{i-1}(t).
\end{align*}
Recall trajectory $x_{i-1}'$ is the optimal trajectory for vehicle $i-1$ given by the ${\tt MotionSynthesize}$ procedure at time $t_0$.
Thus, since $x_{i-1}' \in \mathscr{C}_{i-1}'$, we have
\begin{align}
	x_{i-1}'(t) \leq x_{i-1}(t). \label{equation:inductionD}
\end{align}
Let $x \in \mathscr{C}_i'$.
We show for all $t$:
\begin{align*}
	x(t) \leq x_i(t).
\end{align*}
Since trajectory $x$ is contained in $\mathscr{C}_i'$, trajectories $x$ and $x_{i-1}'$ are safe.
Furthermore, $x$ must also be safe with $x_{i-1}$.
Note that for all time $t$:
\begin{align*}
	x(t) 	&\leq x_{i-1}'(t) - l \\
		&\leq x_{i-1}(t) - l,
\end{align*}
where the second inequality holds from the induction step in Inequality~\eqref{equation:inductionD}. 
Now, suppose there exists $t_a$ such that $x(t_a) > x_i(t_a)$. 
Recall that terminal time of $x \in \mathscr{C}_i'$ is equal to $t_f'$; and the terminal time of $x_i$ is equal to $t_f$.
Since $t_f' \geq t_f$, we can apply Lemma~\ref{lemma:minimumTrajectory} to create trajectory $y \in \mathscr{C}(z_i(t_0'),t_0',t_f,x_{i-1})$ satisfying for all $t$:
\begin{align*}
	|y(t)| 		&\leq |x_i(t)| \\
	|y(t_a)| 	&< |x_i(t_a)|.
\end{align*}
By continuity of $y$:
\begin{align*}
	\int_{t_0'}^{t_f} |y(t)| dt < \int_{t_0'}^{t_f} |x_i(t)| dt,
\end{align*}
contradicting the optimality of $x_i$.
Hence, Inequality~\eqref{equation:d3} is established.


\subsection*{Proofs used in Section~\ref{subsection:D}}
{\it Proof of Lemma~\ref{lemma:minimumTrajectory}:}
%
Define
\begin{align*}
	\check{t} 		:= \sup \{ t \leq \bar{t} \colon x'(t) = x(t) \}  \\
	\hat{t} 	:= \inf \{ t \geq \bar{t} \colon x'(t) = x(t)\}.
\end{align*}
From these definitions, we have
\begin{align}
	x'(\check{t}) 	&= x(\check{t})	\label{equation:16-1}	\\
	x(\hat{t}) 	&= x'(\hat{t}). \label{equation:16-2}
\end{align}

%
{\it Case 1.}
First, suppose
\begin{align}
	\dot{x}'(\check{t}) 	&= \dot{x}(\check{t})	\label{equation:16-3}	\\
	\dot{x}(\hat{t}) 	&= \dot{x}'(\hat{t}). \label{equation:16-4}
\end{align}
Consider Figure~\ref{figure:ycase1}. Define trajectory $y_1$:
\begin{align} \label{equation:y1}
y_1(t) = 
	\begin{cases}
		x(t) 	&\colon 	t \in [t_0,\check{t})\\
		x'(t)	&\colon 	t \in [\check{t},\hat{t}) \\
		x(t)	&\colon	t \in [\hat{t},t_f].
	\end{cases}
\end{align}
\begin{figure}
	\centerline{\includegraphics[clip=true, trim=2in 6.6in 4.2in 3.3in, width=0.5\textwidth]{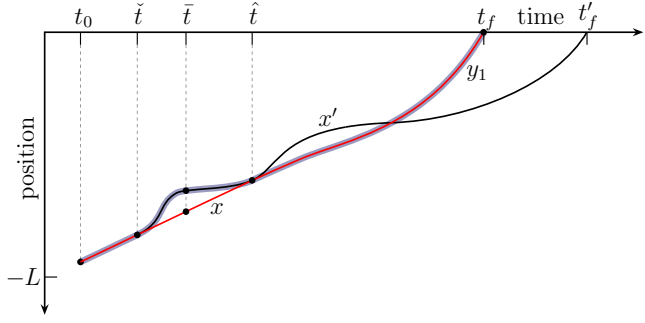}}
	\caption{Notation for Case 1 in the proof of Lemma~\ref{lemma:minimumTrajectory} used in Section D. Trajectory $y_1$ as defined in Equation~\eqref{equation:y1} is depicted in shaded blue. Trajectory $x$ is shown in red, and trajectory $x'$ in black. The assumptions in this scenario, given in Equations~\eqref{equation:16-3} and~\eqref{equation:16-4}, are depicted clearly. The result of Lemma~\ref{lemma:minimumTrajectory} is shown, \ie, $y(t) \geq x(t)$ for all time $t$ and also $y(\bar{t}) > x(\bar{t})$.} \label{figure:ycase1}
\end{figure}
%
Note that trajectory $y_1$ satisfies Inequalities~\eqref{equation:higher} and~\eqref{equation:highest}.

{\it Case 2.}
Now, suppose
\begin{align}
	\dot{x}'(\check{t}) 	&= \dot{x}(\check{t})	\label{equation:16-5}	\\
	\dot{x}(\hat{t}) 		&\not= \dot{x}'(\hat{t}). \label{equation:16-6}
\end{align}
Recall for any $t$ we define $z'(t) := (x'(t),\dot{x}'(t))$ and $z(t) := (x(t),\dot{x}(t))$.
From Inequalities~\eqref{equation:16-2} and~\eqref{equation:16-6}, by Lemma~\ref{lemma:accelerationTrajectoryPatch}, there exists trajectory $\tilde{x} \in \mathscr{C}(z'(t_1),t_2,\emptyset,x_s)$ such that for all $t \in [t_1,t_2]$:
\begin{align*}
	\tilde{x}(t) 		&\geq \max \{x'(t), x(t) \}  \\
	\tilde{x}(t_2)	 &= x(t_2) \\
	\dot{\tilde{x}}(t_1) &= \dot{x}(t_2)
\end{align*}
where
\begin{align*}
	t_1 	&:= \sup \{t \leq \hat{t} \colon \exists\, \tau, \, z'(t) = z(\tau) \} \\
	t_2 	&:= \inf \{ t \geq \hat{t} \colon \exists\, \tau, \, z'(\tau) = z(t) \}.
\end{align*}
%
Note that $t_1 \geq \check{t}$.
Now, we construct trajectory $y_2$:
\begin{align} \label{equation:y2}
y_2(t) = 
	\begin{cases}
		x(t)	 	&\colon 	t \in [t_0,\check{t})\\
		x'(t)		&\colon 	t \in [\check{t},t_1) \\
		\tilde{x}(t)	&\colon	t \in [t_1,t_2) \\
		x(t)		&\colon	t \in [t_2,t_f].
	\end{cases}
\end{align}
\begin{figure}
	\centerline{\includegraphics[clip=true, trim=2in 6.6in 4.2in 3.3in, width=0.5\textwidth]{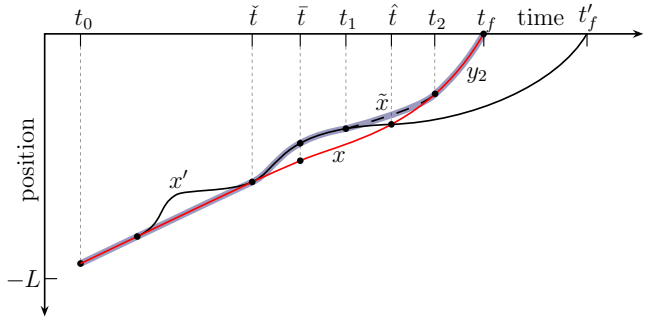}}
	\caption{Notation for Case 2 in the proof of Lemma~\ref{lemma:minimumTrajectory} used in Section D. Trajectory $y_2$ as defined in Equation~\eqref{equation:y2} is depicted in shaded blue. Trajectory $x$ is shown in red, and trajectory $x'$ in black. Trajectory $\tilde{x}$ is the dashed black line, and exists from a direct application of Lemma~\ref{lemma:accelerationTrajectoryPatch}.  The assumptions of Case 2, given in Equations~\eqref{equation:16-5} and~\eqref{equation:16-6}, are depicted. The difference between Cases 1 and 2, is that we now assume $\dot{x}(\hat{t}) \not= \dot{x}'(\hat{t})$. The result of Lemma~\ref{lemma:minimumTrajectory} is shown, \ie, $y(t) \geq x(t)$ for all time $t$ and also $y(\bar{t}) > x(\bar{t})$.} \label{figure:ycase2}
\end{figure}
%
Note that trajectory $y_2$ satisfies Inequalities~\eqref{equation:higher} and~\eqref{equation:highest}.

Case 1 and 2 together exhaust the scenario $\dot{x}'(\check{t}) = \dot{x}(\check{t})$. See Equations~\eqref{equation:16-1} and~\eqref{equation:16-5}. In the remaining cases, we consider the following scenario:
\begin{align*}
	\dot{x}'(\check{t}) &\not= \dot{x}(\check{t}).
\end{align*}
Then, from the definition of $\bar{t}$ and $\check{t}$:
\begin{align}
	\dot{x}(\check{t}) < \dot{x}'(\check{t}). \label{equation:16-7}
\end{align}
From Inequalities~\eqref{equation:16-1} and~\eqref{equation:16-7}, by Lemma~\ref{lemma:accelerationTrajectoryPatch}, there exists trajectory $\tilde{x}_1 \in \mathscr{C}(z(t_3),t_3,\emptyset,x_s)$ satisfying for all $t \in [t_3,t_4]$: 
\begin{align*}
	\tilde{x}_1(t) 	&\geq \max\{x(t) , x'(t)\} \\
	\tilde{x}_1(t_4) 	 &= x'(t_4) \\
	\dot{\tilde{x}}_1(t_4) &= \dot{x}'(t_4),
\end{align*}
where
\begin{align*}
	t_3 &:=  	\sup \{t \leq \check{t} \colon \exists\, \tau, \, z(t) = z'(\tau) \} \\
	t_4 &:=	 \inf \{ t \geq \check{t} \colon \exists\, \tau, \, z(\tau) = z'(t)\}.
\end{align*}
Inequality~\eqref{equation:16-7} and the definition of $t_4$ imply that $x'(t_4) \geq x(t_4)$, from which we define
\begin{align*}
	t_6 := \inf \{ t \geq t_4 \colon x(t) = x'(t) \}.
\end{align*}
Clearly, $t_6 \geq t_4$. 

{\it Case 3.}
Suppose
\begin{align} \label{equation:t6}
	\dot{x}(t_6) &= \dot{x}'(t_6).
\end{align}
Consider Figure~\ref{figure:ycase3}. We construct trajectory $y_3$:
\begin{align} \label{equation:y3}
y_3(t) = 
	\begin{cases}
		x(t)	 	&\colon 	[t_0,t_3)\\
		\tilde{x}_1(t)&\colon	[t_3,t_4) \\
		x'(t)		&\colon 	[t_4,t_6) \\
		x(t)		&\colon	[t_6,t_f].
	\end{cases}
\end{align}
\begin{figure}
	\centerline{\includegraphics[clip=true, trim=2in 6.6in 4.2in 3.3in, width=0.5\textwidth]{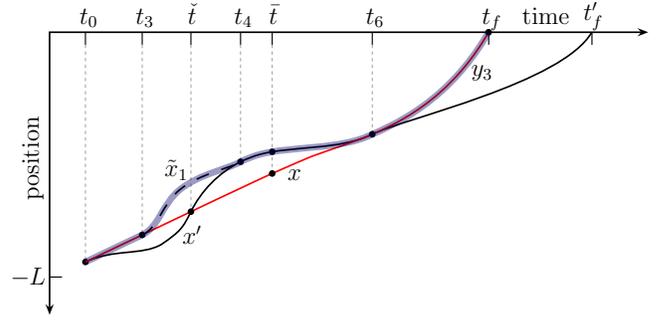}}
	\caption{Notation for Case 3 in the proof of Lemma~\ref{lemma:minimumTrajectory} used in Section D. Trajectory $y_3$ as defined in Equation~\eqref{equation:y3} is depicted in shaded blue. Trajectory $x$ is shown in red, and trajectory $x'$ in black. Trajectory $\tilde{x}_1$ is the dashed black line, and exists from a direct application of Lemma~\ref{lemma:accelerationTrajectoryPatch}.  The assumptions used in Case 3, given in Equations~\eqref{equation:16-7} and~\eqref{equation:t6}, are depicted. The result of Lemma~\ref{lemma:minimumTrajectory} is shown, \ie, $y(t) \geq x(t)$ for all time $t$ and also $y(\bar{t}) > x(\bar{t})$.} \label{figure:ycase3}
\end{figure}
%
Note that trajectory $y_3$ satisfies Inequalities~\eqref{equation:higher} and~\eqref{equation:highest}.

{\it Case 4.}
Now, suppose
\begin{align*}
	\dot{x}(t_6) \not= \dot{x}'(t_6).
\end{align*}
By Inequality~\eqref{equation:16-7} and the definition of $t_6$:
\begin{align} \label{equation:case4}
	\dot{x}'(t_6) < \dot{x}(t_6).
\end{align}
By Lemma~\ref{lemma:accelerationTrajectoryPatch}, there exists trajectory $\tilde{x}_2 \in \mathscr{C}(z'(t_5),t_5,\emptyset ,x_s)$ satisfying for all $t \in [t_5,t_7]$:
\begin{align*} 
	\tilde{x}_2(t) 	&\geq \max\{x(t), x'(t) \}  \\
	\tilde{x}_2(t_7) 	&= x(t_7) \\
	\dot{\tilde{x}}_2(t_7) 	&= \dot{x}(t_7)
\end{align*}
where
\begin{align*}	
	t_5 &:= \sup \{ t < t_6 \colon \exists\, \tau, \, z(\tau) = z'(t) \} \\
	t_7 &:= \inf \{ t > t_6 \colon \exists\, \tau, \, z(t) = z'(\tau) \}.
\end{align*}
Note $t_4 \geq t_3$.
Next, note that $x(t) \leq x'(t)$ for all $t \in [t_3,t_4]$.
Consider Figure~\ref{figure:ycase4}. We define trajectory $y_4$:
\begin{align}  \label{equation:y4}
y_4(t) = 
	\begin{cases}
		x(t)	 	&\colon 	[t_0,t_3)\\
		\tilde{x}_1(t)&\colon 	[t_3,t_4) \\
		x'(t)		&\colon	[t_4,t_5) \\
		\tilde{x}_2(t)&\colon	[t_5,t_7) \\
		x(t)		&\colon	[t_7,t_f].
	\end{cases}
\end{align}
Note that trajectory $y_4$ satisfies Inequalities~\eqref{equation:higher} and~\eqref{equation:highest}. \qed
\begin{figure}
	\centerline{\includegraphics[clip=true, trim=2in 6.6in 4.2in 3.3in, width=0.5\textwidth]{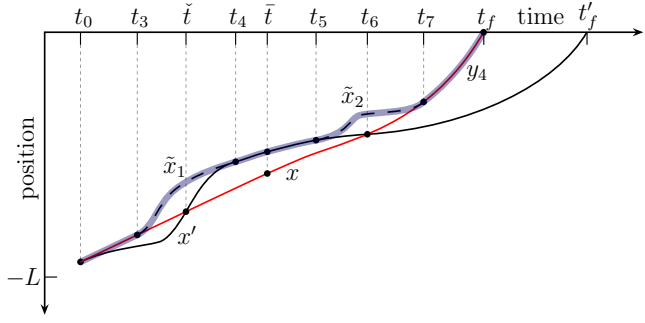}}
	\caption{Notation for Case 4 in the proof of Lemma~\ref{lemma:minimumTrajectory} used in Section D. Trajectory $y_4$ as defined in Equation~\eqref{equation:y4} is depicted in shaded blue. Trajectory $x$ is shown in red, and trajectory $x'$ in black. Trajectory $\tilde{x}_1$ is the dashed black line, and exists from a direct application of Lemma~\ref{lemma:accelerationTrajectoryPatch} at time $\check{t}$.  Trajectory $\tilde{x}_2$, depicted in the dashed black line, also comes from applying Lemma~\ref{lemma:accelerationTrajectoryPatch} at time $t_6$. The assumptions used in Case 4, given in Equations~\eqref{equation:16-7} and~\eqref{equation:case4}, are depicted. The result of Lemma~\ref{lemma:minimumTrajectory} is shown, \ie, $y(t) \geq x(t)$ for all time $t$ and also $y(\bar{t}) > x(\bar{t})$.} \label{figure:ycase4}
\end{figure}
%


\begin{lemma} \label{lemma:accelerationTrajectoryPatch}
Denote $\mathbf{z}_0 := (p_0,v_0)$. 
Denote $\mathbf{x}(t)$ and $\mathbf{y}(t)$ denote $(x(t),\dot{x}(t))$ and $(y(t),\dot{y}(t))$, respectively.
Let $x \in \mathscr{C}(\mathbf{z}_0,t_0,t_f,x_s)$ and $y \in \mathscr{C}(\mathbf{z}_0,t_0,t_f',x_s)$. 
Suppose there exists time $\hat{t}$ satisfying
\begin{align*}
	x(\hat{t}) 		&= y(\hat{t}) \\
	\dot{x}(\hat{t}) 	&> \dot{y}(\hat{t}).
 \end{align*}
Then, there exists trajectory $\tilde{x} \in \mathscr{C}(\mathbf{y}(t_a),t_a,\emptyset,x_s)$ satisfying for all $t \in [t_a,t_b']$:
\begin{align}
	\tilde{x}(t)		&\geq \max\{ x(t),y(t)\} \label{equation:tildep}\\
	\tilde{x}(t_b') 	&= x(t_b') \label{equation:tildep1} \\
	\dot{\tilde{x}}(t_b') &= \dot{x}(t_b') \label{equation:tildep2}
\end{align}
where 
\begin{align*}
	t_a 	&:= \sup \{ t \leq \hat{t} \colon \exists \, \tau, \, \mathbf{y}(t) = \mathbf{x}(\tau) \} \\
	t_b' 	&:= \inf \{  t \geq \hat{t} \colon \exists \, \tau, \, \mathbf{x}(t) = \mathbf{y}(\tau) \}.
\end{align*}
\end{lemma}
\begin{figure}
	\centerline{\includegraphics[clip=true,trim=2.1in 6.5in 4.1in 3.4in, width=0.5\textwidth]{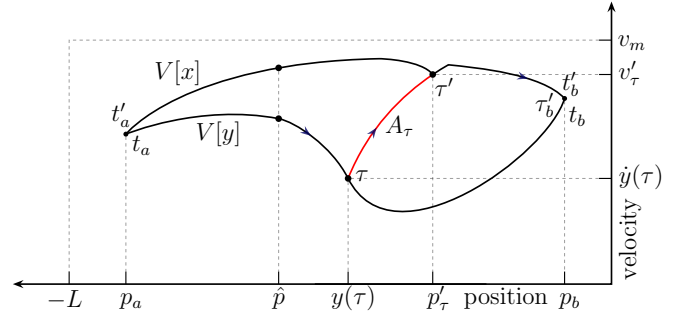}}
	\caption{Notation used in proof of Lemma~\ref{lemma:accelerationTrajectoryPatch}. Velocity path $x$ is the higher velocity path, while velocity path $y$ is the lower one. Trajectories $x$ and $y$ are at position $\hat{p}$ at the same time $\hat{t}$, and velocity of trajectory $x$ is higher than that of $y$, \ie, $\dot{x}(\hat{t}) > \dot{y}(\hat{t}).$ Velocity paths $x$ and $y$ intersect at a position before and after position $\hat{p}$; the closest to $\hat{p}$, we denote $p_a$ and $p_b$, respectively. However, trajectories $x$ and $y$ depart from position $p_a$ and arrive at position $p_b$ at different times. We denote $t_a'$ and $t_a$ as the times that trajectories $x$ and $y$ depart from position $p_a$, respectively. Similarly, we denote $t_b'$ and $t_b$ as the times that trajectories $x$ and $y$ arrive at position $p_b$, respectively. As we move along the lower velocity path $\velpath{y}$, we can choose at any time $\tau \in [t_a,t_b]$ to break away from velocity path $\velpath{y}$ in order to follow a maximum acceleration curve $A_\tau$, depicted in red, until intersection with the higher velocity path $\velpath{x}$. State $(p_\tau',v_\tau')$ is the earliest state at which acceleration curve $A_\tau$ and velocity path $\velpath{x}$ intersect; earliest in the sense of earliest position. (See definition~\eqref{definition:earliest}). Time $\tau'$ is the earliest time that the trajectory following black arrows, which we denote $x[\tau]$, arrives at this state. Trajectory $x[\tau]$ arrives at position $p_b$ at time $\tau_b'$.} \label{figure:hat}
\end{figure}
\begin{figure}
	\centerline{\includegraphics[clip=true,trim=2in 6.6in 4.2in 3.3in, width=0.5\textwidth]{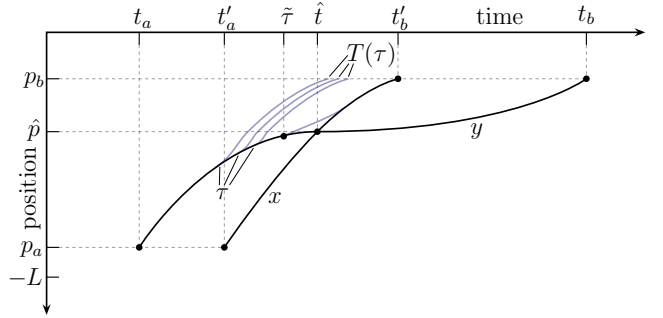}}
	\caption{Notation used in proof of Lemma~\ref{lemma:accelerationTrajectoryPatch}. Trajectories $x$ and $y$ are at position $\hat{p}$ at the same time $\hat{t}$, and velocity of trajectory $x$ is higher than that of $y$, \ie, $\dot{x}(\hat{t}) > \dot{y}(\hat{t}).$ Times $t_a'$ and $t_a$ are the times when trajectories $x$ and $y$ depart from position $p_a$, respectively. Note that $t_a' \geq t_a$. Also, times $t_b'$ and $t_b$ are the times when trajectories $x$ and $y$ arrive at position $p_b$, respectively. Note that $t_b' \leq t_b$. For any time $\tau \in [t_a,t_b]$, we can construct trajectory $x[\tau]$. We denote the time that trajectory $x[\tau]$ arrives at position $p_b$ as $T(\tau)$. Sample trajectories $x[\tau]$ with their associated terminal times $T(\tau)$ are depicted in light blue. Since the mapping $\tau \mapsto T(\tau)$ as defined in Equation~\eqref{equation:terminaltimemapping} is continuous, we can find a time $\tilde{\tau}$ at which to break away from trajectory $y$ as shown in Figure~\ref{figure:hat}, so that the resulting trajectory $x[\tilde{\tau}]$ has terminal time $T(\tilde{\tau})$ equal to $t_b'$.} \label{figure:hat2}
\end{figure}

\begin{proof}
See Figures~\ref{figure:hat} and~\ref{figure:hat2}.
Denote
\begin{align*}
	\hat{p} &:= x(\hat{t}) = y(\hat{t}) \\
	p_a 	&:= y(t_a) \\
	p_b &:= x(t_b').
\end{align*}

Note that for any $p \in (p_a,p_b)$:
\begin{align*}
	\velpath{x}(p) > \velpath{y}(p).
\end{align*}
Note that $t_a$ is the departure time of trajectory $y$ from position $p_a$.
Similarly, let $t_a' := \depart{x}(p_a)$ be the departure time of trajectory $x$ from position $p_a$.
We claim $t_a \leq t_a'$.
See Figure~\ref{figure:hat2}. Essentially, this happens since trajectory $x$ travels faster than $y$ before time $\hat{t}$.
More precisely, note:
\begin{align*}
	t_a' = \depart{x}(p_a) &= \hat{t} - \int_{p_a}^{\hat{p}} \frac{1}{\velpath{x}(p)} dp \\
			&\geq \hat{t} - \int_{p_a}^{\hat{p}} \frac{1}{\velpath{y}(p)} dp
			\geq t_a.
\end{align*}
Also, note that $t_b'$ is the arrival time of trajectory $x$ at position $p_b$.
Similarly, let $t_b := \arrive{y}(p_b)$ be the arrival time of trajectory $y$ at position $p_b$.
Since trajectory $x$ travels faster than $y$ after time $\hat{t}$, we have $t_b' \leq t_b$.
See Figure~\ref{figure:hat2}.
More precisely, note:
\begin{align*}
	t_b' = \arrive{x}(p_b) &= \hat{t} + \int_{\hat{p}}^{p_b} \frac{1}{\velpath{x}(p)} dp \\
		 &\leq \hat{t} + \int_{\hat{p}}^{p_b} \frac{1}{\velpath{y}(p)} dp
		 \leq t_b.
\end{align*}
For any $\tau \in [t_a,t_b]$, we define path $A_\tau$ as follows.
See Figure~\ref{figure:hat2}.
For any $p \in [y(\tau),0]$:
\begin{align*}
	A_\tau(p) := \sqrt{\dot{y}(\tau)^2 + 2a_m(y(\tau) - p)}.
\end{align*}
Define
\begin{align}
	p_\tau' &:= \inf \{ p \colon \exists \, p, \, \velpath{x}(p) = A_\tau(p) \} \label{definition:earliest} \\
	v_\tau' &:=  \velpath{x}(p_\tau) = A_\tau(p_\tau).
\end{align}
Define trajectory $x[\tau] := \mathscr{T}(\mathbf{y}(t_a),u[\tau])$, where control action $u[\tau]$ is defined as follows:
\begin{align*}
u[\tau](t) = 
	\begin{cases}
		\ddot{y}(t) &\colon t \in [t_a,\tau)\\
		a_m 		&\colon t \in [\tau,\tau')\\
		\ddot{x}(t-\zeta[\tau]) &\colon t \in [\tau',\tau_b'],
	\end{cases}
\end{align*}
where
\begin{align*}
	\tau' &:= \tau + (v_\tau' - \dot{y}(\tau))/a_m \\
	\tau_b' &:= \tau' + \int_{p_\tau'}^{p_b} \frac{1}{\velpath{x}(p)} dp \\
	\zeta[\tau] &:= \tau_b' - t_b'.
\end{align*}
See Figure~\ref{figure:hat}.
For any $\tau \in [t_a,t_b]$, let $T(\tau)$ be the arrival time of $x[\tau]$ at position $p_b$, which is simply the terminal time of $x[\tau]$ defined above. 
More precisely, we define:
\begin{align} \label{equation:terminaltimemapping}
	T(\tau) := \tau + \int_{x[\tau](\tau)}^{p_b} \frac{1}{\velpath{x[\tau]}(p)} dp .
\end{align}
We claim that $T$ is continuous.
To see this, consider Figure~\ref{figure:hat2}.
A small change in $\tau$ creates a small change in $T$.
We provide the proof of this claim at the end. 

Assuming $T$ is continuous, we continue to establish the lemma.
Since
\begin{align*}
	T(t_a) &= t_a + t_b '- t_a' \leq t_b' \\
	T(t_b) &= t_b \geq t_b',
\end{align*}	
there exists $\tilde{\tau} \in [t_a,t_b]$ such that $T(\tilde{\tau}) = t_b'$, from which $\zeta[\tilde{\tau}] = 0$. 
We define trajectory $\tilde{x}$ as follows:
\begin{align*}
\tilde{x}(t) = 
	\begin{cases}
		x[\tilde{\tau}](t)	&\colon t \in [t_a,t_b') \\
		x(t)			&\colon t \in [t_b',t_f]. 
	\end{cases}
\end{align*}
By construction of $x[\tilde{\tau}]$:
\begin{align*}
	\tilde{x}(t_b) &= x[\tilde{\tau}](t_b) = x(t_b) \\
	\dot{\tilde{x}}(t_b) &= \dot{x}[\tilde{\tau}](t_b) = \dot{x}(t_b),
\end{align*}
satisfying Inequalities~\eqref{equation:tildep1} and~\eqref{equation:tildep2}.
%

For all $p \in (p_a,p_b)$:
\begin{align*}
	\velpath{\tilde{x}}(p) &\geq \velpath{y}(p) \\
	\wait{\tilde{x}}(p) 	&\leq \wait{y}(p)
\end{align*}
and
\begin{align*}
	\depart{\tilde{x}}(p_a)  = \depart{y}(p_a).
\end{align*}	
Thus, by Proposition~\ref{proposition:boundedTimes}, for all $t \in [t_a,t_b']$:
\begin{align}
	\tilde{x}(t) \geq y(t). \label{equation:inequality1}
\end{align}
Similarly, for all $p \in (p_a,p_b)$:
\begin{align*}
	\velpath{\tilde{x}}(p) \leq \velpath{x}(p) \\
	\wait{\tilde{x}}(p) \geq \wait{x}(p),
\end{align*}	
and
\begin{align*}
	\arrive{\tilde{x}}(p_b) = \arrive{x}(p_b).
\end{align*}
Thus, by Proposition~\ref{proposition:boundedTimesFinal}, for all $t \in [t_a,t_b']$:
\begin{align}
	\tilde{x}(t) \geq x(t). \label{equation:inequality2}
\end{align}
Inequalities~\eqref{equation:inequality1} and~\eqref{equation:inequality2} together establish Inequality~\eqref{equation:tildep}.
To complete the proof, it remains to show that $\tilde{x}$ is safe with $x_s$. However, this follows directly from either Inequality~\eqref{equation:inequality1} or~\eqref{equation:inequality2}. For all $t$:
\begin{align*}
	\tilde{x}(t) \geq y(t) \geq x_s(t) - l.
\end{align*}

We conclude by showing $T$ is continuous.
Choose $\tau_1 \in [t_a,t_b]$. Let $\epsilon > 0$. We show that there exists $\gamma > 0$ such that for any $\tau_2$ satisfying
\begin{align*}
	|\tau_2-\tau_1| \leq \gamma, 
\end{align*}
we have
\begin{align*}
	|T(\tau_2) - T(\tau_1)| \leq \epsilon.
\end{align*}
For shorthand, we denote
\begin{align*}
	V_i &:= \velpath{x[\tau_i]} \\
	p_i &:= y(\tau_i) \\
	v_i &:= V_i(p_i),
\end{align*}
for $i = 1, 2$.
See Figure~\ref{figure:d1}.
We also denote
\begin{align*}
	p_3  &:= \inf \{ p \colon A_{\tau_1}(p) = \velpath{x}(p) \}   \\
	p_4  &:= \inf \{ p \colon A_{\tau_2}(p) = \velpath{x}(p) \}   .
\end{align*}
See Figures~\ref{figure:d1} and~\ref{figure:d2}.

{\it Case 1.}
First, suppose
\begin{align*}
	v_1 		&> 0 \\
	\tau_2 	&< \tau_1.
\end{align*} 
Note that for any $p \in [p_a,p_b]$:
\begin{align*}
	|V_2(p) - V_1(p) | \leq  c(\gamma),
\end{align*}
where
\begin{align}
	c(\gamma) := \sqrt{v_1^2 + 4a_m v_m \gamma} - v_1. \label{equation:cbound1}
\end{align}
\begin{figure}
\centerline{\includegraphics[clip=true, trim=2.1in 6.5in 4.1in 3.4in, width=0.5\textwidth]{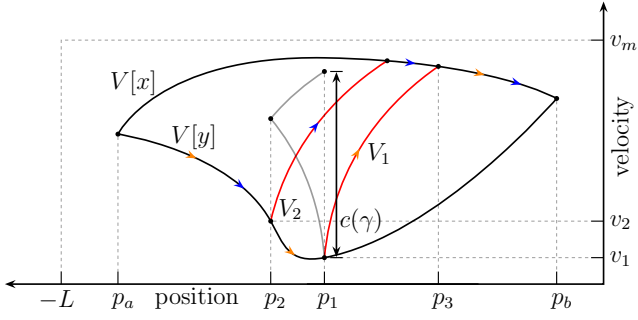}}
\caption{Notation for Case 1 in Proof of Lemma~\ref{lemma:accelerationTrajectoryPatch}. We show that the mapping $\tau \mapsto T(\tau)$, as defined in Equation~\eqref{equation:terminaltimemapping}, is continuous. Velocity paths $\velpath{x}$ and $\velpath{y}$ are the upper and lower black paths, respectively. Velocity paths $V_1$ and $V_2$ are depicted along orange and blue arrows, respectively. The red arcs are maximum acceleration curves. Velocity $v_1$ is the velocity of $V_1$ at time $\tau_1$. Time $\tau_2$ is the time that $V_2$ breaks away from $\velpath{y}$. The conditions of Case 1 are depicted: $v_1 > 0$, and $\tau_2 < \tau_1$.  By construction, $V_1$ and $V_2$ disagree only on the interval $(p_2,p_3)$. Moreover, we bound the absolute value of the difference between $V_1$ and $V_2$, by using maximum acceleration curves over the interval $(p_2,p_1)$, depicted in grey.} \label{figure:d1}
\end{figure}
Note that $V_1$ and $V_2$ only differ on the interval $(p_2,p_3)$. 
Bounding the difference between $V_1$ and $V_2$ amounts to bounding the difference between two square root functions.
Note that the maximum difference between $V_1$ and $V_2$ occurs at position $p_1$.
We denote $c(\gamma)$ as the maximum difference at this position.
To derive the expression for $c(\gamma)$, consider Figure~\ref{figure:d1}. 
The maximum possible distance between positions $p_2$ and $p_1$ is equal to $v_m \gamma$.
The maximum possible velocity (of either $V_1$ or $V_2$) at position $p_2$ is equal to $\sqrt{v_1^2 + 2a_m v_m\gamma}$.
By using full acceleration from position $p_2$ to $p_1$, the maximum possible velocity (of either $V_1$ or $V_2$) at position $p_2$ is equal to $\sqrt{v_1^2 + 4a_m v_m \gamma}$.
The lowest possible velocity (of either $V_1$ or $V_2$) at position $p_1$ is simply $v_1$.
Note that for all $p \in [p_1,p_3]$:
\begin{align*}
	 V_2(p) \geq v_1.
\end{align*}
Also, note that
\begin{align*} 
	\int_{p_2}^{p_1} \frac{1}{V_2(p)} dp &\leq \int_{p_2}^{p_1} \frac{1}{\velpath{y}(p)} dp \\
						&\leq |\tau_2 - \tau_1| \leq \gamma. \numberthis \label{equation:case1p}
\end{align*}
Now, 
\begin{align*}
	|T(\tau_2) &- T(\tau_1)| \\
	&\leq |\tau_2-\tau_1| + \bigg\rvert \int_{p_2}^{p_b}\frac{1}{V_2(p)}dp - \int_{p_1}^{p_b}\frac{1}{V_1(p)}dp \bigg\rvert \\
	&\leq \gamma + \int_{p_2}^{p_1} \frac{1}{V_2(p)} dp + \int_{p_1}^{p_b}\Big\rvert \frac{1}{V_2(p)}-\frac{1}{V_1(p)} \Big\rvert dp   \\
	&\leq  2\gamma + \int_{p_1}^{p_3}\frac{|V_2(p) -V_1(p) |}{V_1(p) V_2(p)} dp \\
	&\leq 2\gamma + c(\gamma)\int_{p_1}^{p_3}\frac{1}{V_1(p) V_2(p)} dp  \\
	&\leq 2\gamma + 2c(\gamma) L / v_1^2 := d(\gamma). 
\end{align*}
Since $d(\gamma)$ is continuous and $\lim_{\gamma \to 0} d(\gamma) = 0$, we can always find $\gamma > 0$ so that $d(\gamma) \leq \epsilon$.

{\it Case 2.}
Now, suppose
\begin{align*}
	v_1 		&> 0 \\
	\tau_2	&> \tau_1.
\end{align*} 
Note that for any $p \in [p_a,p_b]$:
\begin{align*}
	|V_2(p) - V_1(p)| \leq \tilde{c}(\gamma),
\end{align*}	
where
\begin{align*}
	\tilde{c}(\gamma) := \sqrt{v_1^2 + 2a_m v_m \gamma} - \sqrt{v_1^2 - 2 a_m v_m\gamma}.
\end{align*}
\begin{figure}
\centerline{\includegraphics[clip=true, trim=2.1in 6.5in 4.1in 3.4in, width=0.5\textwidth]{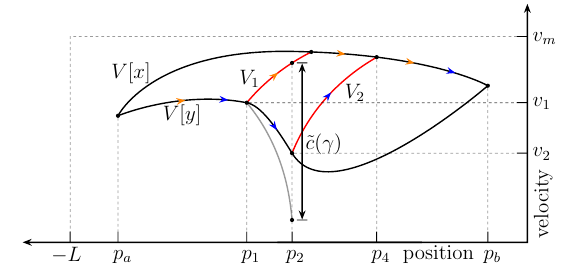}}
\caption{Notation for Case 2 in Proof of Lemma~\ref{lemma:accelerationTrajectoryPatch}. We show that the mapping $\tau \mapsto T(\tau)$, as defined in Equation~\eqref{equation:terminaltimemapping}, is continuous. Velocity paths $\velpath{x}$ and $\velpath{y}$ are the upper and lower black paths, respectively. Velocity paths $V_1$ and $V_2$ are depicted along orange and blue arrows, respectively. The red arcs are maximum acceleration curves. Velocity $v_2$ is the velocity of $V_2$ at time $\tau_2$. Time $\tau_1$ is the time that $V_1$ breaks away from $\velpath{y}$. The conditions of Case 2 are depicted: $v_1 > 0$, and $\tau_2 > \tau_1$.  By construction, $V_1$ and $V_2$ disagree only on the interval $(p_1,p_4)$. Moreover, we bound the absolute value of the difference between $V_1$ and $V_2$ by $\tilde{c}(\gamma)$, by using maximum acceleration curves over the interval $(p_1,p_2)$, depicted in grey.} \label{figure:d2}
\end{figure}
Note that $V_1$ and $V_2$ only differ on the interval $(p_1,p_4)$. Consider Figure~\ref{figure:d2}.
The maximum difference between $V_1$ and $V_2$ now occurs at position $p_2$.
We use two maximum acceleration arcs pivoted at $(p_1,v_1)$.
The lowest achievable velocity (by either $V_1$ or $V_2$) at position $p_2$ is $\sqrt{v_1^2 - 2 a_m v_m\gamma}$, while the highest achievable velocity is $\sqrt{v_1^2+2a_m v_m \gamma}$.

Note that for all $p \in [p_2,p_4]$:
\begin{align*}
	P_2(\xi) \geq v_1/2,
\end{align*}
for sufficiently small $\gamma$, \ie, any $\gamma$ satisfying $\gamma \leq 3v_1 / (2a_m)$ is sufficiently small.
Also, note that
\begin{align*}
	\int_{p_1}^{p_2} \frac{1}{V_1(p)} dp  &\leq \int_{p_1}^{p_2} \frac{1}{\velpath{y}} dp \\
						&\leq |\tau_2 - \tau_1| \leq \gamma. \numberthis \label{equation:case2p}
\end{align*}
Then, 
\begin{align*}
	|T(\tau_2) &- T(\tau_1)| \\
	&\leq |\tau_2-\tau_1| + \bigg\rvert \int_{p_2}^{p_b}\frac{1}{V_2(p)}dp - \int_{p_1}^{p_b}\frac{1}{V_1(p)}dp \bigg\rvert \\
				&\leq \gamma + \int_{p_1}^{p_2} \frac{1}{V_1(p)} dp + \int_{p_2}^{p_b} \Big\rvert \frac{1}{V_2(p)} - \frac{1}{V_1(p)} \Big\rvert dp \\
				&\leq  2\gamma + \int_{p_2}^{p_4}\frac{|V_2-V_1|}{V_1(p)V_2(p)} dp \\
				&\leq 2\gamma + \tilde{c}(\gamma)\int_{p_2}^{p_4}\frac{1}{V_1(p)V_2(p)} dp \\
				&\leq  2\gamma + 2\tilde{c}(\gamma) L / v_1^2 := \tilde{d}(\gamma),
\end{align*}				
where the last inequality holds for sufficiently small $\gamma$.
Since $\tilde{d}(\gamma)$ is continuous and $\lim_{\gamma \to 0} \tilde{d}(\gamma) = 0$, we can always find $\gamma > 0$ sufficiently small so that $\tilde{d}(\gamma) \leq \epsilon$.

{\it Case 3.}
Now, suppose 
\begin{align*}
	v_1 		&=  0 \\
	\tau_2 	&<  \tau_1.
\end{align*}
Note that $V_1$ and $V_2$ differ only on the interval $(p_2,p_3)$.
On this interval we bound the difference in velocities by using Equation~\eqref{equation:cbound1} and setting $v_1$ equal to 0.
Thus, for all $p \in [p_2,p_3]$:
\begin{align*}
	|V_2(p) - V_1(p)| 	&\leq C \sqrt{\gamma} \\
	V_1(p) 			&\leq V_2(p),
\end{align*}
where $C := \sqrt{4a_m v_m}$.
Also, note that Inequality~\eqref{equation:case1p} still holds when $v_1$ is equal to 0.
For any $\gamma > 0$ the following holds:
\begin{align*}
	&|T(\tau_2) - T(\tau_1)| \\
	&\leq |\tau_2-\tau_1| + \bigg\rvert \int_{p_2}^{p_b}\frac{1}{V_2(p)}dp - \int_{p_1}^{p_b}\frac{1}{V_1(p)}dp \bigg\rvert \\
	&\leq \gamma + \int_{p_2}^{p_1} \frac{1}{V_2(p)} dp + \int_{p_1}^{p_b} \Big\rvert \frac{1}{V_2(p)} - \frac{1}{V_1(p)} \Big\rvert dp  \\
	&\leq 2\gamma + \int_{p_1}^{p_3} \Big\rvert \frac{1}{V_2(p)} - \frac{1}{V_1(p)} \Big\rvert dp \\
	&\leq 2\gamma + \int_{p_1}^{p_1+\gamma} \Big\rvert \frac{1}{V_2(p)} - \frac{1}{V_1(p)} \Big\rvert dp  + \int_{p_1+\gamma}^{p_3} \frac{|V_2 - V_1|}{V_1 V_2} dp \\
	&\leq 2\gamma + \int_{p_1}^{p_1+\gamma} \Big( \frac{1}{V_2} + \frac{1}{V_1} \Big) dp + C\sqrt{\gamma} \int_{p_1+\gamma}^{p_3} \frac{1}{V_1 V_2} dp \\
	&\leq 2\gamma + 2\int_{p_1}^{p_1+\gamma} \frac{1}{V_1(p)} dp +C\sqrt{\gamma} \int_{p_1+\gamma}^{p_3} \frac{1}{V_1(p)^2} dp  \\
	&=  2\gamma + 2\sqrt{2\gamma/a_m} + C\sqrt{\gamma} \int_{p_1+\gamma}^{p_3} \frac{1}{2a_m(p-p_1)} dp \\
	&= 2\gamma + 2\sqrt{2\gamma/a_m}+ \frac{C\sqrt{\gamma}}{2a_m} \log\Big(\frac{p_3-p_1}{\gamma}\Big) := q(\gamma).
\end{align*}
Note that $q(\gamma)$ is continuous and $\lim_{\gamma \to 0} q(\gamma) = 0$.
Thus, we can always find $\gamma > 0$ such that $q(\gamma) \leq \epsilon$.

{\it Case 4.}
Now, suppose 
\begin{align*}
	v_1 		&=  0 \\
	\tau_2 	&>  \tau_1.
\end{align*}
Note that $V_1$ and $V_2$ differ only on interval $(p_1,p_4)$.
We bound the difference in velocities as follows.
The maximum difference in velocities occurss at position $p_2$. The maximum distance between $p_1$ and $p_2$ is equal to $v_m \gamma$.
The maximum possible velocity (of either $V_1$ or $V_2$) at position $p_2$ is, thus, $\sqrt{2a_m v_m \gamma}$, while the lowest possible velocity is 0. 
Hence, for all $p \in [p_1,p_4]$:
\begin{align*}
	|V_2(p) - V_1(p)| \leq \tilde{C} \sqrt{\gamma}, 
\end{align*}
where $\tilde{C} := \sqrt{2a_m v_m}$. Note also that for all $p \in [p_2,p_4]$:
\begin{align*}
	V_1(p) \geq V_2(p) \geq \frac{1}{\sqrt{2a_m(p-p_2)}}.
\end{align*}
Also, note that Inequality~\eqref{equation:case2p} still holds when $v_1$ is equal to 0.
Then, 
\begin{align*}
	&|T(\tau_2) - T(\tau_1)| \\
	&\leq |\tau_2-\tau_1| + \bigg\rvert \int_{p_2}^{p_b}\frac{1}{V_2(p)}dp - \int_{p_1}^{p_b}\frac{1}{V_1(p)}dp \bigg\rvert \\
	&\leq \gamma + \int_{p_1}^{p_2} \frac{1}{V_1(p)} dp + \int_{p_2}^{p_b} \Big\rvert \frac{1}{V_2(p)} - \frac{1}{V_1(p)} \Big\rvert dp  \\
	&\leq 2\gamma + \int_{p_2}^{p_4} \Big\rvert \frac{1}{V_2(p)} - \frac{1}{V_1(p)} \Big\rvert dp \\
	&\leq 2\gamma + \int_{p_2}^{p_2+\gamma} \Big\rvert \frac{1}{V_2(p)} - \frac{1}{V_1(p)} \Big\rvert dp  + \int_{p_2+\gamma}^{p_4} \frac{|V_2 - V_1|}{V_1 V_2} dp \\
	&\leq 2\gamma + \int_{p_2}^{p_2+\gamma} \Big( \frac{1}{V_2} + \frac{1}{V_1} \Big) dp + \tilde{C}\sqrt{\gamma} \int_{p_2+\gamma}^{p_4} \frac{1}{V_1 V_2} dp \\
	&\leq 2\gamma + 2\int_{p_2}^{p_2+\gamma} \frac{1}{V_2(p)} dp +\tilde{C}\sqrt{\gamma} \int_{p_2+\gamma}^{p_4} \frac{1}{V_2(p)^2} dp  \\
	&\leq  2\gamma + 2\sqrt{2\gamma/a_m} + \tilde{C}\sqrt{\gamma} \int_{p_2+\gamma}^{p_4} \frac{1}{2a_m(p-p_2)} dp \\
	&\leq 2\gamma + 2\sqrt{2\gamma/a_m}+ \frac{\tilde{C}\sqrt{\gamma}}{2a_m} \log\Big(\frac{p_4-p_2}{\gamma}\Big) := \tilde{q}(\gamma).
\end{align*}
Note that $\tilde{q}(\gamma)$ is continuous and $\lim_{\gamma\to0}\tilde{q}(\gamma) = 0$.
Thus, we can always find $\gamma > 0$ such that $\tilde{q}(\gamma) \leq \epsilon$.

Hence, have shown that $T$ is continuous. \end{proof}

\end{document}